\documentclass[onecolumn,journal,12pt,dvipdfmx]{IEEEtran}

\usepackage[dvipdfmx]{graphicx}

\usepackage{amssymb,amsfonts,amstext}
\usepackage{amsmath}	
\usepackage{latexsym}
\usepackage{color}

\newcommand{\bZ}{\mathbb{Z}}

\newcommand{\Hi}{\mathcal{H}}

\newcommand{\Tr}{\mathrm{Tr}}

\newcommand{\ket}[1]{| #1 \rangle}
\newcommand{\bra}[1]{\langle #1 |}
\newcommand{\braket}[2]{\langle #1 \vert #2 \rangle}
\newcommand{\be}{\begin{equation}}
\newcommand{\ee}{\end{equation}}
\newtheorem{definition}{Definition}
\newtheorem{proposition}[definition]{Proposition}
\newtheorem{theorem}[definition]{Theorem}

\newtheorem{lemma}[definition]{Lemma}
\def\QED{\mbox{\rule[0pt]{1.5ex}{1.5ex}}}
\def\endproof{\hspace*{\fill}~\QED\par\endtrivlist\unskip}

 \newenvironment{proofof}[1]{\vspace*{5mm} \par \noindent
         \quad{\it Proof of #1:\hspace{2mm}}}{\endproof}
 \newenvironment{proofsof}[1]{\vspace*{5mm} \par \noindent
         \quad{\it Proofs of #1:\hspace{2mm}}}{\endproof}

\def\PF#1{\noindent{\sf #1}:\quad}

\def\Label#1{\label{#1}\ [\ \text{\protect{#1}}\ ]\ }
\def\Label{\label}

\begin{document}
\title{Tight asymptotic bounds on local hypothesis testing between a pure bipartite state and the white noise state\footnote{This paper was
presented in part at 
Workshop on Quantum Metrology, Interaction, and Causal Structure, Beijing, China, December, 2014,
The 17th workshop on Quantum Information Processing (QIP 2015), Sydney, NSW, Australia, January, 2015,
and
2015 IEEE International Symposium on Information Theory, 
Hong-Kong, June 2015.}}

\author{Masahito Hayashi \IEEEmembership{Senior Member, IEEE}, Masaki Owari
\thanks{M. Hayashi is with Graduate School of Mathematics, Nagoya University, Nagoya, 464-8602, Japan and Centre for Quantum Technologies, National University of Singapore, Singapore. e-mail: masahito@math.nagoya-u.ac.jp}% <-this % stops a space
\thanks{M. Owari was with NTT Communication Science Laboratories, NTT Corporation
3-1, Morinosato Wakamiya Atsugi-Shi, Kanagawa, 243-0198, Japan. 
Now, he is with Faculty of Informatics, Shizuoka University, Hamamatsu, Shizuoka, 432-8011, Japan 
e-mail: masakiowari@inf.shizuoka.ac.jp}
%owari.masaki@lab.ntt.co.jp}% <-this % stops a space
}

\maketitle

\begin{abstract}
We consider asymptotic hypothesis testing (or state 
discrimination with asymmetric treatment of errors) between an
 arbitrary 
fixed
bipartite pure state $\ket{\Psi}$ and the white noise state (the completely mixed state) under one-way LOCC (local operations and classical
communications), two-way LOCC, and separable POVMs.
As a result, we derive
the Hoeffding bounds under
two-way LOCC POVMs and separable POVMs. 
Further,
we derive a Stein's lemma type of
optimal error exponents under one-way LOCC, two-way LOCC, and separable POVMs
up to the third order,
which clarifies the difference between
one-way and two-way LOCC POVM. 
Our results clarify the relationship between the entanglement of Renyi 
 entropy and the hypothesis testing under LOCC, 
since the entanglement of Renyi entropy appears in the formula of both the Hoeffding bounds and the 
 Stein's lemma type of error exponents.
Our study gives a very rare example in which
the optimal performance under the infinite-round two-way LOCC 
is also equal to that under separable operations
and can be attained with two-round communication,
but not with the one-way LOCC.
\end{abstract}

\section{Introduction}
%$ \square $ 
When a quantum system consists of two distinct parties,
Alice and Bob, it is 
natural to restrict their operations to local operation and classical communication (LOCC) \cite{HHHH09}
because it is not so easy to realize a quantum operation across both of the distant parties. 
LOCC operations can be classified by the direction
of classical communication. When the direction of classical communication is
restricted to only one direction, the LOCC operation is called a one-way LOCC.
Otherwise, it is called a two-way LOCC.
Such  constraint for our measurement is called a locality restriction.
In this paper, we focus on the effect for distinguishing quantum states.
Such a state discrimination problem 
has been studied very actively by many researchers
\cite{PW91,BDFMRSSW99,WSHV00,GV01,VSPM01,GKRSS01,TDL01,Wa05,HMMOV06,HMT06,OH06,KTYI07,C07,OH08,IHHH08,MW08,DFXY09,H09,JRZZG09a,CVMB10,JRWZGF10,Ba10,N10,KKB11,LW12,CH13a,CLMO13,FLM13,BHLP13}.

In this paper, we concentrate on the detection of a given entangled state from the completely mixed state, 
which is often called the white noise state because it has no biased noise.
Since this problem deals with two states as candidates for the true state in an asymmetric way,
it is usually referred to as the binary simple hypothesis testing.
Since we impose the locality restriction, we call it the local hypothesis testing.
Since, as was pointed out from a Shannon theoretical viewpoint \cite{VH, 
HK02, Nag, Hay-Nag, H06, Nag07, H07, WR, PPV, Pol, RH, Tomamichel}, 
hypothesis testing is related to so many information theoretic 
problems, 
quantum hypothesis testing 
with the asymptotic and asymmetric setting has attracted much attention in quantum information theory 
\cite{BHLP13,Tomamichel,ANSV08,Nagaoka, Li,ON00, HP91, HT, OM, CMW,  SB, Notzel, H-opt}.
In order to discuss the relation between the locality constraint and these information theoretic problems, 
it is natural to deeply investigate quantum hypothesis testing with locality restriction.

One might consider that hypothesis testing with the white noise state is too specialized.
However, as known in classical information theory, 
this type of hypothesis testing is directly related to 
data compression \cite{HK02,Nag07}, uniform random generation \cite{HK02}, 
channel coding with additive noise \cite{VH}, and resolvability of distribution \cite{RH}.
Thus, this problem can be regarded as the first step for extending these topics to the case with the locality constraint.
Indeed, based on a similar motivation, a recent paper \cite{CMW} 
treats the hypothesis testing of quantum channel with a special case as 
a quantum extension of a special case of the paper \cite{H-ch}. 
Further, hypothesis testing even with the white noise state is highly non-trivial when we impose any locality restriction,
although it is trivial without one.
Hence, this problem represents the difficulty caused by the locality restriction in the simplest way,
and it can be considered as one of the most important types of local hypothesis testing.
Therefore, to characterize the accessible information under locality condition,
we tackle the local hypothesis testing with the white noise state in this paper.

On the other hand, since this problem can be described in terms of the entangled pure state to be detected, 
this problem is closely related to the amount of entanglement of the entangled pure state.
Hence, it has a great significance as a study of entanglement.
In fact, several entanglement measures have been proposed even for pure entangled states.
One is the entanglement of entropy \cite{BBPS96}, and its relation with hypothesis testing with the white noise state
has been clarified \cite{OH11}.
As other measures, the geometric measure of entanglement \cite{WG03}
and the robustness of entanglement \cite{VT99} are known.
However, their relations with this problem have only been  
partially resolved \cite{OH11}.
To discuss the relation between entanglement measures and hypothesis testing,
we employ 
the entanglement of R\'{e}nyi entropy \cite{Vidal}, i.e.,
the R\'{e}nyi entropy of the reduced density matrix of a pure entangled state,
which contains the entanglement of entropy, the geometric measure of entanglement, and the logarithmic robustness of entanglement as special cases.
Since R\'{e}nyi entropy is also closely related to 
the asymptotic performance of quantum information protocols,
we may predict that the entanglement of Renyi entropy is also closely 
related to the asymptotic performance of quantum information processing 
under the locality condition. 
%it is natural to adopt the entanglement of R\'{e}nyi entropy as a measure of entanglement.
In this paper, we show that this prediction is correct. That is, we clarify the relation between 
our hypothesis testing problem and the entanglement of R\'{e}nyi entropy.

Before discussing the history of the local hypothesis testing,
we focus on the quantum hypothesis testing without a locality condition, 
in which a general asymptotic theory can be established even for the 
quantum case where multiple copies of unknown states are available.
Firstly,
Hiai et al. \cite{HP91} and Ogawa et al. \cite{ON00} derived the quantum 
version of Stein's bound \cite{C52}, i.e.,  
the optimal exponent of the type-2 error under the constant constraint for the type-1 error.
Audenaert et al. \cite{ACMBMAV07} and Nussbaum et al. \cite{Nussbaum} derived the quantum version of the
Chernoff bound \cite{C52}, i.e., 
the optimal exponent of the sum of type-1 and type-2 errors.
Other papers \cite{H07,Nagaoka} derived the quantum version of the 
Hoeffding bound \cite{Hoeffding,B74,CL71}, 
%i.e.,  
which is the optimal exponent of the type-2 error under 
the exponential constraint for the type-1 error and 
can be considered to be a  generalization of the Chernoff bound. 
However, when we impose the one-way or two-way LOCC constraint on our measurement,
these problems become very difficult, and they have not been solved completely.
In particular, it is quite difficult to solve these problems 
%in a general setting.
for an arbitrary fixed pair of quantum states.
In the following, we mainly address the Hoeffding bound and will hardly 
mention the Chernoff bound.
This treatment does not lose generality
because our results for the Hoeffding bound include the results for the Chernoff bound as special cases. 

Before proceeding to the detailed discussion of the local hypothesis testing between a pure entangled state $\ket{\Psi}$ 
and the white noise state,
we prepare a detailed classification of two-way LOCC operation.
whereas a one-way LOCC operation requires only one-round classical communication,
a two-way LOCC operation requires multiple-round classical communication.
In this case,
a two-way LOCC protocol with $k$-round classical communication has $k+1$ steps.
For example, in the case of  two-round classical communication, 
the total protocol is given as follows when the initial operation is done by Alice:
Alice performs her operation with her measurement and sends her outcome to Bob.
Bob receives Alice's outcome, performs his operation with his measurement, and sends his outcome to Alice.
Alice then receives Bob's outcome and performs her measurement.
Therefore, we focus on the difference among these locality restrictions.
under the local hypothesis testing between a pure entangled state and the white noise state.

In the non-asymptotic setting,
our previous paper \cite{OH08} addressed the problem under the constraint that 
%the given pure entangled state 
$\ket{\Psi}$ is detected with probability $1$.
Our more recent paper \cite{OH10} addressed it in a more general setting.
In particular, that paper \cite{OH10} proposed concrete two-round classical communication
two-way LOCC protocols that are not reduced to one-way LOCC.
Then, we  extended the problem to the case when the entangled state is 
given as the $n$-copy state of a certain entangled state \cite{OH11}. 
As asymptotic results, we showed that 
there is no difference between one-way and two-way LOCC for Stein's bound, 
i.e., the optimal exponent of the type-2 error under the constant constraint for the type-1 error.
To make an upper bound of the optimal performance of the two-way LOCC case,
our papers \cite{OH08,OH11,OH10} also considered the performance for 
separable operations, which can be easily treated because of their 
mathematically simple forms. 
The class of separable operations includes LOCC, but there exist  
separable operations that are not LOCC \cite{BDFMRSSW99}.
Unfortunately, our previous paper \cite{OH11} could not derive the Hoeffding bound for 
two-way LOCC, i.e., 
the optimal exponent of the type-2 error under the exponential constraint for the type-1 error,
while it derived it for one-way LOCC.
Further, even under the constant constraint for the type-1 error,
the paper did not consider the higher order of the decreasing rate of the type-2 error. 
Indeed, in information theory, Strassen \cite{Strassen} 
derived the decreasing rate of the type-2 error 
up to the third-order $\log n$ under the same constraint
in the classical setting
when $n$ is the number of available copies.
Tomamichel et al. \cite{Tomamichel} and Li \cite{Li} extended this result up to the second-order $\sqrt{n}$.

In this paper, we derive the Hoeffding bound for two-way LOCC and the
optimal decreasing rate of the type-2 error under the constant constraint for the
type-1 error up to the third-order $\log n$ for one-way and two-way LOCC. We also
derive them for separable measurements. 
The obtained results are summarized as follows.
%%%

\begin{description}
\item[(1)] 
There is a difference in the Hoeffding bound between the one-way and two-way LOCC constraints 
unless the entangled state $\ket{\Psi}$ is maximally entangled.

\item[(2)] 
There is no difference in the Hoeffding bound between two-way LOCC and separable constraints. 

\item[(3)]
The optimal decreasing rate of the type-2 error under the constant constraint 
for the type-1 error 
has no difference between the one-way and two-way LOCC constraints 
up to the second-order $\sqrt{n}$.

\item[(4)]
The optimal decreasing rate of the type-2 error under the constant constraint 
for the type-1 error 
is different between the one-way and two-way LOCC constraints 
in the third-order $\log {n}$
unless the entangled state $\ket{\Psi}$ is maximally entangled.

\item[(5)]
The optimal decreasing rate of the type-2 error under the constant constraint 
for the type-1 error 
is not different between the two-way LOCC and separable constraints 
up to the third-order $\log{n}$.

\item[(6)]
The three-step two-way LOCC protocol proposed in \cite{OH10}
can achieve the Hoeffding bound for two-way LOCC.

\item[(7)] 
The three-step two-way LOCC protocol proposed in \cite{OH10}
can achieve the optimal decreasing rate of the type-2 error under the 
	   constant constraint for the type-1 error up to the third-order 
	   $\log n$ for two-way LOCC. 

\item[(8)] 
The entanglement of Renyi entropy appears in the 
	   formulas of
	   the Hoeffding bounds and the optimal decreasing rate of the 
	   type-2 error under the constant constraint for the type-1 
	   error for all the one-way LOCC, the two-way LOCC, and 
	   separable constraints.
\end{description}

Finally, we discuss our result from the mathematical point of view.
The difficulty of the above results can be classified into two parts.
One is the asymptotic evaluation of optimal performance of separable operations.
The other is the asymptotic evaluation of optimal performance of
the three-step two-way LOCC protocol proposed in \cite{OH10}.
To evaluate the exponential decreasing rates in the latter case,
we employ the type method \cite{CKbook},
the saddle point approximation \cite{Dembo98,Moulin13}.

The evaluation of the former case, we need complicated discussions.
Firstly, as mentioned in \cite{OH10}, 
we convert our local hypothesis testing with separable operations
into a specific composite hypothesis testing.
Then, we evaluate the exponential decreasing rates of error probabilities in the converted specific composite hypothesis testing.
Usually, to evaluate the exponential decreasing rate, 
we employ large deviation theory, e.g., Cram\'{e}r Theorem.
However, for our analysis, we need more detailed analysis.
Hence, we employ the strong large deviation initiated by Bhadur-Rao \cite{BR},
which enables us to analyze the tail probability up to the constant order of exponentially small probability. 
(See Proposition \ref{11-4-4} in Appendix \ref{A1}.)
Indeed, although Bhadur-Rao \cite{BR} obtained such detailed evaluation for the tail probability in 1960,
they were rarely applied to information theoretical topics.
That is, our analysis is a good application of the strong large deviation.
Based on this analysis for the specific composite hypothesis testing,
we derive our analysis for the former case. 

Indeed, after the first submission of this paper, 
the recent paper \cite{Li2} discussed the composite hypothesis testing with the large deviation formalism.
Our converted composite hypothesis testing is different from the discussion in \cite{Li2} in the following point.
The paper \cite{Li2} fixes the number of possible states in the hypothesis, which does not increase dependently of the number $n$ of tensor product.
However, in our composite hypothesis testing, the number of possible states in the hypothesis increases double exponentially with respect to 
the number $n$ of tensor product.
Due to the double exponential increase, 
the method in the paper \cite{Li2} cannot be applied to our problem, which requires a special treatment as explained the above.  

This paper is organized as follows: 
In Section \ref{sec preliminary}, 
we summarize the known results for simple hypothesis testing
and
explain the main results by preparing the mathematical descriptions of our hypothesis
testing problem.
Then, we derive the analytical expressions of the optimal error exponents under one-way LOCC POVMs 
in Section \ref{sec one-way}. 
Next, in Section \ref{sec separable},
we derive 
the analytical expressions of the optimal error exponents under separable LOCC POVMs.
For this derivation, we discuss a specific composite hypothesis testing
by using the strong large deviation \cite{BR}.
In Section \ref{sec two-way}, we analyze a special class of two-round classical communication LOCC (thus, two-way LOCC) for this local hypothesis testing problem
by using 
the type method \cite{CKbook}
and the saddle point approximation \cite{Dembo98,Moulin13}.
Finally, we summarize the results of our paper in Section
\ref{sec summary}.
Our notation is the same as in our previous paper \cite{OH11}.
It therefore might be helpful for readers to refer to the list of notations given in the appendix of \cite{OH11}.
In Appendix \ref{A2}, we summarize the formulation and results of \cite{OH10}
needed in Subsubsection \ref{subsub1}.
In Appendix \ref{A1}, we summarize the basic knowledge for
the strong large deviation \cite{BR}.

\section{Preliminary and main results} \Label{sec preliminary}
\subsection{Preliminary I: General quantum hypothesis testing}
This paper mainly treats hypothesis testing in a bipartite quantum system and its
$n$-copies extension. 
For this purpose, we firstly discuss hypothesis testing in a general quantum system ${\cal H}$ and 
its $n$-copies extension. 
%A space of all operators on a Hilbert space $\Hi$ is written as $\B \left(\Hi\right)$. 
In quantum hypothesis testing, we consider two hypotheses, the null hypothesis and the alternative hypothesis.
When a hypothesis consists of one element, it is called simple.
Otherwise, it is called composite.
This paper mainly addresses simple hypotheses, but it discusses a composite hypothesis partially.
Here, we assume that the null hypothesis is a state $\rho$ and the alternative hypothesis is state $\sigma$.
In the $n$-copies setting,
the quantum system is given by ${\cal H}^{\otimes n}$.
Then, the null and alternative hypotheses are the states $\rho^{\otimes n}$ and $\sigma^{\otimes n}$.
Our decision is given by
a two-valued POVM consisting of two POVM elements $T_n$ and $I^n-T_n$, where 
$I^n$ is the identity operator on $\Hi^{\otimes n}$ and $T_n$ is an positive-semi definite operator on $\Hi^{\otimes n}$.
When the measurement outcome corresponds to $T_n$, 
we judge an unknown state as $\sigma^{\otimes n}$,
and when the measurement outcome is $I^n-T_n$, we judge it as $\rho^{\otimes n}$. 

Thus, type-1 error is written as 
\begin{equation}\Label{eq def type 1 error}
\alpha _n(T_n)\stackrel{\rm def}{=}\Tr \rho^{\otimes n}T_n,
\end{equation}
and type-2 error is written as 
\begin{equation}\Label{eq def type 2 error}
\beta _n (T_n) \stackrel{\rm def}{=} \Tr \sigma^{\otimes n}  \left ( I^n-T_n \right ).
\end{equation}
The optimal type-2 error under the condition 
that the type-1 error is no more than a constant $\alpha\ge 0$ is written as
\begin{align}\Label{eq def beta n C alpha rho sigma1}
\beta_{n}(\alpha|\rho \| \sigma) \stackrel{\rm def}{=}  
 \min _{T_n} \left \{
 \beta_n (T_n) \  | \ \alpha_n (T_n) \le \alpha, I^n \ge T_n \ge 0 \right \}.
\end{align}

Now, we give the asymptotic properties of $\beta_{n}(\alpha|\rho\| \sigma)$.
For this purpose, we introduce 
the cumulative distribution function (CDF) of the standard normal distribution
$\Phi(x)\stackrel{\rm def}{=}\int_{-\infty}^{x} \frac{e^{-y^2/2}}{\sqrt{2\pi}} dy$,
the quantum relative entropy $D(\rho\|\sigma)\stackrel{\rm def}{=}\Tr \rho (\log \rho - \log \sigma)$,
and the quantities $V(\rho\|\sigma)\stackrel{\rm def}{=}\Tr \rho (\log \rho - \log \sigma- D(\rho\|\sigma))^2$,
and $\psi(s|\rho\|\sigma)\stackrel{\rm def}{=}- \log \Tr  \rho^{1-s} \sigma^s$.
Then, when $V(\rho\|\sigma) >0$, we have the asymptotic expansions \cite{Hoeffding,CL71,B74,Strassen}
\begin{align}
\log \beta_{n}(\epsilon|\rho\| \sigma) &= -n D(\rho\|\sigma)- \sqrt{n} \sqrt{V(\rho\|\sigma)}\Phi^{-1}(\epsilon) +O(\log n) \Label{2-4-2}\\
\log \beta_{n}(e^{-nr}|\rho\| \sigma) &= -n \sup_{0 \le s <1}\frac{\psi(s|\sigma\|\rho)-s r}{1-s}+ o(n).\Label{2-4-3}
\end{align}
Expansions (\ref{2-4-2}) and (\ref{2-4-3})
are called the Stein-Strassen and the Hoeffding expansions, respectively.

When $\rho$ and $\sigma$ commute each other,
we have the more detailed expansion
\begin{align}
\log \beta_{n}(\epsilon|\rho\| \sigma) = -n D(\rho\|\sigma)- \sqrt{n} \sqrt{V(\rho\|\sigma)}\Phi^{-1}(\epsilon) -\frac{1}{2} \log n + O (1).
\Label{2-4-4}
\end{align}

\subsection{Preliminary II: Known results of local hypothesis testing}
Now, we proceed to the hypothesis testing on a bipartite quantum system 
and its $n$-copies extension, which is the main topic of this paper.
A single copy of a bipartite Hilbert space is written as 
$\Hi_{AB} \stackrel{\rm def}{=} \Hi _A \otimes \Hi_B$, 
and its local dimensions are written as  $d_A \stackrel{\rm def}{=} \dim \Hi_A$ 
and $d_B \stackrel{\rm def}{=} \dim \Hi_B$. 
We use notations like 
$I_A$, $I_B$, $I_{AB}$, $I_A^n$, $I_B^n$, and $I_{AB}^n$ for identity
operations on $\Hi_A$, $\Hi_B$, $\Hi _{AB}$, $\Hi_A^{\otimes n}$,
$\Hi_B^{\otimes n}$, and $\Hi_{AB}^{\otimes n}$, respectively. 
When it is  easy to identify the domain of an identity operator, we 
abbreviate them to $I$ hereafter.  

In this paper, 
we define $d$ as
\begin{equation}\label{eq def d}
d \stackrel{\rm def}{=} \min (d_A,d_B),
\end{equation}
and
consider asymptotic hypothesis testing between
$n$-copies of an arbitrary known pure-bipartite
state $\ket{\Psi}$ with the Schmidt decomposition as 
\begin{equation}\Label{eq schmidt decomposisiton psi}
\ket{\Psi}\stackrel{\rm def}{=}\sum_{i=1}^{d}
\sqrt{\lambda_i}\ket{i}\otimes\ket{i}, 
\end{equation}
and $n$-copies of the white noise state (the completely mixed state) 
\begin{equation}\label{eq def rho mix}
\rho _{mix}\stackrel{\rm def}{=}\frac{I_{AB}}{d_Ad_B}
\end{equation} 
under the various restrictions on available POVMs: global POVMs, separable POVMs, one-way LOCC POVMs, and two-way LOCC POVMs \cite{HHHH09,VP07}. 
We choose the white noise state (the completely mixed state)
$\rho_{mix}^{\otimes n}$ as
a null hypothesis and the state $\ket{\Psi}^{\otimes n}$ as an
alternative hypothesis. 
%In the following part of this section, we address various types of formulation with a general simple null hypothesis $\rho^{\otimes n}$  and a general simple alternative hypothesis $\sigma^{\otimes n}$.
%Thus, $\rho=\rho_{mix}$ and $\sigma=\Psi\stackrel{\rm def}{=}\ket{\Psi}\bra{\Psi}$ in our local hypothesis testing problem.

As variants of $\beta_{n}(\alpha|\rho \| \sigma)$,
the optimal type-2 error under the condition 
that the type-1 error is no more than a constant $\alpha\ge 0$ is written as
\begin{align}\label{eq def beta n C alpha rho sigma}
\beta_{n,C}(\alpha|\rho \| \sigma) \stackrel{\rm def}{=}  
 \min _{T_n} \left \{
 \beta_n (T_n) \  | \ \alpha_n (T_n) \le \alpha, \{ T_n, I^n-T_n \} \in C \right \},
\end{align}
where $C$ is either $\rightarrow$, $\leftrightarrow$, $Sep$, and $g$
corresponding to classes of one-way LOCC, two-way LOCC, separable
and global POVMs, respectively.
Here, we note that although $\rightarrow$, $Sep$, and $g$ are compact sets, 
$\leftrightarrow$ is not compact by its original definition \cite{CLMOW13}. 
Further, we denote the class of two-way LOCCs with $k$-round classical communication by $\leftrightarrow, k$.
In this notation, $\leftrightarrow, 1$ is equivalent to $\rightarrow$.
In this case, the opposite one way LOCC $\leftarrow$ can be obtained by 
swapping systems ${\cal H}_A$ and ${\cal H}_B$. So, we do not discuss 
the opposite one way LOCC $\leftarrow$.

Hence, in this paper, the class $\leftrightarrow$ is defined as a closure of the set of all two-way LOCC POVMs, 
which involves infinite-step LOCC protocols as well \cite{BDFMRSSW99,KKB11,OBNM08,Ch11,CCL12}.
This definition of the class $\leftrightarrow$ justifies the use of 
$\min$ in Eq.(\ref{eq def beta n C alpha rho sigma}) for
$C=\leftrightarrow$.  
In the global POVMs $g$, 
since 
\begin{align}
\log \beta_{n,g}(\epsilon|\Psi\|\rho_{mix}) &= -n \log d_A d_B + \log (1-\epsilon) ,
\end{align}
%applying (\ref{2-4-1}), (\ref{2-4-4}), and (\ref{2-4-3}),
as is shown in \cite{OH11},
we have 
\begin{align}
\beta_{n,g}(\epsilon|\rho_{mix}\|\Psi) &=0 \\
\beta_{n,g}(e^{-nr}|\rho_{mix}\|\Psi) &= 0 \hbox{ with } r \in [0, \log d_A d_B]\\
\beta_{n,g}(e^{-nr}|\rho_{mix}\|\Psi) &= 1 \hbox{ with } r \in (\log d_A d_B, +\infty),
\end{align}
and the following expansions
\begin{align}
\log \beta_{n,g}(e^{-nr}|\Psi\|\rho_{mix}) 
&= -n \log d_A d_B +\log (1 -e^{-nr}) \nonumber \\
&= -n \log d_A d_B -e^{-nr}+o(e^{-nr}) .
\end{align}

To discuss the remaining cases,
we introduce the R\'{e}nyi entropy $H_{1-s}(\Psi)$ of the reduced density of the entangled state $\ket{\Psi}$
and its derivative as follows.
\begin{align}
H_{1-s}(\Psi) \stackrel{\rm def}{=} \frac{\log \sum_{i}\lambda_i^{1-s}}{s} ,
\quad
H_{\alpha}'(\Psi)\stackrel{\rm def}{=}\frac{d}{d \alpha} H_{\alpha}(\Psi).
\end{align}
Here, $H_{1}(\Psi)$ is defined as the limit $\lim_{s\to 0}H_{1-s}(\Psi)$.
By the R\'{e}nyi entropy $H_{1-s}(\Psi)$, 
the entropy of the entanglement $E\left(\ket{\Psi}\right)$,
%defined by Eq.(\ref{eq def entropy of entanglement}).
the Schmidt rank $R_{S}(\ket{\Psi})$ \cite{VP07,HHHH09}, %defined by Eq.(\ref{eq def schmidt rank}),
and the logarithmic robustness of entanglement $LR(\ket{\Psi})$ \cite{Brandao05,HMMOV,Datta09} %defined by Eq.(\ref{eq def LR})
are characterized as 
\begin{align}
E\left(\ket{\Psi}\right)=H_{1}(\Psi),~
\log R_{S}(\ket{\Psi})=H_{0}(\Psi), ~
LR(\ket{\Psi})=H_{1/2}(\Psi).
\end{align}
In the following, for the unified treatment, 
we only use the notation $H_{1-s}(\Psi)$.
Also, we abbreviate $V(\Psi\|\rho_{mix})$ to $V(\Psi)$. 
That is, we have $V(\Psi) = \sum_{i} \lambda_i (\log \lambda_i+ H_1(\Psi))^2$.

Then, our previous paper \cite{OH11} shows the following 
propositions.
The Stein bounds are given as follows.
\begin{proposition}{\cite[Theorem 2]{OH11}}\Label{pro1} 
Given a real number $\epsilon \in (0,1)$ and a pure entangled state 
$|\Psi\rangle$, 
there exists a sufficiently large number $N$ such that
\begin{align}
& \beta_{n,\rightarrow}\left(\epsilon|\rho_{mix}\|\Psi \right)  = 
\beta_{n,\leftrightarrow}\left(\epsilon|\rho_{mix}\|\Psi \right)  = 
\beta_{n,sep}\left(\epsilon|\rho_{mix}\|\Psi \right)  = 0
\Label{2-5-6}
\end{align}
for $n \ge N$.
Further, for a given $\epsilon>0$, 
we have the following expansion.
\begin{align}
\log \beta_{n,\rightarrow}\left(\epsilon|\Psi \|\rho_{mix}\right) 
=& 
-n (\log d_A d_B -H_{1}(\Psi)) +o(n),
 \Label{2-5-7b}
\\
\log \beta_{n,\leftrightarrow}\left(\epsilon|\Psi \|\rho_{mix} \right)  
=& \log \beta_{n,sep}\left(\epsilon|\Psi \|\rho_{mix} \right)  +o(n)
\nonumber \\
=& -n (\log d_A d_B -H_{1}(\Psi)) +o(n).
\Label{2-5-8b}
\end{align}
\hfill $\square$\end{proposition}

The Hoeffding bounds are characterized as follows.
\begin{proposition}{\cite[(40) and (110)]{OH11}}\Label{pro2} 
Given a real number $r>0$ and a pure entangled state 
$|\Psi\rangle$, 
we have the following relation.
\begin{align}
H_{\rightarrow}\left(r|\Psi\|\rho_{mix} \right) 
\stackrel{\rm def}{=}&
\lim_{n \to \infty}-\frac{1}{n}
\log \beta_{n,\rightarrow}\left(e^{-nr}|\Psi\|\rho_{mix} \right)  
\nonumber\\
= &
\sup _{0 \le s <1} 
\frac{-s}{1-s}r-H_{s}(\Psi) +\log d_A d_B .
\Label{2-5-9}
\end{align}
This relation implies the following equation
for $r \ge r_{\to}\stackrel{\rm def}{=}-H_0'(\Psi)$:
\begin{align}
H_{\rightarrow}(r|\Psi\|\rho_{mix}) =\log d_A d_B - H_0(\Psi).
\end{align}
Further, when 
$r \ge \log d_A d_B -H_{1/2}(\Psi)$,
we have
\begin{align}
H_{sep}\left(r|\Psi\|\rho_{mix} \right) 
\stackrel{\rm def}{=}
\lim_{n \to \infty}-\frac{1}{n}\log 
\beta_{n,sep}\left(e^{-nr}|\Psi\|\rho_{mix} \right)  
=
\log d_A d_B - H_{1/2}(\Psi).\Label{9-11-1}
\end{align}
\hfill $\square$\end{proposition}

\subsection{Main results}
In this subsection, we give a short description of the main 
results of this paper. 
As a refinement of Proposition \ref{pro1}, we obtain the following theorem for Stein-Strassen bounds.
Here, remember that we have defined the function 
$\Phi(x)\stackrel{\rm def}{=}
\int_{-\infty}^{x}\frac{1}{\sqrt{2\pi}} e^{\frac{-y^2}{2}} d y $.

\begin{theorem}\Label{thm1}
When the Schmidt coefficient $\lambda_i$ in \eqref{eq schmidt decomposisiton psi} is not uniform,
we have the following expansions for a given $\epsilon>0$.
\begin{align}
&\log \beta_{n,\rightarrow}\left(\epsilon|\Psi \|\rho_{mix}\right) \nonumber \\
=& 
-n (\log d_A d_B -H_{1}(\Psi)) -\sqrt{n} \sqrt{V(\Psi)} \Phi^{-1}(\epsilon)- \frac{1}{2}\log n +O(1) , \Label{2-5-7}
\\
&\log \beta_{n,\leftrightarrow,2}\left(\epsilon|\Psi \|\rho_{mix} \right)  
= 
\log \beta_{n,\leftrightarrow}\left(\epsilon|\Psi \|\rho_{mix} \right)  +O(1)
\nonumber \\
=& \log  \beta_{n,sep}\left(\epsilon|\Psi \|\rho_{mix} \right)  
+O(1) \nonumber \\
=& -n (\log d_A d_B -H_{1}(\Psi)) -\sqrt{n} \sqrt{V(\Psi)} \Phi^{-1}(\epsilon)- \log n +O(1).
\Label{2-5-8}
\end{align}
\hfill $\square$\end{theorem}
Relations (\ref{2-5-7}) and (\ref{2-5-8}) show that the difference between 
$\log \beta_{n,\rightarrow}\left(\epsilon|\Psi \|\rho_{mix}\right) $
and 
$\log \beta_{n,\leftrightarrow}\left(\epsilon|\Psi \|\rho_{mix}\right) $
exists only on the order of $\log n$.
However, there is no difference with the uniform Schmidt coefficient as follows.

\begin{theorem}\Label{thm1-B}
When the Schmidt coefficient $\lambda_i$ in \eqref{eq schmidt decomposisiton psi} is uniform,
we have the following expansions for a given $\epsilon>0$.
\begin{align}
&\beta_{n,\rightarrow}\left(\epsilon|\Psi \|\rho_{mix}\right) 
=\beta_{n,\leftrightarrow}\left(\epsilon|\Psi \|\rho_{mix}\right) 
=\beta_{n,sep}\left(\epsilon|\Psi \|\rho_{mix}\right) 
=\max\{0,1- \bar{d}^n \epsilon\}
\Label{2-5-8B},
\end{align}
where $\bar{d}:= \max (d_A, d_B)$.
\hfill $\square$\end{theorem}

\begin{theorem}\Label{thm1-B2}
For the Hoeffding bounds of two-way LOCC and separable cases,
we obtain the following relations.
\begin{align}
& 
\lim_{n \to \infty}
-\frac{1}{n}\log \beta_{n,\leftrightarrow,2}\left(e^{-nr} |\Psi\|\rho_{mix} \right)  = 
\lim_{n \to \infty}
-\frac{1}{n}\log \beta_{n,\leftrightarrow}\left(e^{-nr} |\Psi\|\rho_{mix} \right) \nonumber \\
= &
\lim_{n \to \infty}
-\frac{1}{n}\log \beta_{n,sep}\left(e^{-nr} |\Psi\|\rho_{mix} \right) 
\nonumber\\
=&
\sup _{0 \le s <1} 
\frac{-2s}{1-s}r-H_{\frac{1+s}{2}}(\Psi) +\log d_A d_B.
%- n H_{\leftrightarrow}\left(r|\Psi\|\rho_{mix} \right) +o(n) ,
\Label{2-5-10}
\end{align}
\hfill $\square$\end{theorem}
This theorem concludes that
the Chernoff bound for the two-way LOCC case equals that for the separable case, which was an open problem in the previous paper \cite{OH11}.

Since $H_{\frac{1+s}{2}}(\Psi)$ monotonically decreases for $s$,
the supremum $\sup _{0 \le s <1} 
\frac{-2s}{1-s}r-H_{\frac{1+s}{2}}(\Psi) +\log d_A d_B$
is realized with $s\to 0$ when $r \ge r_{\leftrightarrow}\stackrel{\rm def}{=}-\frac{1}{4}H_{1/2}'(\Psi)$.
In this case, the Hoeffding bounds for two-way LOCC and separable cases
coincide with the right hand side of \eqref{9-11-1}.
Since the convexity of $s H_{1+s}(\Psi)$ implies that
\begin{align*}
&\log d -H_{\frac{1}{2}}(\Psi) \ge H_{0}(\Psi) -H_{\frac{1}{2}}(\Psi) \\
=&\frac{1}{2} \frac{\frac{-1}{2}H_{1-\frac{1}{2}}(\Psi)-(-H_{1-1}(\Psi))}{-\frac{-1}{2}-(-1)}
 -\frac{1}{2}H_{\frac{1}{2}}(\Psi) \\
=&\frac{1}{2} \frac{d sH_{1+s}(\Psi)}{ds}\Bigr|_{s=-\frac{1}{2}} -\frac{1}{2}H_{\frac{1}{2}}(\Psi) 
=-\frac{1}{4}H_{1/2}'(\Psi),
\end{align*}
this argument can be regarded as an extension of \eqref{9-11-1} in Proposition \ref{pro2}.

The right hand sides of \eqref{2-5-9} and \eqref{2-5-10} are numerically 
calculated as shown in  Figs. \ref{g1} and \ref{g2} 
when the pure entangled state $\ket{\Psi}$
is given as a pure state $\ket{\Psi(\lambda)}$:
\begin{equation}\label{eq def Psi lambda}
 \ket{\Psi(\lambda)}=\sqrt{\lambda}\left(\sum_{i=1}^{d-1} \ket{ii} \right)+\sqrt{1-(d-1)\lambda}\ket{dd},
\end{equation}
where $\lambda$ satisfies $0\le \lambda \le 1/\sqrt{d}$.
The graphs in Figs. \ref{g1} and \ref{g2} show 
the typical points  
$r_{\to}$ and $r_{\leftrightarrow}$ on the horizontal line and
$\log d_A d_B - H_0(\Psi)$, $\log d_A d_B - H_{1/2}(\Psi)$, and
$\log d_A d_B - H_1(\Psi)$ on the vertical line.
Note that $\ket{\Psi(0)}$ is a product state
and $\ket{\Psi(1/\sqrt{d})}$ is a maximally entangled state. 
The results in Figs. \ref{g1} and \ref{g2} show that
two-way LOCC improves the Hoeffding bound when $r$ is large.

\begin{figure}[htbp]
\begin{center}
%\scalebox{0.6}{
\includegraphics[scale=0.8]{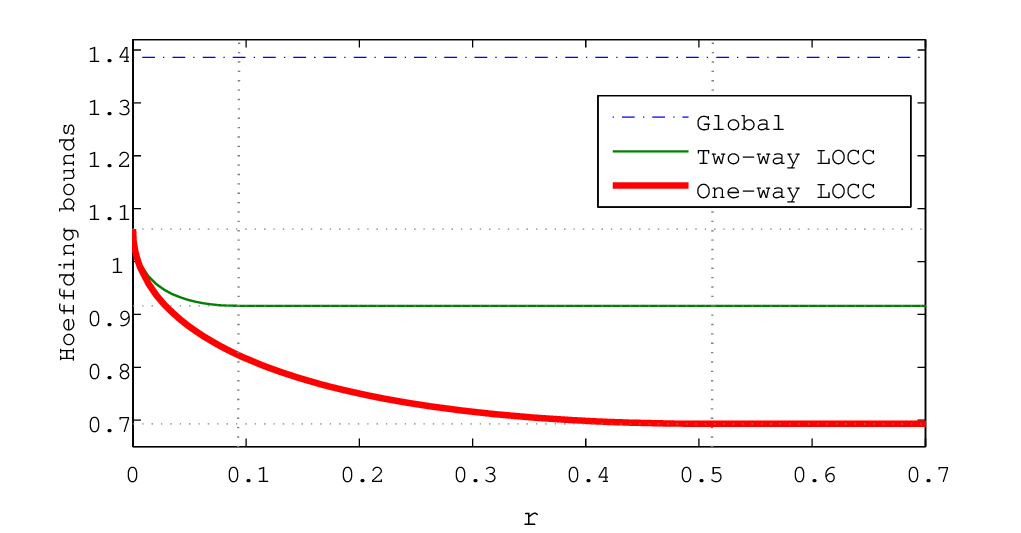}
%}
\end{center}
\caption{Comparison of Hoeffding bounds in one-way LOCC and two-way LOCC 
 when $d=2$ and $\lambda=0.1$.
In this case, we have $r_{\to}=0.511$, $r_{\leftrightarrow}=0.092$,
$\log d_A d_B - H_0(\Psi)=0.693$,
$\log d_A d_B - H_{1/2}(\Psi)=0.916$,
and
$\log d_A d_B - H_1(\Psi)=1.061$.}
\Label{g1}
\end{figure}%

\begin{figure}[htbp]
\begin{center}
%\scalebox{0.6}{
\includegraphics[scale=0.8]{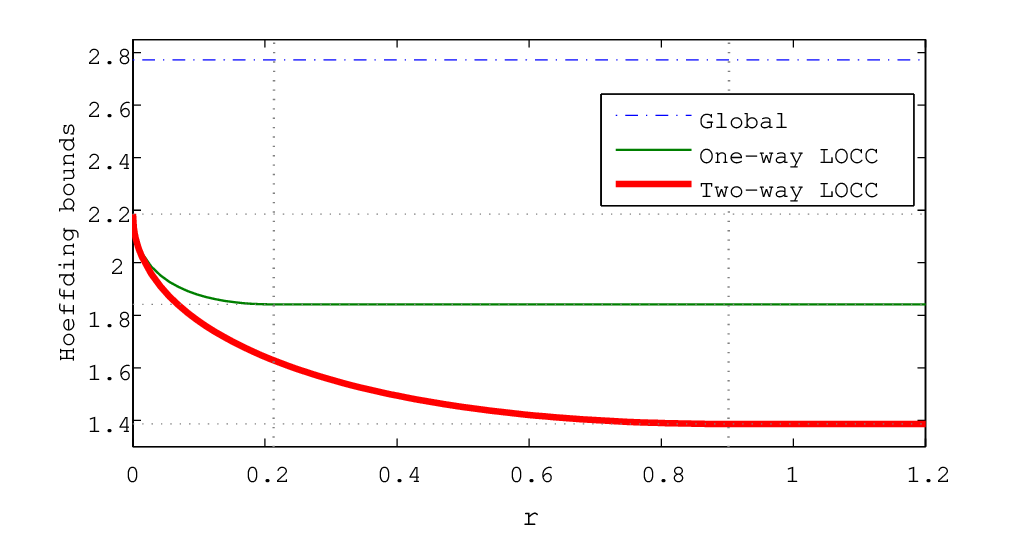}
%}
\end{center}
\caption{Comparison of Hoeffding bounds in one-way LOCC and two-way LOCC 
 when $d=4$ and $\lambda=0.05$.
In this case, we have $r_{\to}=0.911$, $r_{\leftrightarrow}=0.212$,
$\log d_A d_B - H_0(\Psi)=1.386$,
$\log d_A d_B - H_{1/2}(\Psi)=1.841$,
and
$\log d_A d_B - H_1(\Psi)=2.185$.}
\Label{g2}
\end{figure}%

\section{Hypothesis testing under one-way LOCC POVMs} \Label{sec one-way}
In this section, 
to show the relations for the one-way LOCC POVMs in Proposition \ref{pro2} and Theorems \ref{thm1} and \ref{thm1-B} 
(\eqref{2-5-9},
(\ref{2-5-7}), and 
$\beta_{n,\rightarrow}\left(\epsilon|\Psi \|\rho_{mix}\right) 
=\max\{0,1- \bar{d}^n \epsilon\}$),
we consider $C=\rightarrow$, that is, the local hypothesis testing under one-way LOCC POVMs.
In this case, it turns out that our results can be formulated in terms of the following state
\begin{equation}\Label{eq def sigma psi}
 \sigma_{\Psi}\stackrel{\rm def}{=} \sum_{i=1}^d\lambda_i\ket{i}\bra{i}\otimes \ket{i}\bra{i},
\end{equation}
where $\left\{\ket{i}\otimes \ket{j}\right\}_{i,j}$ 
is the Schmidt basis of $\ket{\Psi}$ [see Eq.(\ref{eq schmidt decomposisiton psi})].
Then, our hypothesis testing is reduced to that with states
$ \sigma_{\Psi}$ and $\rho_{mix}$.
That is, the last paper \cite{OH11} showed the following lemma:
\begin{proposition}{Lemma 1 of \cite{OH11}} \Label{sec one-way lemma}
For all $\alpha >0$, we have
\begin{align}
%\alpha_{n,\rightarrow}
% (\beta|\rho_{mix}\|\Psi)&=\alpha_{n}(\beta|\rho_{mix}\|\sigma_{\Psi}), \Label{eq one-way equal classical 1}\\
\beta_{n,\rightarrow}
 (\alpha|\Psi \| \rho_{mix})&=\beta_{n}(\alpha|\sigma_{\Psi}\| \rho_{mix}), \Label{eq
 one-way equal classical 2}
\end{align}
where the optimal type-2 error probability $\beta_{n,\rightarrow}
 (\alpha|\rho \|\sigma)$ is defined by Eq.
(\ref{eq def beta n C alpha rho sigma}). 
\hfill $\square$\end{proposition}

\begin{proofsof}{(\ref{2-5-7}) and (\ref{2-5-9})}
Since
\begin{align}
D(\Psi \| \rho_{mix})&= \log d_A d_B- H_1(\Psi),\quad
V(\Psi \| \rho_{mix})=V(\Psi),\\
\psi(s|\rho_{mix}\|\Psi )&= (1-s)(H_{s}(\Psi) -\log d_A d_B ),
\end{align}
by applying (\ref{2-4-4})
to the commutative states 
$\rho=\rho_{mix}$ and $\sigma= \sigma_{\Psi}$,
Proposition \ref{sec one-way lemma} yields
(\ref{2-5-7}).
Similarly, 
applying (\ref{2-4-3}),
Proposition \ref{sec one-way lemma} reproduces the existing result
(\ref{2-5-9}).
Therefore, we obtain the results for the one-way LOCC case.
\end{proofsof}

\begin{proofof}{$\beta_{n,\rightarrow}\left(\epsilon|\Psi \|\rho_{mix}\right) =\max\{0,1- \bar{d}^n \epsilon\}$}
For the two hypotheses $\sigma_{\Psi}$ and $ \rho_{mix}$,
the optimal test $T$ has the support in the $n$-tensor product space of the subspace spanned by $\{|ii\rangle \}_{i=1}^d$
when $\epsilon \le \frac{1}{\bar{d}^n}$. 
In this case, when 
$\Tr (I^n-T_n) \rho_{mix}^{\otimes n}= \epsilon$,
we have 
$\Tr (I^n-T_n) \sigma_{\Psi}^{\otimes n}= \bar{d}^n \epsilon$.
So, we obtain $\beta_{n,\rightarrow}\left(\epsilon|\Psi \|\rho_{mix}\right) =1- \bar{d}^n \epsilon$ 
\end{proofof}

\section{Hypothesis testing under separable POVM}\Label{sec separable}
\subsection{Uniform case: Proof of Theorem \ref{thm1-B}}
First, we consider the most simple case when 
the Schmidt coefficient is uniform, i.e., $\sqrt{\frac{1}{d^n}}$
because $d=\min (d_A,d_B)$.
Then, for any separable POVM $\{T_n ,I^n-T_n\}$, we have\cite{HMMOV06}
\begin{align}
\Tr T_n |\Psi\rangle \langle \Psi|
\le \frac{1}{d^n}\Tr T_n
= d^n \Tr T_n \rho_{mix}.
\end{align}
%The equality can be attained when we choose suitable separable POVM.
Hence, when the first kind of error probability is restricted to
$\Tr T_n |\Psi\rangle \langle \Psi|=\epsilon$,
the second kind of error probability
is evaluated as $\Tr (I^n-T_n) \rho_{mix} \ge 1-\bar{d} \epsilon$.
Hence, we have
\begin{align}
\beta_{n,sep}\left(\epsilon|\Psi \|\rho_{mix}\right) 
\ge \max\{0,1- \bar{d} \epsilon\}
\end{align}
Since this lower bound can be attained by one-way LOCC, as mentioned in Section \ref{sec one-way},
we obtain \eqref{2-5-8B}.

\subsection{Hypothesis testing with a composite hypothesis: Proof of Theorem \ref{thm1-B2}}
In this subsection, 
in order to consider hypothesis testing under separable POVM
for a pure state with the Schmidt decomposition 
$\sum_{i=1}^d \sqrt{p_i}|i\rangle\otimes |i\rangle$,
we consider a pure state $|\varphi\rangle\stackrel{\rm def}{=}\sum_{i=1}^d \sqrt{p_i}|i\rangle\in \mathbb{C}^d$ and
a specific composite hypothesis testing on $(\mathbb{C}^d)^{\otimes n}$
by employing the results in \cite{OH10}. 
Here, we assume that $p_1\ge p_2 \ge \ldots \ge p_d$.
\subsubsection{Single-shot setting}\Label{subsub1}
Although our problem is based on $n$-fold setting,
it is quite hard to find the relation between our problem and the results in  \cite{OH10}.
To reduce the difficulty, we firstly discuss this relation 
with the single-shot setting.
That is, in this subsubsection, we consider this specific composite hypothesis testing with the single-shot setting.
Here, we assume that the distribution $p=(p_i)$ is not uniform
due to the assumption of Theorem \ref{thm1}.
The following type of composite hypothesis testing
plays a key role in our analysis of our hypothesis testing in the bipartite system.
The null hypothesis is given as the pure state $|\varphi\rangle$ 
in the system $\mathbb{C}^d$.
To give the alternative hypothesis, we introduce a notation.
In the quantum system $\mathbb{C}^d$,
the basis is written as $\ket{j}$ by using $j \in \{1, \ldots, d\}$.
Hence, the quantum system $\mathbb{C}^d$ is spanned by
$\{\ket{j} \}_{j \in {\cal D}}$,
where ${\cal D}\stackrel{\rm def}{=}\{1, \ldots, d\}$.
Then, the alternative hypothesis ${\cal S}_{0}$ is the set of states
$\left \{ \ket{\phi_{L}} \right \}_{L \in \mathbb{Z}^{d}_2}$,
where $\ket{\phi_{L}} \in \mathbb{C}^{d}$ is defined as 
\begin{equation}\Label{eq def phi l nB}
 \ket{\phi_{L}}\stackrel{\rm def}{=}
\sqrt{\frac{1}{d}} \sum _{j \in {\cal D}}
  (-1)^{L_j} \ket{j}, \quad L \in \mathbb{Z}_2^{d},
\end{equation}
where $L_j \in \mathbb{Z}_2$ is the $j$th entry of $L\in \mathbb{Z}_2^{d}$.
That is, an element of the alternative hypothesis is characterized by
an element of $ \mathbb{Z}_2^{d}$.
Hence, the cardinality of the alternative hypothesis is $2^{d}$. 

For a two-valued POVM $\{ S, I-S \}$ on
$\mathbb{C}^{d}$, 
the type-1 error $\alpha(S)$ and type-2 error $\beta(S)$ are defined as 
\begin{align}
\alpha(S) &\stackrel{\rm def}{=} 
\Tr S \ket{\varphi}\bra{\varphi}
\Label{eq def a nB}  \\
\beta(S) &\stackrel{\rm def}{=} 
\max_{\rho \in {\cal S}_{0}} \Tr (I_{d}-S) \rho,
\Label{eq def b nB}
\end{align}
where $I_{d}$ is an identity operator on $\mathbb{C}^{d}$.
The optimal type-2 error under the restriction on the condition that the
type-1 error is no more than $\alpha \ge 0$ can be written as 
\begin{equation}\Label{eq def gamma n aB}
 \beta ( \alpha |\varphi ) 
  \stackrel{\rm def}{=} \min_{0 \le S \le I_{d }} \{ \beta(S) |
  \alpha(S) \le \alpha\}.
\end{equation}
Similarly, we define $\alpha ( \beta |\varphi )$ as
\begin{equation}
 \alpha ( \beta |\varphi ) 
  \stackrel{\rm def}{=} \min_{0 \le S \le I_{d }} \{ \alpha(S) |
  \beta(S) \le \beta\}.\Label{eq sep x epsilon = max lastB}
\end{equation}
In the rest of this subsection, 
we often abbreviate $ \beta( \alpha |\varphi)$
as $ \beta \left( \alpha \right )$.

Now, we define the subset $S(R)\stackrel{\rm def}{=}
\{ j \in {\cal D}
| \log p_{j} \ge  R\} $ of ${\cal D}$.
We also employ the following notations:
\begin{align}
P_{s}(R) \stackrel{\rm def}{=} & \sum_{j \in S(R)} (p_{j})^s 
\hbox{ for }s=0,1/2,1,
\\
%, \quad
%P_{n,1/2}(R) \stackrel{\rm def}{=} \sum_{J \in S_n(R)} \sqrt{p^n_{J}} , \quad
%P_{n,0}(R) \stackrel{\rm def}{=} | S_n(R)| 
\epsilon(R)\stackrel{\rm def}{=} &
\Big(\frac{P_{1/2}(R)^2}{d P_{1}(R)}\Big)^{\frac{1}{2}}.
\end{align}
When 
\begin{align}
1\ge \frac{P_{1/2}(\hat{R})^2}{P_{0}(R) P_{1}(\hat{R})},
\Label{C2-1}
\end{align}
we define
\begin{align}
%P_{n,1}^c(R) &\stackrel{\rm def}{=} 1- P_{n,1}(R) \\
&a(R,\hat{R}) \nonumber \\
\stackrel{\rm def}{=}&1-
P_{1}({R})
\Biggl(
\frac{P_{1/2}(R)P_{1/2}(\hat{R})}{P_{1}(\hat{R})^{\frac{1}{2}}P_{1}(R)^{\frac{1}{2}} P_{0}(R)}
+
\Bigl(
1-\frac{P_{1/2}(\hat{R})^2}{P_{1}(\hat{R})P_{0}(R)}
\Bigr)^{\frac{1}{2}}
\Bigl(
1-\frac{P_{1/2}(R)^2}{P_{1}(R) P_{0}(R)}
\Bigr)^{\frac{1}{2}}
\Biggr)^2
 \Label{11-05-7B}.
\end{align}
Then, we have the following lemma.

\begin{lemma}\Label{Th1B-1}
When the inequality $\epsilon(R) \ge \epsilon(\hat{R})$ holds,
the condition \eqref{C2-1} holds.
\hfill $\square$\end{lemma}
\begin{proof}
Since Schwarz inequality implies that
$P_{0}(R) \ge \frac{P_{1/2}(R)^2}{P_{1}(R)}$, which implies 
the condition \eqref{C2-1}.
\end{proof}

Lemma \ref{L-2-18B} in Appendix \ref{A4} yields the following lemma.
\begin{lemma}\Label{Th1B-2}
Any two distinct elements $R$ and  $\hat{R}$ of $\{ \log p_i\}_i$ with $R>\hat{R}$
satisfy the inequality $\epsilon(\hat{R}) > \epsilon({R})$.
\hfill $\square$\end{lemma}

So, we have the following lemma.
\begin{lemma}\Label{Th1B}
%We assume that two real numbers $R<\hat{R}$ belong to the set $\{\log p_{i}\}_{i}$. 
%Then, 
We fix $\hat{R} \in \{\log p_{i}\}_{i}$.
Then, we have the following items.
\begin{description}
\item[(1)]
When a real number $R(\in \{\log p_{i}\}_{i}) <\hat{R}$
satisfies %\eqref{C2-1} and
%Given $\hat{R}$, we choose the minimum $R$ satisfying
\begin{align}
\frac{P_{0}(R)}{P_{1/2}(R)}
e^{\frac{ R}{2}} & \ge
1-
\frac{\Bigl(\frac{P_{1}(R)P_{0}(R)}{P_{1/2}(R)^2}-1\Bigr)^{\frac{1}{2}}
}
{\Bigl(\frac{P_{1}(\hat{R})P_{0}(R)}{P_{1/2}(\hat{R})^2}
-1\Bigr)^{\frac{1}{2}}
},
\Label{11-4-1B} 
\end{align}
we have
\begin{align}
\alpha\left(
\left.
\frac{P_{1/2}(\hat{R})^2}{d P_{1}(\hat{R})}
\right|\varphi\right)
\le
a(R,\hat{R}) 
\Label{4-6-1bB} .
\end{align}
\item[(2)]
We assume that there exists an element $R_\alpha$ in $\{\log p_{i}\}_{i}$ 
satisfying the inequality \eqref{11-4-1B} and $ R_\alpha< \hat{R}$.
We denote all of distinct elements of $\{\log p_{i} \}_{i}$ 
by $\tilde{R}_k < \tilde{R}_{k-1}< \ldots < \tilde{R}_2 < \tilde{R}_{1}$.
We also assume that an element $R_\beta$ in $\{\log p_{i}\}_{i}$ satisfying the following condition;
Any element $\tilde{R}_j (\le R_{\beta})\in \{\log p_{i}\}_{i}$ satisfies
%the following condition.
%We denote all of distinct elements of $\{\log p_{i}| \log p_{i} \le \tilde{R} \}_{i}$ by $\tilde{R}_k < \tilde{R}_{k-1}< \ldots < \tilde{R}_2 < \tilde{R}_{1}=R_\beta$.
%These elements satisfy
\begin{align}
\frac{P_{0}(\tilde{R}_{j+1})}{P_{1/2}(\tilde{R}_{j})}
e^{\frac{ \tilde{R}_{j}}{2}} & <
1-
\frac{\Bigl(\frac{P_{1}(\tilde{R}_{j+1})
P_{0}(\tilde{R}_{j+1})}{P_{1/2}(\tilde{R}_{j})^2}-1\Bigr)^{\frac{1}{2}}
}
{\Bigl(\frac{P_{1}(\hat{R})P_{0}(\tilde{R}_{j})}{P_{1/2}(\hat{R})^2}
-1\Bigr)^{\frac{1}{2}}
}.
\Label{11-4-1B-P} 
\end{align}
Then, the real numbers $\hat{R}$ and $R_\beta$ satisfy the inequality
\begin{align}
\alpha\left(
\left.
\frac{P_{1/2}(\hat{R})^2}{d P_{1}(\hat{R})}
\right|\varphi\right)
\ge
a(R_\beta,\hat{R}) 
\Label{4-6-1bB-P} .
\end{align}
\hfill $\square$
\end{description}
\end{lemma}

\begin{proofof}{Lemma \ref{Th1B-2}}
Here, we employ notations summarized in Appendix \ref{A2}.
That is, we define the real vectors $u_l$ and $v_l$ on $\mathbb{R}^l$ as 
$u_l \stackrel{\rm def}{=}\left( \sqrt{p_1}, \cdots , \sqrt{p_l} \right)$ 
and $v_l \stackrel{\rm def}{=} \left( 1,  \cdots , 1 \right) /\sqrt{d}$
for an integer $l$ satisfying $1 \le l \le d$.
We consider only the case when 
\begin{align}
l:= | S(R)| =P_{0}(R),\quad 
\hat{l}:=P_{0}(\hat{R}), \quad
\epsilon^2:=\epsilon(\hat{R})^2
=(\frac{u_{\hat{l}}\cdot v_{\hat{l}}}{\|u_{\hat{l}}\|})^2
=\frac{P_{1/2}(\hat{R})^2}{d P_{1}(\hat{R})}.
\Label{2-18-1}
\end{align}
Since $R$ and $\hat{R}$
are two distinct elements of $\{ \log p_i\}_i$ and $R>\hat{R}$,
we have the inequality $l < \hat{l}$.
Due to the above final relation, 
Lemma \ref{L-2-18B} in Appendix \ref{A4} directly implies Lemma \ref{Th1B-2}.
\end{proofof}

\begin{proofof}{Lemma \ref{Th1B}}
We prove Lemma \ref{Th1B} by using the notations in the above proof of Lemma \ref{Th1B-2}.
For this purpose, we employ results in \cite{OH10}, which are summarized in Appendix \ref{A2}.
%Note that $\hat{l} \le l $.
Due to \eqref{2-18-1}, we have
\begin{align}
u_l \cdot v_l= \frac{P_{1/2}(R)}{d^{1/2}},\quad
\|v_l \|^2= \frac{P_{0}(R)}{d},\quad
\|u_l \|^2= P_{1}(R).
\Label{2-18-2}
\end{align}
Since $R< \hat{R}$, we have $l \ge \hat{l}$.
%Also, since $u_{\hat{l}} \cdot v_l= \frac{P_{1/2}(\hat{R})}{d^{1/2}}$,
%the inequality \eqref{C2-1} implies $\|u_{\hat{l}}\|^2 \|v_l\|^2-  (u_{\hat{l}} \cdot v_l)^2 \ge 0$.

\PF{Item (1)}
Firstly, we show Item (1) by using the properties of
${x}^*(u_l,v_l,\epsilon)$ given in Proposition \ref{Lemma 9}.
That is, we show \eqref{4-6-1bB} by assuming %\eqref{C2-1} and 
\eqref{11-4-1B}.
Since $R<\hat{R}$, Lemma \ref{Th1B-2} guarantees that
$\frac{u_l \cdot v_l}{\|u_l \| }=
\epsilon(R) > \epsilon(\hat{R})=\epsilon$, which implies 
\eqref{C2-1} by using Lemma \ref{Th1B-1}.
Hence, the vector ${x}^*(u_l,v_l,\epsilon)$ 
defined in \eqref{eq def x*} of Proposition \ref{Lemma 9} in Appendix \ref{A2} is written as
\begin{align}
& {x}^*(u_l,v_l,\epsilon) \nonumber \\
=&
\frac{1}{\sqrt{
\|u_l\|^2 \|v_l\|^2-  (u_l \cdot v_l)^2}}
\bigg(
\sqrt{\|v_l\|^2 -\epsilon^2}u_l \nonumber \\
& \hspace{23ex}
+
\frac{
\epsilon \sqrt{\|u_l\|^2\|v_l\|^2- (u_l\cdot v_l)^2}
-u_l\cdot v_l \sqrt{\|v_l\|^2-\epsilon^2}}
{\|v_l\|^2}v_l \bigg).
\Label{19-eq}
\end{align}
Due to Lemma \ref{L-2-18C} in Appendix \ref{A4},
all entries of ${x}^*(u_l,v_l,\epsilon)$
are non-negative if and only if
\begin{align}
\sqrt{p_l}
\frac{d^{1/2} \|v_l\|^2}{u_l \cdot v_l}
\ge
\left(1-
\frac{\sqrt{\frac{\|u_l\|^2\|v_l\|^2}{(u_l \cdot v_l)^2}-1}}
{\sqrt{\frac{\|v_l\|^2}{\epsilon^2}-1}}
\right) ,
\end{align}
which is equivalent to \eqref{11-4-1B}, due to the relations
\eqref{2-18-1} and \eqref{2-18-2}.
So, all entries of ${x}^*(u_l,v_l,\epsilon)$ are non-negative.
Thus, %when the inequality \eqref{11-4-1B} holds,
for any $L \in \mathbb{Z}_2^d$,
we find that 
\begin{align}
\langle x^*(u_l,v_l,\epsilon)|\phi_L\rangle^2
\stackrel{(a)}{\le }
\Big\langle x^*(u_l,v_l,\epsilon)\Big| \sum_{i=1}^d\frac{1}{\sqrt{d}}\Big|i\Big\rangle^2
\stackrel{(b)}{=}
\langle x^*(u_l,v_l,\epsilon)| v_l \rangle^2
\stackrel{(c)}{\le }
 \epsilon^2, \Label{19-eq2}
\end{align}
where 
$(a)$, $(b)$, and $(c)$ follow from 
the non-negativity of all entries of $x^*(u_l,v_l,\epsilon)$,
the equations \eqref{19-eq}, and 
the property of ${x}^*(u_l,v_l,\epsilon)$ given in Proposition \ref{Lemma 9}, respectively.
Thus, since $|\varphi\rangle= u_d$,
using \eqref{19-eq} and \eqref{19-eq2}, we have
\begin{align}
&1- \alpha\left(
\left.
\epsilon^2
\right|\varphi\right)
\stackrel{(a)}{\ge }
\langle \varphi| x^*(u_l,v_l,\epsilon) \rangle^2
\nonumber\\
\stackrel{(b)}{=}&
\frac{(
u_l \cdot v_l \epsilon+\sqrt{
(\|v_l\|^2 -\epsilon^2)
(\|v_l\|^2 \|u_l\|^2- (u_l \cdot v_l)^2)})^2
}{\|v_l\|^4} \nonumber \\
\stackrel{(c)}{=}&
\frac{
\left(
\frac{P_{1/2}(R)}{d^{1/2}}
\frac{P_{1/2}(\hat{R})}{d^{1/2}P_{1}(\hat{R})^{1/2}}
+
\sqrt{
(\frac{P_{0}(R)}{d} -\frac{P_{1/2}(\hat{R})^2}{d P_{1}(\hat{R})})
(\frac{P_{0}(R)P_{1}(R)}{d} -\frac{P_{1/2}(R)^2}{d})
}
\right)^2
}{\frac{P_{0}(R)^2}{d^{2}}} \nonumber \\
=& 
P_{1}({R})
\Biggl(
\frac{P_{1/2}(R)P_{1/2}(\hat{R})}{P_{1}(\hat{R})^{\frac{1}{2}}P_{1}(R)^{\frac{1}{2}} P_{0}(R)}
+
\Bigl(
1-\frac{P_{1/2}(\hat{R})^2}{P_{1}(\hat{R})P_{0}(R)}
\Bigr)^{\frac{1}{2}}
\Bigl(
1-\frac{P_{1/2}(R)^2}{P_{1}(R) P_{0}(R)}
\Bigr)^{\frac{1}{2}}
\Biggr)^2,
\Label{11-07-26}
\end{align}
where
$(a)$, $(b)$, and $(c)$ follow from 
\eqref{19-eq2}, \eqref{19-eq} with $|\varphi\rangle= u_d$, and \eqref{2-18-2}, respectively.
So, we obtain the inequality \eqref{4-6-1bB}.

\PF{Item (2)}
\noindent {\it Step 1:)}
Next, we proceed to the proof of Item (2) by combining Propositions \ref{Lemma 9} and \ref{Theorem 4}.
That is, we will show \eqref{4-6-1bB-P} by assuming \eqref{11-4-1B-P}.
Now, we outline the derivation of \eqref{4-6-1bB-P}. 
For the preparation, we choose 
$l_\beta:= | S(R_\beta)| =P_{0}(R_\beta)$,
$\hat{l}:= | S(\hat{R})| =P_{0}(\hat{R})$,
$l_\alpha:= | S(R_\alpha)| =P_{0}(R_\alpha)$, $\epsilon:= \epsilon(\hat{R})
=\frac{u_{\hat{l}}\cdot v_{\hat{l}}}{\|u_{\hat{l}}\|}$,
and $\eta:=\eta_\epsilon$,
where $\eta_\epsilon$ is defined in Appendix \ref{A2}.
In Step 2:), we show the inequality $l_\beta > \eta $. 
In Step 3:), we show 
\begin{align}
1- \alpha\left(
\left.
\epsilon^2
\right|\varphi\right)
=
\frac{\Big(
u_\eta \cdot v_\eta \epsilon+\sqrt{
(\|v_\eta\|^2 -\epsilon^2)
(\|v_\eta\|^2 \|u_\eta\|^2- (u_\eta \cdot v_\eta)^2)}\Big)^2
}{\|v_\eta\|^4} \Label{1-5-10C},
\end{align}
and
\begin{align}
& \max
\{ \langle u_{l_\beta} | \phi\rangle ^2
| |\phi\rangle \in \mathbb{R}^{l_\beta}, \|\phi\|^2=1,
\langle v_{l_\beta} | \phi\rangle \le \epsilon\} \nonumber \\
=&
\frac{\Big(
u_{l_\beta} \cdot v_{l_\beta} \epsilon+\sqrt{
(\|v_{l_\beta}\|^2 -\epsilon^2)
(\|v_{l_\beta}\|^2 \|u_{l_\beta}\|^2-  
(u_{l_\beta} \cdot v_{l_\beta})^2)}\Big)^2
}{\|v_{l_\beta}\|^4}  \Label{1-5-10}.
\end{align}
In Step 4:), combining these relations, we show the inequality \eqref{4-6-1bB-P}.

\noindent {\it Step 2:)}
Firstly, we show that the condition A1), A2), nor A3) 
in Appendix \ref{A2} does not hold for any integer $l$ satisfying $l \ge l_\beta$. 
Due to Lemma \ref{L-2-18}, A2) does not hold because $R_\beta < \log p_1$.
Since $l_\beta > \hat{l} $, 
Lemma \ref{L-2-18B} guarantees that A1) does not hold for any integer $l$ satisfying $l \ge l_\beta$. 

Now, to show the inequality $l_\beta > \eta $,
we show that A3) does not hold for any integer $l$ satisfying $l \ge l_\beta$.
We choose $\tilde{l}_{j}:=| S(\tilde{R}_{j})| =P_{0}(\tilde{R}_{j})$.
For a given integer $l \ge l_\beta$, we choose $j$ such that $\tilde{l}_{j+1}> l \ge \tilde{l}_{j}$, which implies the relations 
\begin{align}
p_l \le p_{\tilde{l}_{j}} ,~ 
\|v_{\tilde{l}_{j}}\|
\le\|v_l\| \le \|v_{\tilde{l}_{j+1}}\|,~
\|u_l\| \le \|u_{\tilde{l}_{j+1}}\|,~
u_{\tilde{l}_{j}} \cdot v_{\tilde{l}_{j}}.
\le u_l \cdot v_l \le 
u_{\tilde{l}_{j+1}} \cdot v_{\tilde{l}_{j+1}}.
\Label{8-29-1}
\end{align}
Then, we have
\begin{align}
\sqrt{p_l}
\frac{d^{1/2} \|v_l\|^2}{u_l \cdot v_l}
\stackrel{(a)}{\le}
\sqrt{p_{\tilde{l}_{j}}}
\frac{d^{1/2} \|v_{\tilde{l}_{j+1}}\|^2}{u_{\tilde{l}_{j}} \cdot v_{\tilde{l}_{j}}}
\stackrel{(b)}{<}
\left(1-
\frac{\sqrt{\frac{\|u_{\tilde{l}_{j+1}}\|^2\|v_{\tilde{l}_{j+1}}\|^2}{(u_{\tilde{l}_{j}} \cdot v_{\tilde{l}_{j}})^2}-1}}
{\sqrt{\frac{\|v_{\tilde{l}_{j}}\|^2}{\epsilon^2}-1}}
\right)
\stackrel{(c)}{\le}
\left(1-
\frac{\sqrt{\frac{\|u_l\|^2\|v_l\|^2}{(u_l \cdot v_l)^2}-1}}
{\sqrt{\frac{\|v_l\|^2}{\epsilon^2}-1}}
\right) ,
\end{align}
where 
$(b)$ follows from the condition \eqref{11-4-1B-P},
and 
$(a)$ and $(c)$ follow from \eqref{8-29-1}.
This inequality shows that the condition \eqref{L-2-18C-Eq} in Lemma \ref{L-2-18C} does not hold. 
Since Lemma \ref{L-2-18B} guarantees that $\frac{u_l \cdot v_l}{\| u_l \|}$
is strictly monotone increasing for $l$,
we have
$\frac{u_l \cdot v_l}{\| u_l \|} 
\ge \frac{u_{l_\beta} \cdot v_{l_\beta}}{\| u_{l_\beta} \|} 
> \frac{u_{\hat{l}} \cdot v_{\hat{l}}}{\| u_{\hat{l}} \|} = \epsilon$
because $l_\beta> \hat{l}$.
By using these two statements,
Lemma \ref{L-2-18C} guarantees that
the ${l}$-th entry of $x^*(u_{{l}},v_{{l}},\epsilon)$ is negative for the integer $l$.
So, A3) does not hold for any integer $l$ satisfying $l \ge l_\beta$.
Thus, the assumption of Item (2) implies that neither A1), A2), nor A3) does not hold for any integer $l$ satisfying $l \ge l_\beta$.
Hence, we have the desired inequality $l_\beta > \eta $.

\noindent {\it Step 3:)}
Since $R_\alpha<\hat{R}$, Lemma \ref{Th1B-2} implies that
$\frac{u_{l_\alpha} \cdot v_{l_\alpha}}{\|u_{l_\alpha}\|}
= \epsilon (R_\alpha) >\epsilon (\hat{R})=\epsilon $.
Item (1) guarantees that $l_\alpha$ satisfies Condition A3).
So, $\eta\ge l_\alpha$.
Thus, Lemma \ref{L-2-18B} yields that 
$\frac{u_{\eta} \cdot v_{\eta}}{\|u_{\eta}\|}
\ge \frac{u_{l_\alpha} \cdot v_{l_\alpha}}{\|u_{l_\alpha}\|}
>\epsilon$.
Hence, B1) does not hold. 
Since $R_\alpha<\hat{R}$ implies $R_\alpha<\log p_1$, 
we have $\log p_\eta \le R_\alpha<\log p_1$.
So, B2) does not hold due to Lemma \ref{L-2-18}.
Thus, B3) holds.
So, Proposition \ref{Theorem 4} guarantees \eqref{1-5-10C},
and the maximum 
$1- \alpha\left( \left. \epsilon^2 \right|\varphi\right)$ 
in \eqref{1-5-10C} is attained by
the vector $x^*(u_\eta,v_\eta,\epsilon)$. 

Then, we apply Proposition \ref{Lemma 9} to the case with 
$y=u_{l_\beta}$ and $z=v_{l_\beta}$.
Since the condition D3), i.e., 
the relation $y/\|y\| \neq z/\|z\| $ and $y\cdot z > \epsilon \|y\|$ holds,
we obtain \eqref{1-5-10}.

\noindent {\it Step 4:)}
We show the inequality \eqref{4-6-1bB-P}.
Since the inequality $l_\beta > \eta $ implies the equation
$\langle v_{l_\beta} | x^*(u_\eta,v_\eta,\epsilon)\rangle
=\langle v_{\eta} | x^*(u_\eta,v_\eta,\epsilon)\rangle$,
we find that 
the vector $x^*(u_\eta,v_\eta,\epsilon)$ also satisfies the condition for 
the real vector $|\phi\rangle$ in the maximum in the LHS of \eqref{1-5-10}.
So, we have
\begin{align}
& \max
\{ \langle u_{l_\beta} | \phi\rangle ^2
| |\phi\rangle \in \mathbb{R}^{l_\beta}, \|\phi\|^2=1,
\langle v_{l_\beta} | \phi\rangle \le \epsilon\} \nonumber \\
\ge &
\frac{\Big(
u_\eta \cdot v_\eta \epsilon+\sqrt{
(\|v_\eta\|^2 -\epsilon^2)
(\|v_\eta\|^2 \|u_\eta\|^2- (u_\eta \cdot v_\eta)^2)}\Big)^2
}{\|v_\eta\|^4} \Label{1-5-10B}.
\end{align}
Combining \eqref{1-5-10C}, \eqref{1-5-10}, and \eqref{1-5-10B},
we have
\begin{align}
& 1- \alpha \left(
\left.
\epsilon^2
\right|\varphi\right) \nonumber \\
\le &
\frac{\Big(
u_{l_\beta} \cdot v_{l_\beta} \epsilon+\sqrt{
(\|v_{l_\beta}\|^2 -\epsilon^2)
(\|v_{l_\beta}\|^2 \|u_{l_\beta}\|^2- 
(u_{l_\beta} \cdot v_{l_\beta})^2)}\Big)^2
}{\|v_{l_\beta}\|^4} .
\end{align}
Hence, combining the same discussion as \eqref{11-07-26}, we obtain 
the inequality \eqref{4-6-1bB-P}.
\end{proofof}

Note that
it is quite difficult to derive the tight evaluation of 
$\alpha\left(
\left.
\frac{P_{1/2}(\hat{R})^2}{d P_{1}(\hat{R})}
\right|\varphi\right)$
because our choice of $l$ is limited to $l= | S(R)| =P_{0}(R)$.
We obtain lower and upper bounds as \eqref{4-6-1bB}.

Using $P_{1}^c(R)\stackrel{\rm def}{=} 1-P_{1}(R)$, we have the following lemma.
\begin{lemma}\Label{11-06-6B}
When $R<\hat{R}$,
the number $a(R,\hat{R}) $ is bounded as follows.
\begin{align}
P_{1}^c(R)
\le
a(R,\hat{R}) 
\le
P_{1}^c(\hat{R}).
\end{align}
\hfill $\square$\end{lemma}

\begin{proof}
To show Lemma \ref{11-06-6B}, we will show the following.
\begin{align}
P_{1}(\hat{R})
\le
1- a(R,\hat{R}) 
\le
P_{1}(R) \Label{11-07-22b}.
\end{align}

First, we show the second inequality of \eqref{11-07-22b}.
Since
\begin{align*}
&\frac{P_{1/2}(\hat{R})^2}{P_{0}(R)P_{1}(\hat{R})}
+
\frac{P_{1/2}(R)^2}{P_{0}(R)P_{1}(R)}
\ge
2 
\Big(\frac{P_{1/2}(\hat{R})^2}{P_{0}(R)P_{1}(\hat{R})}
\cdot
\frac{P_{1/2}(R)^2}{P_{0}(R)P_{1}(R)}\Big)^{\frac{1}{2}} \\
=&
2 
\frac{P_{1/2}(R)P_{1/2}(\hat{R})}{P_{0}(R)
P_{1}(R)^{\frac{1}{2}} P_{1}(\hat{R})^{\frac{1}{2}}},
\end{align*}
we have
\begin{align*}
& \Bigl(
1-
\frac{P_{1/2}(\hat{R})^2}{P_{0}(R)P_{1}(\hat{R})}
\Bigr)
\Bigl(
1-
\frac{P_{1/2}(R)^2}{P_{0}(R)P_{1}(R)}
\Bigr) \\
\le &
\Bigl(
1-
\frac{P_{1/2}(R)P_{1/2}(\hat{R})}{P_{0}(R)
P_{1}(R)^{\frac{1}{2}} P_{1}(\hat{R})^{\frac{1}{2}}}
\Bigr)^2,
\end{align*}
i.e.,
\begin{align*}
& \Bigl(
1-
\frac{P_{1/2}(\hat{R})^2}{P_{0}(R)P_{1}(\hat{R})}
\Bigr)^{\frac{1}{2}}
\Bigl(
1-
\frac{P_{1/2}(R)^2}{P_{0}(R)P_{1}(R)}
\Bigr)^{\frac{1}{2}} \\
\le &
1-
\frac{P_{1/2}(R)P_{1/2}(\hat{R})}{P_{0}(R)
P_{1}(R)^{\frac{1}{2}} P_{1}(\hat{R})^{\frac{1}{2}}}.
\end{align*}
Therefore,
\begin{align*}
&
P_{1}(R)
\Biggl(
\frac{P_{1/2}(R)P_{1/2}(\hat{R})}{P_{0}(R)
P_{1}(R)^{\frac{1}{2}} P_{1}(\hat{R})^{\frac{1}{2}}} \\
& \hspace{18ex} 
 +
\Bigl(
1-\frac{P_{1/2}(\hat{R})^2}{P_{0}(R)P_{1}(\hat{R})}
\Bigr)^{\frac{1}{2}}
\Bigl(
1-\frac{P_{1/2}(R)^2}{P_{0}(R)P_{1}(R)}
\Bigr)^{\frac{1}{2}}
\Biggr)^2 \\
\le &
P_{1}(R)
\Biggl(
\frac{P_{1/2}(R)P_{1/2}(\hat{R})}{P_{0}(R)
P_{1}(R)^{\frac{1}{2}} P_{1}(\hat{R})^{\frac{1}{2}}} 
%& \hspace{34ex} 
+1-
\frac{P_{1/2}(R)P_{1/2}(\hat{R})}{P_{0}(R)
P_{1}(R)^{\frac{1}{2}} P_{1}(\hat{R})^{\frac{1}{2}}}
\Biggr)^2 \\
=&
P_{1}(R).
\end{align*}
Then, we obtain the second inequality of \eqref{11-07-22b}.

To show the first inequality of \eqref{11-07-22b},
we employ the notation given in Appendix \ref{A2},
and choose the integers $l:= | S(R)| =P_{0}(R)$ and $\hat{l}:= | S(\hat{R})| =P_{0}(\hat{R})$
in the same way as the proof of Lemma \ref{Th1B}. 
So, the condition $R<\hat{R}$ implies that $\hat{l} \le l $.
Hence, we have $\|u_{\hat{l}}\|^2=u_l\cdot u_{\hat{l}} $.
We apply Proposition \ref{Lemma 9} to the case when 
$y=u_l$, $z=v_l$, and $\epsilon= \frac{u_{\hat{l}}\cdot v_{\hat{l}}}{\|u_{\hat{l}}\|}$.
Then, we find that $x=\frac{u_{\hat{l}}}{\|u_{\hat{l}}\|}$ satisfies the condition in $M(u_l,v_l,\epsilon)$ given in \eqref{H5}.

Now, we show that
\begin{align}
\|u_{\hat{l}}\|^2
\le 
\left(\frac{
u_l \cdot v_l \epsilon+\sqrt{
(\|v_l\|^2 -\epsilon^2)
(\|v_l\|^2 \|u_l\|^2- (u_l \cdot v_l)^2)}
}{\|v_l\|^2}\right)^2
\Label{11-07-20aT}.
\end{align}
When $l=\hat{l}$, the RHS of \eqref{11-07-20aT} equals 
$\|u_{\hat{l}}\|^2$.
So, we show \eqref{11-07-20aT} when $l>\hat{l}$ as follows.
In this case, Lemma \ref{L-2-18B} implies that 
$\frac{u_l \cdot v_l}{\|u_l\|}
> \frac{u_{\hat{l}} \cdot v_{\hat{l}}}{\|u_{\hat{l}}\|}
=\epsilon$.
Proposition \ref{Lemma 9} with Case D3 guarantees that
\begin{align}
\|u_{\hat{l}}\|
=u_l
\cdot \frac{u_{\hat{l}}}{\|u_{\hat{l}}\|}
\le 
\frac{
u_l \cdot v_l \epsilon+\sqrt{
(\|v_l\|^2 -\epsilon^2)
(\|v_l\|^2 \|u_l\|^2- (u_l \cdot v_l)^2)}
}{\|v_l\|^2},
\end{align}
which implies \eqref{11-07-20aT}.

Therefore,
\begin{align}
& P_{1}(\hat{R})=
\|u_{\hat{l}}\|^2
\le 
\left(\frac{
u_l \cdot v_l \epsilon+\sqrt{
(\|v_l\|^2 -\epsilon^2)
(\|v_l\|^2 \|u_l\|^2- (u_l \cdot v_l)^2)}
}{\|v_l\|^2}\right)^2
\nonumber \\
\stackrel{(a)}{=} &
P_{1}({R})
\Biggl(
\frac{P_{1/2}(R)P_{1/2}(\hat{R})}{P_{1}(\hat{R})^{\frac{1}{2}}P_{1}(R)^{\frac{1}{2}} P_{0}(R)}
+
\Bigl(
1-\frac{P_{1/2}(\hat{R})^2}{P_{1}(\hat{R})P_{0}(R)}
\Bigr)^{\frac{1}{2}}
\Bigl(
1-\frac{P_{1/2}(R)^2}{P_{1}(R) P_{0}(R)}
\Bigr)^{\frac{1}{2}}
\Biggr)^2
=
1- a(R,\hat{R}) ,
\Label{11-07-20a}
\end{align}
where $(a)$ follows from \eqref{11-07-26}.
Then, we obtain the first inequality of \eqref{11-07-22b}.
\end{proof}

\subsubsection{$n$-fold i.i.d. setting}
In this subsection, we rewrite the results in the previous subsection in 
$n$-fold i.i.d. setting.
In this setting,
The null hypothesis is given as the pure state $|\varphi^{\otimes n}\rangle$ in the $n$-tensor product system 
$(\mathbb{C}^d)^{\otimes n}$.
To give the alternative hypothesis, we introduce a notation.
In the quantum system $(\mathbb{C}^d)^{\otimes n}$,
the basis $\ket{i_1}\otimes \cdots \otimes \ket{i_n} $
is simplified to $\ket{J}$ by using $J \in \{1, \ldots, d\}^n$.
Hence, 
the quantum system $(\mathbb{C}^d)^{\otimes n}$ is spanned by
$\{\ket{J} \}_{J \in {\cal D}^n}=\{ 
\ket{i_1}\otimes \cdots \otimes \ket{i_n} \}_{i_1,\dots, i_n}$,
where ${\cal D}^n\stackrel{\rm def}{=}\{1, \ldots, d\}^n$.
Then, the alternative hypothesis ${\cal S}_{n,0}$
is the set of states
$\left \{ \ket{\phi_{L}^n} \right \}_{L \in \mathbb{Z}^{d^n}_2}$,
where $\ket{\phi_{L}^{n}} \in 
 \left(\mathbb{C}^{d}\right)^{\otimes n}$ is defined as 
\begin{equation}\Label{eq def phi l n}
 \ket{\phi_{L}^n}\stackrel{\rm def}{=}
\sqrt{\frac{1}{d^n}} \sum _{J \in {\cal D}^n}
  (-1)^{L_J} \ket{J}, \quad L \in \mathbb{Z}_2^{d^n},
\end{equation}
where $L_J \in \mathbb{Z}_2$ is the $J$th entry of $L\in \mathbb{Z}_2^{d^n}$.
That is, an element of the alternative hypothesis is characterized by
an element of $ \mathbb{Z}_2^{d^n}$.
Hence, the cardinality of the alternative hypothesis is $2^{d^n}$,
which is double exponential with respect to the number $n$.

For a two-valued POVM $\{ S_n, I_d^n-S_n \}$ on
$\left(\mathbb{C}^{d}\right)^{\otimes n}$, 
the type-1 error $\alpha_n(S_n)$ and type-2 error $\beta_n(S_n)$ are defined as 
\begin{align}
\alpha_n(S_n) &\stackrel{\rm def}{=} 
\Tr S_n \ket{\varphi}\bra{\varphi}^{\otimes n}
\Label{eq def a n}  \\
\beta_n(S_n) &\stackrel{\rm def}{=} 
\max_{\rho \in {\cal S}_{n,0}} \Tr (I_{d}^n-S_n) \rho,
\Label{eq def b n}
\end{align}
where $I_{d}^n$ is an identity operator on
$(\mathbb{C}^{d})^{\otimes n}$.
The optimal type-2 error under the restriction on the condition that the
type-1 error is no more than $\alpha \ge0$ can be written as 
\begin{equation}\Label{eq def gamma n a}
 \beta_n ( \alpha |\varphi ) 
  \stackrel{\rm def}{=} \min_{0 \le S_n \le I_{d }^n} \{ \beta_n(S_n) |
  \alpha_n(S_n) \le \alpha\}.
\end{equation}
Similarly, we define $\alpha_n ( \beta |\varphi )$ as
\begin{equation}
 \alpha_n ( \beta |\varphi ) 
  \stackrel{\rm def}{=} \min_{0 \le S_n \le I_{d }^n} \{ \alpha_n(S_n) |
  \beta_n(S_n) \le \beta\}.
\end{equation}
In the rest of this subsection, 
we often abbreviate $ \beta_n( \alpha |\varphi)$
as $ \beta_n \left( \alpha \right )$.

Now, we define the subset $S_n(R)\stackrel{\rm def}{=}
\{ J \in {\cal D}^n
| \log p^n_{J} \ge n R\} $ of ${\cal D}^n$,
where $p^n_{J}\stackrel{\rm def}{=} p_{i_1} \cdots p_{i_n} $ for $J=(i_1, \ldots, i_n) $.
We employ the following notations:
\begin{align*}
P_{n,s}(R) \stackrel{\rm def}{=} & \sum_{J \in S_n(R)} (p^n_{J})^s 
\hbox{ for }s=0,1/2,1, \\
\epsilon_n(R)\stackrel{\rm def}{=} &
\Big(\frac{P_{n,1/2}(R)^2}{d^n P_{n,1}(R)}\Big)^{\frac{1}{2}}.
\end{align*}
When
\begin{align}
1\ge \frac{P_{n,1/2}(\hat{R})^2}{P_{n,0}(R) P_{n,1}(\hat{R})},
\Label{C2-1A}
\end{align}
we define
\begin{align}
%P_{n,1}^c(R) &\stackrel{\rm def}{=} 1- P_{n,1}(R) \\
&a_n(R,\hat{R}) \nonumber \\
\stackrel{\rm def}{=}&1-
P_{n,1}({R})
\Biggl(
\frac{P_{n,1/2}(R)P_{n,1/2}(\hat{R})}{P_{n,1}(\hat{R})^{\frac{1}{2}}P_{n,1}(R)^{\frac{1}{2}} P_{n,0}(R)}
+
\Bigl(
1-\frac{P_{n,1/2}(\hat{R})^2}{P_{n,1}(\hat{R})P_{n,0}(R)}
\Bigr)^{\frac{1}{2}}
\Bigl(
1-\frac{P_{n,1/2}(R)^2}{P_{n,1}(R) P_{n,0}(R)}
\Bigr)^{\frac{1}{2}}
\Biggr)^2
\Label{11-05-7}.
\end{align}

Then, Lemmas \ref{Th1B} and \ref{11-06-6B}
are rewritten as follows.
\begin{lemma}\Label{Th1}
We fix $\hat{R} \in \{ \frac{1}{n}\log p_J^n\}_{J \in {\cal D}^n}$.
Then, we have the following items.
\begin{description}
\item[(1)]
When a real number $R(\in \{ \frac{1}{n}\log p_J^n\}_{J \in {\cal D}^n})< \hat{R}$
satisfies %\eqref{C2-1} and
%Given $\hat{R}$, we choose the minimum $R$ satisfying
\begin{align}
\frac{P_{n,0}(R)}{P_{n,1/2}(R)}
e^{\frac{n R}{2}} 
+
\frac{\Bigl(\frac{P_{n,1}(R)P_{n,0}(R)}{P_{n,1/2}(R)^2}-1\Bigr)^{\frac{1}{2}}
}
{\Bigl(\frac{P_{n,1}(\hat{R})P_{n,0}(R)}{P_{n,1/2}(\hat{R})^2}
-1\Bigr)^{\frac{1}{2}}
}
\ge 1,
\Label{11-4-1} 
\end{align}
we have
\begin{align}
\alpha_n\left(
\left.
\frac{P_{n,1/2}(\hat{R})^2}{d^n P_{n,1}(\hat{R})}
\right|\varphi\right)
\le
a_n(R,\hat{R}) 
\Label{4-6-1b} .
\end{align}
\item[(2)]
We assume that there exists an element $R_\alpha$ in $\{ \frac{1}{n}\log p_J^n\}_{J \in {\cal D}^n}$
satisfying the inequality \eqref{11-4-1} and $ R_\alpha< \hat{R}$.
We denote all of distinct elements of $\{ \frac{1}{n}\log p_J^n\}_{J \in {\cal D}^n}$ 
by $\tilde{R}_k < \tilde{R}_{k-1}< \ldots < \tilde{R}_2 < \tilde{R}_{1}$.
We also assume that an element $R_\beta$ in $\{ \frac{1}{n}\log p_J^n\}_{J \in {\cal D}^n}$ satisfying the following condition;
Any element $\tilde{R}_j (\le R_{\beta})\in \{ \frac{1}{n}\log p_J^n\}_{J \in {\cal D}^n}$ satisfies
\begin{align}
\frac{P_{n,0}(\tilde{R}_{j+1})}{P_{n,1/2}(\tilde{R}_{j})}
e^{\frac{ \tilde{R}_{j}}{2}} +
\frac{\Bigl(\frac{P_{n,1}(\tilde{R}_{j+1})
P_{n,0}(\tilde{R}_{j+1})}{P_{n,1/2}(\tilde{R}_{j})^2}-1\Bigr)^{\frac{1}{2}}
}
{\Bigl(\frac{P_{n,1}(\hat{R})P_{n,0}(\tilde{R}_{j})}{P_{n,1/2}(\hat{R})^2}
-1\Bigr)^{\frac{1}{2}}
}
< 1.
\Label{11-4-1-P} 
\end{align}
Then, the real numbers $\hat{R}$ and $R_\beta$ satisfy the inequality
\begin{align}
\alpha_n\left(
\left.
\frac{P_{n,1/2}(\hat{R})^2}{d^n P_{n,1}(\hat{R})}
\right|\varphi\right)
\ge
a_n(R_\beta,\hat{R}) 
\Label{4-6-1b-P} .
\end{align}
\hfill $\square$
\end{description}
\end{lemma}

\begin{lemma}\Label{11-06-6}
When $R<\hat{R}$,
the number $a_n(R,\hat{R}) $ is evaluated as
\begin{align}
P_{n,1}^c(R)
\le
a_n(R,\hat{R}) 
\le
P_{n,1}^c(\hat{R}),
\end{align}
where
$P_{n,1}^c(R)\stackrel{\rm def}{=} 1-P_{n,1}(R)$.
\hfill $\square$\end{lemma}

\subsubsection{Constant constraint for type-1 error}\Label{s412}
Under a constant constraint for type-1 error, we have the following theorem.
\begin{theorem}\Label{11-07-2T}
We have 
\begin{align}
\log \beta_n(\epsilon|\varphi) 
=& n (H(p) -\log d) 
- \sqrt{n} \sqrt{V(p)} \Phi^{-1}(\epsilon)
-\log n +O(1),\Label{11-07-2}
\end{align}
where 
$H(p)\stackrel{\rm def}{=} -\sum_i p_i \log p_i$
and $V(p)\stackrel{\rm def}{=} \sum_i p_i (H(p)+\log p_i)^2$.
\hfill $\square$\end{theorem}

For a preparation of the proof of Theorem \ref{11-07-2T}, 
we introduce several notations.
First, we choose $A_\epsilon \stackrel{\rm def}{=} \sqrt{V(p)} \Phi^{-1}(\epsilon)$. 
Remember that $\Phi$ is the cumulative distribution function of the standard Gaussian distribution.
We fix $d_S$ to be the lattice span of the random variable $-\log p_I$
when the index $I$ is subject to the distribution $p$.
Hence, the set $\{ \frac{1}{n}\log p_J^n\}_{J \in {\cal D}^n}$ has the lattice structure with the span $\frac{d_S}{n}$. 
For the precise definition of $d_S$, see Appendix \ref{A1}.
Then, we define the functions $g_1$, $g_2$, and $g_3$ as
\begin{align}
g_1(d_S)
&:=
\left\{
\begin{array}{ll}
-\log 2
& \hbox{ if } d_S=0 \\
\log \frac{1-e^{-\frac{1}{2}d_S}}{1-e^{-d_S}}
& \hbox{ if } d_S>0,
\end{array}
\right. 
\\
g_2(d_S)
&:=
\left\{
\begin{array}{ll}
-\frac{1}{2}\log 2 \pi+\frac{1}{2V(p)}+2\log 2
& \hbox{ if } d_S=0 \\
-\frac{1}{2}\log 2 \pi+\frac{1}{2V(p)}+
\log \frac{1-e^{-d_S}}{(1-e^{-\frac{1}{2}d_S})^2}
& \hbox{ if } d_S>0,
\end{array}
\right. \\
g_3(d_S)
&:=
\left\{
\begin{array}{ll}
-\frac{1}{2}\log 2 \pi +\log 2
+ \frac{1}{2 V(p)}
& \hbox{ if } d_S=0 \\
-\frac{1}{2}\log 2 \pi + \frac{1}{2 V(p)}
+ \log \frac{d_S}{1-e^{-\frac{1}{2}d_S}} 
& \hbox{ if } d_S>0.
\end{array}
\right. 
\end{align}

Then, we have the following lemma, which will be shown after the proof of Theorem \ref{11-07-2T}.
\begin{lemma}\Label{L2-1}
For real numbers $B_i$ with $i=1,2,3,4,5$,
we define ${R}_{n,i} \stackrel{\rm def}{=} 
-H(p)+\frac{A_\epsilon}{\sqrt{n}}+\frac{B_i}{n}$
with $i=1,2,3,4,5$.
\begin{align}
\log \frac{P_{n,1/2}(R_{n,1})P_{n,1/2}(R_{n,2})}
{P_{n,1}(R_{n,3})^{\frac{1}{2}} P_{n,1}(R_{n,4})^{\frac{1}{2}}
P_{n,0}(R_{n,5})}
=
-\frac{1}{2} \log n 
+B_5-\frac{B_3+B_4}{2} +g_2(d_S)
-\log (1-\epsilon) +o(1).
\Label{1-5-2X}
\end{align}
The convergences of the differences between the LHSs and RHSs 
are compact uniform for $B_i$.

Assume that
$\hat{R}_{n}:= -H(p)+\frac{A_\epsilon}{\sqrt{n}}+\frac{\hat{B}_n}{n}$,
$R_n=-H(p)+\frac{A_\epsilon}{\sqrt{n}}+\frac{B_n}{n}$,
and $R_n'=-H(p)+\frac{A_\epsilon}{\sqrt{n}}+\frac{B_n'}{n}$.

When $B_n$ and $\hat{B}_{n}$ are bounded, and $B_n-B_n'$ converges,
we have
\begin{align}
&
\log \epsilon_n(R_{n})
=
\log \frac{P_{n,1/2}(R_{n})^2}{d^n P_{n,1}(R_{n})} \nonumber \\
=&n (H(p)-\log d )
-\sqrt{n} A_\epsilon -\log n 
-B_n-\frac{A_\epsilon^2}{V(p)} +2 g_3(d_S)
-\log (1-\epsilon) +o(1), \Label{11-07-6B} \\
&\lim_{n \to \infty}\log
\frac{P_{n,0}(R_{n}')}{P_{n,1/2}(R_{n})}
e^{\frac{n R_{n}}{2}} 
= g_1(0)+\lim_{n \to \infty} \frac{B_n-B_n'}{2} , \Label{9-7-2b} \\
& a_n(R_{n},R_{n}') =\epsilon+o(1) ,\Label{11-07-5B} \\
&\frac{\Bigl(\frac{P_{n,1}(R_{n}')P_{n,0}(R_{n}')}{P_{n,1/2}(R_{n})^2}-1\Bigr)^{\frac{1}{2}}
}
{\Bigl(\frac{P_{n,1}(\hat{R}_n)P_{n,0}(R_{n})}{P_{n,1/2}(\hat{R}_n)^2}
-1\Bigr)^{\frac{1}{2}}
}
%\cong
%\Bigl(
%\frac{P_{n,1}(R_{n}')P_{n,0}(R_{n}')P_{n,1/2}(\hat{R}_n)^2}{P_{n,1/2}(R_{n})^2P_{n,1}(\hat{R}_n)P_{n,0}(R_{n})}
%\Bigr)^{\frac{1}{2}}
=
e^{\frac{2B_n -B_n '-\hat{B}_{n}}{2}}+o(1) \Label{9-7-3b}.
\end{align}

When $B_n \to - \infty$, $\hat{B}_{n}$ is bounded, and $B_n-B_n'$ converges,
\begin{align}
&\lim_{n \to \infty}\log
\frac{P_{n,0}(R_{n}')}{P_{n,1/2}(R_{n})}
e^{\frac{n R_{n}}{2}} 
\le g_1(0)
+\lim_{n \to \infty} \frac{B_n-B_n'}{2}
\Label{9-7-2}\\
&\lim_{n \to \infty}
\frac{\Bigl(\frac{P_{n,1}(R_{n}')P_{n,0}(R_{n}')}{P_{n,1/2}(R_{n})^2}-1\Bigr)^{\frac{1}{2}}
}
{\Bigl(\frac{P_{n,1}(\hat{R}_n)P_{n,0}(R_{n})}{P_{n,1/2}(\hat{R}_n)^2}
-1\Bigr)^{\frac{1}{2}}
}
%\cong
%\Bigl(\frac{P_{n,1}(R_{n}')P_{n,0}(R_{n}')P_{n,1/2}(\hat{R}_n)^2}{P_{n,1/2}(R_{n})^2P_{n,1}(\hat{R}_n)P_{n,0}(R_{n})}\Bigr)^{\frac{1}{2}}
=0\Label{9-7-4} .
\end{align}
\hfill $\square$\end{lemma}

\begin{proofof}{Theorem \ref{11-07-2T}}\par
\PF{Non-lattice case}
\noindent{\it Step 1:)}
For simplicity, we first consider the case when $d_S=0$, i.e., the non-lattice case.
We fix $\hat{B}$.
Due to the non-lattice property (Lemma \ref{L9-20}), we can choose
we can choose
$\hat{B}_n$ such that
$\lim_{n\to \infty}\hat{B}_n = \hat{B}$
and $\hat{R}_{n}:= -H(p)+\frac{A_\epsilon}{\sqrt{n}}+\frac{\hat{B}_n}{n}
 \in \{ \frac{1}{n}\log p_J^n\}_{J \in {\cal D}^n}$.
Then, we will show 
\begin{align}
\lim_{n \to \infty}
\alpha_n\left(
\left.
\frac{P_{n,1/2}(\hat{R}_n)^2}{d^n P_{n,1}(\hat{R}_n)}
\right|\varphi\right)
=\epsilon.\Label{7-6-eq}
\end{align}
Since $\frac{P_{n,1/2}(\hat{R}_n)^2}{d^n P_{n,1}(\hat{R}_n)}$ is characterized 
by \eqref{11-07-6B},
\eqref{7-6-eq} implies the desired argument when $d_S=0$.
Now, we outline the derivation of \eqref{7-6-eq}. 
To show \eqref{7-6-eq}, we find upper and lower bounds of \eqref{7-6-eq} whose limit is $\epsilon$.
For this purpose,
in Step 2:), we find its upper bound by using Item (1) of Lemma \ref{Th1},
and 
in Step 3:), we find its lower bound by using Item (2) of Lemma \ref{Th1}.
In Step 4:),
calculating both bounds, we show \eqref{7-6-eq}.

\noindent{\it Step 2:)}
Assume that $\lim_{n\to \infty}B_n$ converges.
We choose $R_{n}:=-H(p)+\frac{A_\epsilon}{\sqrt{n}}+\frac{B_n}{n}$.
Using 
\eqref{9-7-3b} and \eqref{9-7-2b},
we have
\begin{align}
\frac{P_{n,0}(R_{n})}{P_{n,1/2}(R_{n})}
e^{\frac{n R_{n}}{2}} 
+
\frac{\Bigl(\frac{P_{n,1}(R_{n})P_{n,0}(R_{n})}{P_{n,1/2}(R_{n})^2}-1\Bigr)^{\frac{1}{2}}
}
{\Bigl(\frac{P_{n,1}(\hat{R}_n)P_{n,0}(R_{n})}{P_{n,1/2}(\hat{R}_n)^2}
-1\Bigr)^{\frac{1}{2}}
}
=
e^{g_1(0)}+e^{\frac{B_n- \hat{B}}{2}}+o(1).\Label{9-7-6}
\end{align}
Given $\delta>0$,
due to the non-lattice property (See Lemma \ref{L9-20} in Appendix \ref{A1}),
we can chose $B_{\alpha,n}$ such that
$R_{\alpha,n}:=-H(p)+\frac{A_\epsilon}{\sqrt{n}}+\frac{B_{\alpha,n}}{n}$
belongs to 
$\{ \frac{1}{n}\log p_J^n| \frac{1}{n}\log p_J^n < R_n \}_{J \in {\cal D}^n}$ and
\begin{align}
\lim_{n\to \infty}
B_{\alpha,n}= \hat{B}+2 \log (1-e^{g_1(0)})+\delta
\Label{2-22-D1}.
\end{align}

Then, 
\begin{align}
e^{g_1(0)}+ e^{\frac{\lim_{n\to \infty} B_{\alpha,n}- \hat{B}}{2}}
=e^{g_1(0)}+ (1-e^{g_1(0)})e^{\delta} >1.
\end{align}
With sufficiently large $n$,
$R_{\alpha,n}$
satisfies 
\begin{align}
\frac{P_{n,0}(R_{\alpha,n})}{P_{n,1/2}(R_{\alpha,n})}
e^{\frac{n R_{\alpha,n}}{2}} 
+
\frac{\Bigl(\frac{P_{n,1}(R_{\alpha,n})P_{n,0}(R_{\alpha,n})}{P_{n,1/2}(R_{\alpha,n})^2}-1\Bigr)^{\frac{1}{2}}
}
{\Bigl(\frac{P_{n,1}(\hat{R}_n)P_{n,0}(R_{\alpha,n})}{P_{n,1/2}(\hat{R}_n)^2}
-1\Bigr)^{\frac{1}{2}}
}
& >1\Label{8-27-9Y} \\
R_{\alpha,n} &< \hat{R}_{n}\Label{8-27-9X}.
\end{align}
Thus, we can apply Item (1) of Lemma \ref{Th1} to this case.
Hence, we obtain
\begin{align}
\alpha_n\left(
\left.
\frac{P_{n,1/2}(\hat{R}_n)^2}{d^n P_{n,1}(\hat{R}_n)}
\right|\varphi\right)
\le
a_n(R_{\alpha,n},\hat{R}_n) \Label{2-22-A}.
\end{align}

\noindent{\it Step 3:)}
We choose $B_{n}$ as $R_n=-H(p)+\frac{A_\epsilon}{\sqrt{n}}+\frac{B_n}{n}$.
Then, we choose $R_n'$ as the maximum element in 
$\{ \frac{1}{n}\log p_J^n| \frac{1}{n}\log p_J^n < R_n \}_{J \in {\cal D}^n}$.
So, the non-lattice property (See Lemma \ref{L9-20} in Appendix \ref{A1}) guarantees $\lim_{n \to \infty}n(R_n-R_n')=0$. 
When $B_n \to -\infty$,
\eqref{9-7-4} and \eqref{9-7-2} imply that
\begin{align}
\lim_{n \to \infty}
\frac{P_{n,0}(R_{n}')}{P_{n,1/2}(R_{n})}
e^{\frac{n R_{n}}{2}} 
+
\frac{\Bigl(\frac{P_{n,1}(R_{n}')P_{n,0}(R_{n}')}{P_{n,1/2}(R_{n})^2}-1\Bigr)^{\frac{1}{2}}
}
{\Bigl(\frac{P_{n,1}(\hat{R}_n)P_{n,0}(R_{n})}{P_{n,1/2}(\hat{R}_n)^2}
-1\Bigr)^{\frac{1}{2}}
}
\le e^{g_1(0)} <1\Label{9-7-5b}.
\end{align}
%Now, we consider the case when $\limsup_{n \to \infty}B_{n}<\infty$ and $\liminf_{n \to \infty}B_{n}>-\infty$.
When $ B_{n}$ is bounded, the combination of \eqref{9-7-3b}  and \eqref{9-7-2b} implies that
\begin{align}
\frac{P_{n,0}(R_{n}')}{P_{n,1/2}(R_{n})}
e^{\frac{n R_{n}}{2}} 
+
\frac{\Bigl(\frac{P_{n,1}(R_{n}')P_{n,0}(R_{n}')}{P_{n,1/2}(R_{n})^2}-1\Bigr)^{\frac{1}{2}}
}
{\Bigl(\frac{P_{n,1}(\hat{R}_n)P_{n,0}(R_{n})}{P_{n,1/2}(\hat{R}_n)^2}
-1\Bigr)^{\frac{1}{2}}
}
=e^{g_1(0)}+ e^{\frac{B_{n}- \hat{B}}{2}}+o(1). \Label{9-6-2}
\end{align}
%Given $\deta>0$, 
Then, due to the non-lattice property,
we can chose $B_{\beta,n}$ such that 
$R_{\beta,n}:=-H(p)+\frac{A_\epsilon}{\sqrt{n}}+\frac{B_{\beta,n}}{n}$ 
belongs to 
$\{ \frac{1}{n}\log p_J^n| \frac{1}{n}\log p_J^n < R_n \}_{J \in {\cal D}^n}$ and
\begin{align}
\lim_{n \to \infty}B_{\beta,n}= \hat{B}+2 \log (1-e^{g_1(0)})-\delta \Label{2-22-D1B}.
\end{align}
So, when $B_n \le B_{\beta,n}$,
with sufficiently large $n$, we have
\begin{align}
e^{g_1(0)}+ e^{\frac{B_{n}- \hat{B}}{2}}
\le e^{g_1(0)}+ (1-e^{g_1(0)})e^{-\delta} <1.
\end{align}
In this case,
with sufficiently large $n$,
we have
\begin{align}
\frac{P_{n,0}(R_{n}')}{P_{n,1/2}(R_{n})}
e^{\frac{n R_{n}}{2}} 
+
\frac{\Bigl(\frac{P_{n,1}(R_{n}')P_{n,0}(R_{n}')}{P_{n,1/2}(R_{n})^2}-1\Bigr)^{\frac{1}{2}}
}
{\Bigl(\frac{P_{n,1}(\hat{R}_n)P_{n,0}(R_{n})}{P_{n,1/2}(\hat{R}_n)^2}
-1\Bigr)^{\frac{1}{2}}
}
<1 \Label{9-6-2T}.
\end{align}
Thus, $R_{\beta,n}$ satisfies the conditions for $R_\beta$ in Item (2) of Lemma \ref{Th1} with $\hat{R}=\hat{R}_n$.
Due to \eqref{8-27-9Y} and \eqref{8-27-9X}, we can apply Item (2) of Lemma \ref{Th1} to the case with $\hat{R}=\hat{R}_n$, $R_\alpha=R_{\alpha,n}$, and $R_\beta=R_{\beta,n}$.
Hence, we obtain
\begin{align}
\alpha_n\left(
\left.
\frac{P_{n,1/2}(\hat{R}_n)^2}{d^n P_{n,1}(\hat{R}_n)}
\right|\varphi\right)
\ge
a_n(R_{\beta,n},\hat{R}_n) \Label{2-22-B}.
\end{align}

\noindent{\it Step 4:)}
\eqref{2-22-D1} and \eqref{2-22-D1B} show that 
the sequences $B_{\alpha,n}$ and $B_{\beta,n}$ converge to constants as well as $\hat{B}_n$.
Thus, \eqref{11-07-5B} implies that
\begin{align}
\lim_{n \to \infty}
a_n(R_{\alpha,n},\hat{R}_n)
=
\lim_{n \to \infty}
a_n(R_{\beta,n},\hat{R}_n)
=\epsilon.\Label{8-27A5}
\end{align}
Combining \eqref{2-22-A} and \eqref{2-22-B}, we obtain \eqref{7-6-eq}.

\PF{Lattice case}
Next, we proceed to the lattice case with $d_S>0$.
The different points from the non-lattice case 
are the following.
Firstly, we cannot necessarily choose $\hat{B}_n$ such that 
the limit $\lim_{n \to \infty}\hat{B}_n$ exists.
However, we can choose $\hat{B}_n$ such that 
$\hat{B}_n$ is bounded, i.e., $\hat{B}_n$ behaves within an interval with width $d_S$.
The above proof works even with such a bounded case.
The second point is the relation $\lim_{n \to \infty}n(R_n-R_n')=d_S>0$, which appears only in Steps 2:) and 3:).
In these steps, we need to replace $g_1(0)$ by $g_1(d_S)$.
In Step 2:), the relations \eqref{9-7-6} and \eqref{2-22-D1} are replaced by
\begin{align}
&\frac{P_{n,0}(R_{n})}{P_{n,1/2}(R_{n})}
e^{\frac{n R_{n}}{2}} 
+
\frac{\Bigl(\frac{P_{n,1}(R_{n})P_{n,0}(R_{n})}{P_{n,1/2}(R_{n})^2}-1\Bigr)^{\frac{1}{2}}
}
{\Bigl(\frac{P_{n,1}(\hat{R}_n)P_{n,0}(R_{n})}{P_{n,1/2}(\hat{R}_n)^2}
-1\Bigr)^{\frac{1}{2}}
}
=
e^{g_1(d_S)+\frac{d_S}{2}}+e^{\frac{B_n- \hat{B}}{2}+d_S}+o(1).
\Label{9-7-6B}
\\
&B_{\alpha,n}:= \hat{B}-2 d_S+2 \log (1-e^{g_1(d_S)+\frac{d_S}{2}})
+\delta
\Label{2-22-D1-B}.
\end{align}
In Step 3:), the relations \eqref{9-7-5b}, \eqref{9-6-2}, and \eqref{2-22-D1B}
are replaced by
\begin{align}
&\lim_{n \to \infty}
\frac{P_{n,0}(R_{n}')}{P_{n,1/2}(R_{n})}
e^{\frac{n R_{n}}{2}} 
+
\frac{\Bigl(\frac{P_{n,1}(R_{n}')P_{n,0}(R_{n}')}{P_{n,1/2}(R_{n})^2}-1\Bigr)^{\frac{1}{2}}
}
{\Bigl(\frac{P_{n,1}(\hat{R}_n)P_{n,0}(R_{n})}{P_{n,1/2}(\hat{R}_n)^2}
-1\Bigr)^{\frac{1}{2}}
}
\le e^{g_1(d_S)+d_S} <1\Label{9-7-5b-2},
\\
&\frac{P_{n,0}(R_{n}')}{P_{n,1/2}(R_{n})}
e^{\frac{n R_{n}}{2}} 
+
\frac{\Bigl(\frac{P_{n,1}(R_{n}')P_{n,0}(R_{n}')}{P_{n,1/2}(R_{n})^2}-1\Bigr)^{\frac{1}{2}}
}
{\Bigl(\frac{P_{n,1}(\hat{R}_n)P_{n,0}(R_{n})}{P_{n,1/2}(\hat{R}_n)^2}
-1\Bigr)^{\frac{1}{2}}
}
=e^{g_1(d_S)+d_S}+ e^{\frac{B_{n}- \hat{B}}{2}+d_S}+o(1), \Label{9-6-2-2}
\\
&B_{\beta,n}:= \hat{B}-2d_S+2 \log (1-e^{g_1(d_S)+d_S})-\delta \Label{2-22-D1B-2}.
\end{align}
Hence, the sequence $B_{\beta,n}$ is bounded as well as $\hat{B}_n$ and $B_{\alpha,n}$.
Thus, we obtain \eqref{8-27A5}.
Combining \eqref{2-22-A} and \eqref{2-22-B}, we obtain \eqref{7-6-eq} even in the lattice case
$d_S>0$.
\end{proofof}

\begin{proofof}{Lemma \ref{L2-1}}\par
\PF{Proofs of \eqref{1-5-2X}, \eqref{11-07-6B}, and \eqref{9-7-2b}}
%Now, we show the relations \eqref{11-07-6B}, \eqref{11-4-1fD}, \eqref{1-5-2X}, \eqref{11-4-1fB}, and \eqref{11-07-5B}.
%For this purpose, 
We show the desired relations by applying Proposition \ref{11-4-4} in Appendix \ref{A1}. 
When the distribution $p$ in Proposition \ref{11-4-4} 
is the measure $\{p_{i}\}_{i}$ and $X$ is $ \log p_{i}$, 
we denote the functions given in Proposition \ref{11-4-4} 
by adding superscript $1$, 
like $\chi^{1}_0$, $\chi^{1}_1$, $\tau^{1}(s)$, etc.
Similarly,
when the distribution $p$ in Proposition \ref{11-4-4} 
is the measure $\{\sqrt{p_{i}}\}_{i}$ (the counting measure) and
$X$ is $ \log p_{i}$. 
We denote them by adding superscript $\frac{1}{2}$ ($0$), 
like 
$\chi^{\frac{1}{2}}_0$, $\chi^{\frac{1}{2}}_1$, $\tau^{\frac{1}{2}}(s)$, 
($\chi^{0}_0$, $\chi^{0}_1$, $\tau^{0}(s)$) etc.
We also employ the function
$\psi_p(s) \stackrel{\rm def}{=} \log \sum_{i}p_i^{1+s}$.
Then, we have
\begin{align}
\tau^{t}(s) &=\psi_p(s+1-t) ,\quad
\eta^{t}(R) ={\psi_p'}^{-1}(R) +1-t
\end{align}
for $t=0,\frac{1}{2},1$.
Hence,
\begin{align}
{\eta^{t}}'(R) &=\frac{1}{\psi_p''({\psi_p'}^{-1}(R))}\\
\chi_{0}^{t}(R)
&=  -R ({\psi_p'}^{-1}(R) +1 -t) + \psi_p({\psi_p'}^{-1}(R)) \Label{9-17-eq1}\\
{\chi_{0}^{t}}'(R)
&=-{\psi_p'}^{-1}(R) -1+t\\
{\chi_{0}^{t}}''(R) &={\chi_{0}^{1/2}}''(R)
=-\frac{d{\psi_p'}^{-1}}{d R}(R) 
=-\frac{1}{\psi_p''({\psi_p'}^{-1}(R))} \\
\chi_{1}^{t}(R)
&=
\left\{
\begin{array}{ll}
-\frac{1}{2}\log 2 \pi -\log ({\psi_p'}^{-1}(R)+1-t)
+ \frac{1}{2 \psi_p''({\psi_p'}^{-1}(R))}
& \hbox{ if } d_S=0 \\
-\frac{1}{2}\log 2 \pi + \frac{1}{2 \psi_p''({\psi_p'}^{-1}(R))}
+ \log \frac{d_S}{1-e^{-d_S ({\psi_p'}^{-1}(R)+1-t)}} 
& \hbox{ if } d_S>0.
\end{array}
\right.
\end{align}

Generally, Proposition \ref{11-4-4} implies that
\begin{align}
%\log \frac{P_{n,0}(R)}{P_{n,1/2}(R)}e^{-\frac{nR}{2}}=& h_1(R) \Label{11-06-2}\\
\log P_{n,1/2}(R)
=& 
n (  -R ({\psi_p'}^{-1}(R) +\frac{1}{2} ) + \psi_p ({\psi_p'}^{-1}(R)))
-\frac{1}{2}\log n +\chi_1^{\frac{1}{2}}(R)
+o(1)
\Label{1-32-1}
 \\
\log P_{n,0}(R)
=& 
n (  -R ({\psi_p'}^{-1}(R) +1 ) + \psi_p ({\psi_p'}^{-1}(R)))
-\frac{1}{2}\log n +\chi_1^{0}(R)
+o(1)
\Label{1-32-2}.
\end{align}
Using $\Delta_n \stackrel{\rm def}{=}  \sqrt{n} {\psi_p'}^{-1}(-H(p)+\frac{A_\epsilon}{\sqrt{n}}+\frac{B}{n})
= \frac{A_\epsilon}{V(p)}+o(1)$,
for any real number $B$, 
we have
\begin{align*}
\psi_p({\psi_p'}^{-1}(-H(p)+\frac{A_\epsilon}{\sqrt{n}}+\frac{B}{n}))
&= 
\psi_p'(0)\frac{\Delta_n}{\sqrt{n}}
+\psi_p''(0)\frac{\Delta_n^2}{2n}
+ o(\frac{1}{n}) \\
&= 
-H(p)\frac{\Delta_n}{\sqrt{n}}
+\frac{A_\epsilon^2}{2V(p)n}
+ o(\frac{1}{n}) \\
-(-H(p)+\frac{A_\epsilon}{\sqrt{n}}+\frac{B}{n}) {\psi_p'}^{-1}(-H(p)+\frac{A_\epsilon}{\sqrt{n}}+\frac{B}{n})
&=
H(p)\frac{\Delta_n}{\sqrt{n}}
-A_\epsilon \Delta_n \frac{1}{n}
+ o(\frac{1}{n}) \\
&=
H(p)\frac{\Delta_n}{\sqrt{n}}
-\frac{A_\epsilon^2}{V(p)n}
+ o(\frac{1}{n}) .
\end{align*}
Thus, we have
\begin{align}
-(-H(p)+\frac{A_\epsilon}{\sqrt{n}}+\frac{B}{n}) {\psi_p'}^{-1}(-H(p)+\frac{A_\epsilon}{\sqrt{n}}+\frac{B}{n})
+\psi_p({\psi_p'}^{-1}(-H(p)+\frac{A_\epsilon}{\sqrt{n}}+\frac{B}{n}))
=
-\frac{A_\epsilon^2}{2 V(p)n}+ o(\frac{1}{n}).\Label{2-02-1B}
\end{align}
Applying \eqref{2-02-1B} to \eqref{1-32-1} and \eqref{1-32-2}, 
we have
\begin{align}
\log P_{n,1/2}(-H(p)+\frac{A_\epsilon}{\sqrt{n}}+\frac{B}{n})
=& \frac{1}{2}(n H(p)-\sqrt{n} A_\epsilon-B)
-\frac{1}{2}\log n 
-\frac{A_\epsilon^2}{2 V(p)} +\chi_1^{\frac{1}{2}}(-H(p)) +o(1)
\Label{1-31-5} \\
\log P_{n,0}(-H(p)+\frac{A_\epsilon}{\sqrt{n}}+\frac{B}{n})
=& (n H(p)-\sqrt{n} A_\epsilon-B)
-\frac{1}{2}\log n 
-\frac{A_\epsilon^2}{2 V(p)} +\chi_1^{0}(-H(p)) +o(1)
\Label{1-31-6} .
\end{align}
Here, the LHS minus the RHS approach to zero, whose convergence is compact uniform for the choice of $B$.

Also, the central limit theorem yields
\begin{align}
\lim_{n \to \infty} P_{n,1}(-H(p)+\frac{A_\epsilon}{\sqrt{n}} +\frac{B}{n})
= 1-\Phi(\frac{A_\epsilon}{\sqrt{V(p)}})
%=\Phi(-\frac{A}{\sqrt{V(p)}}) 
=1-\epsilon
\Label{11-07-4}.
\end{align}
Since
\begin{align}
2 \chi_{1}^{1/2}(-H(p)) -\chi_{1}^0(-H(p))
&= g_2(d_S) \\
2 \chi_{1}^{1/2}(-H(p)) 
&=
2 g_3(d_S) \\
\chi_{1}^{0}(-H(p)) - \chi_{1}^{1/2}(-H(p)) 
&= g_1(d_S) ,
\end{align}
combining \eqref{1-31-5}, \eqref{1-31-6}, and \eqref{11-07-4}, we obtain 
\eqref{1-5-2X}, \eqref{11-07-6B}, and \eqref{9-7-2b}.
Indeed, while $B$ depends on $n$ in \eqref{11-07-6B} and \eqref{9-7-2b},
since the convergence is compact uniform for the choice of $B$,
the relations \eqref{11-07-6B} and \eqref{9-7-2b} hold.

\PF{Proof of \eqref{11-07-5B}}
Due to \eqref{1-5-2X}, we find that 
\begin{align}
\frac{P_{n,1/2}(R_{n,1})P_{n,1/2}({R}_{n,2})}
{P_{n,1}({R}_{n,2})^{\frac{1}{2}}P_{n,1}(R_{n,1})^{\frac{1}{2}} 
P_{n,0}(R_{n,1})}
&\to 0 \Label{22-E1}
\\
\frac{P_{n,1/2}({R}_{n,2})^2}{P_{n,1}({R}_{n,2})P_{n,0}(R_{n,1})}
&\to 0 \Label{22-E2}
\\
\frac{P_{n,1/2}(R_{n,1})^2}{P_{n,1}(R_{n,1}) P_{n,0}(R_{n,1})}
&\to 0. \Label{22-E3}
\end{align}
Since \eqref{11-07-4} implies
\begin{align}
P_{n,1}({R}_{n,1}) \to 1-\epsilon,
\end{align}
we obtain \eqref{11-07-5B}.
The compact uniformness of these convergences are guaranteed by the compact uniformness of the convergences in Proposition \ref{11-4-4}.

\PF{Proof of \eqref{9-7-3b}}
When $B_n$ and $\hat{B}_{n}$ are bounded, and $B_n-B_n'$ converges,
using the relation \eqref{1-5-2X}, we have
\begin{align}
\log \frac{P_{n,1}(R_{n}')P_{n,0}(R_{n}')}{P_{n,1/2}(R_{n})^2}
&= \frac{1}{2}\log n -(B_n'-B_n)
-g_2(0) -\log (1-\epsilon)+o(1), \Label{9-8-8} \\
\log \frac{P_{n,1}(\hat{R}_n)P_{n,0}(R_{n})}{P_{n,1/2}(\hat{R}_n)^2}
&= \frac{1}{2}\log n -(B_{n}-\hat{B})
-g_2(0) -\log (1-\epsilon)+o(1) .\Label{9-8-7}
\end{align}
Therefore, we obtain \eqref{9-7-3b}.

\PF{Proof of \eqref{9-7-2}}
The relations \eqref{1-32-1} and \eqref{1-32-2} show that
\begin{align}
%\lim_{n \to \infty}
\log
\frac{P_{n,0}(R_{n}')}{P_{n,1/2}(R_{n})}
e^{\frac{n R_{n}}{2}} 
&=
\frac{B_{n}'-B_n}{2}
+\chi_1^0(R_n')
-\chi_1^{\frac{1}{2}}(R_n)
+o(1) \nonumber \\
&=
\frac{B_{n}'-B_n}{2}
+\chi_1^0(R_n')
-\chi_1^{\frac{1}{2}}(R_n)
+o(1) \Label{9-9-2}.
%g_1(0)+\lim_{n \to \infty} B_n  \Label{9-7-2}.
\end{align}
When $B_n \to - \infty$ and $B_n-B_n'$ converges,
since $\chi_1^0(R)-\chi_1^{\frac{1}{2}}(R)$ is monotone increasing for $R$,
we have
\begin{align}
\chi_1^0(R_n')
-\chi_1^{\frac{1}{2}}(R_n)
& \le 
\chi_1^0(- H(p) +\frac{A_\epsilon}{\sqrt{n}})
-\chi_1^{\frac{1}{2}}(- H(p) +\frac{A_\epsilon}{\sqrt{n}}) \nonumber \\
&= \chi_1^0(- H(p) )-\chi_1^{\frac{1}{2}}(- H(p))+o(1)
=g_1(d_s)+o(1).\Label{9-9-4}
\end{align}
So, combinig \eqref{9-9-2} and \eqref{9-9-4},
we obtain \eqref{9-7-2}.

%When $B_n$ is bounded, and $B_n-B_n'$ converges,
%\eqref{11-4-1fD} implies \eqref{9-7-2b}.

\PF{Proof of \eqref{9-7-4}}
Assume that $B_n \to - \infty$, $\hat{B}_{n}$ is bounded, and $B_n-B_n'$ converges to $C$.
We fix a sufficiently large number $A>0$. 
We have $R_{n}< \hat{R}_{n}-A$ for sufficiently large $n$ because $B_n \to - \infty$.
So,
\begin{align}
P_{n,\frac{1}{2}}(R_{n}) \ge
P_{n,\frac{1}{2}}(\hat{R}_{n}) \Label{9-8-1}.
\end{align}
Since
\begin{align}
P_{n,1}(\hat{R}_n)P_{n,0}(\hat{R}_{n}) \ge P_{n,1/2}(\hat{R}_n)^2,
\Label{9-8-2}
\end{align}
with sufficiently large $n$, we have 
\begin{align}
\frac{P_{n,1}(\hat{R}_n)P_{n,0}(R_{n})}{P_{n,1/2}(\hat{R}_n)^2}
=
\frac{P_{n,1}(\hat{R}_n)P_{n,0}(\hat{R}_{n})}{P_{n,1/2}(\hat{R}_n)^2}
\cdot \frac{P_{n,0}(R_{n})}{P_{n,0}(\hat{R}_{n})}
\ge
\frac{P_{n,0}(R_{n})}{P_{n,0}(\hat{R}_{n})}
\ge
\frac{P_{n,0}(\hat{R}_{n}-A)}{P_{n,0}(\hat{R}_{n})}
\stackrel{(a)}{\to} e^{A}.
\end{align}
where $(a)$ follows from \eqref{1-31-6}.
So, 
\begin{align}
\frac{P_{n,1}(\hat{R}_n)P_{n,0}(R_{n})}{P_{n,1/2}(\hat{R}_n)^2}-1
\ge
\frac{e^{A}-1}{e^{A}}
\frac{P_{n,1}(\hat{R}_n)P_{n,0}(R_{n})}{P_{n,1/2}(\hat{R}_n)^2}.
\Label{8-27U}
\end{align}
Using \eqref{1-31-5} and \eqref{1-31-6},
we have
\begin{align}
\log \frac{P_{n,0}(\hat{R}_{n})}{P_{n,1/2}(\hat{R}_n)^2}
&= -\frac{1}{2}\log n 
-g_2(0) +o(1) ,
\end{align}
i.e.,
\begin{align}
\frac{P_{n,0}(\hat{R}_{n})}{P_{n,1/2}(\hat{R}_n)^2} \to 0.
\Label{8-27T}
\end{align}

Using \eqref{1-32-2}, we have
\begin{align}
\log
\frac{P_{n,0}(R_{n}')}
{ P_{n,0}(R_{n})}
&= (B_{n}-B_{n}') ({\psi_p'}^{-1}(R_n)+1)
+o(1). \Label{8-27S}
\end{align}

With sufficiently large $n$, we have
\begin{align}
\frac{\Bigl(\frac{P_{n,1}(R_{n}')P_{n,0}(R_{n}')}{P_{n,1/2}(R_{n})^2}-1\Bigr)^{\frac{1}{2}}
}
{\Bigl(\frac{P_{n,1}(\hat{R}_n)P_{n,0}(R_{n})}{P_{n,1/2}(\hat{R}_n)^2}
-1\Bigr)^{\frac{1}{2}}
}
\le &
\frac{\Bigl(\frac{P_{n,1}(R_{n}')P_{n,0}(R_{n}')}{P_{n,1/2}(R_{n})^2}\Bigr)^{\frac{1}{2}}
}
{\Bigl(\frac{P_{n,1}(\hat{R}_n)P_{n,0}(R_{n})}{P_{n,1/2}(\hat{R}_n)^2}
-1\Bigr)^{\frac{1}{2}}
}
\stackrel{(a)}{\le}
\frac{e^{A}}{e^{A}-1}
\frac{\Bigl(\frac{P_{n,1}(R_{n}')P_{n,0}(R_{n}')}{P_{n,1/2}(R_{n})^2}\Bigr)^{\frac{1}{2}}
}
{\Bigl(\frac{P_{n,1}(\hat{R}_n)P_{n,0}(R_{n})}{P_{n,1/2}(\hat{R}_n)^2}
\Bigr)^{\frac{1}{2}}
} \nonumber \\
= &
\frac{e^{A}}{e^{A}-1}
\Bigl(
\frac{P_{n,1}(R_{n}')P_{n,0}(R_{n}')P_{n,1/2}(\hat{R}_n)^2}
{P_{n,1/2}(R_{n})^2P_{n,1}(\hat{R}_n)P_{n,0}(R_{n})}
\Bigr)^{\frac{1}{2}}
\nonumber \\
\stackrel{(b)}{\le}
 &
\frac{e^{A}}{e^{A}-1}
\Bigl(
\frac{P_{n,1}(R_{n}')P_{n,0}(R_{n}')
P_{n,0}(\hat{R}_n)
P_{n,1}(\hat{R}_n)
}
{P_{n,1/2}(R_{n})^2P_{n,1}(\hat{R}_n)P_{n,0}(R_{n})}
\Bigr)^{\frac{1}{2}}
\nonumber \\
= &
\frac{e^{A}}{e^{A}-1}
\Bigl(
\frac{P_{n,1}(R_{n}') P_{n,0}(\hat{R}_n)}
{P_{n,1/2}(R_{n})^2 }
\Bigr)^{\frac{1}{2}}
\cdot
\Bigl(
\frac{P_{n,0}(R_{n}')}
{ P_{n,0}(R_{n})}
\Bigr)^{\frac{1}{2}}
\nonumber \\
\stackrel{(c)}{\le}
 &
\frac{e^{A}}{e^{A}-1}
\Bigl(
\frac{ P_{n,0}(\hat{R}_n)}
{P_{n,1/2}(R_{n})^2 }
\Bigr)^{\frac{1}{2}}
\cdot
\Bigl(
\frac{P_{n,0}(R_{n}')}
{ P_{n,0}(R_{n})}
\Bigr)^{\frac{1}{2}}
\nonumber \\
\stackrel{(d)}{\le}
 &
\frac{e^{A}}{e^{A}-1}
\Bigl(
\frac{ P_{n,0}(\hat{R}_n)}
{P_{n,1/2}(\hat{R}_{n})^2 }
\Bigr)^{\frac{1}{2}}
\cdot
\Bigl(
\frac{P_{n,0}(R_{n}')}
{ P_{n,0}(R_{n})}
\Bigr)^{\frac{1}{2}}
\stackrel{(e)}{\to}
 0,\Label{9-9-1}
\end{align}
where
$(a)$, $(b)$, $(c)$, $(d)$, and $(e)$ follow from 
\eqref{8-27U},
\eqref{9-8-2},
$P_{n,1}(R_{n}') \le 1$,
\eqref{9-8-1},
and the combination of \eqref{8-27T} and \eqref{8-27S}, respectively.
So, we obtain \eqref{9-7-4}.
\end{proofof}

\subsubsection{Exponential constraint}

\begin{theorem}\Label{11-07-10T}
\begin{align}
\log \beta_n(e^{-nr}| \varphi)
= n (2 \min_{0 \le s <1} (\frac{s}{1-s}r+\frac{1}{2}H_{\frac{1+s}{2}}(p) )
- \log d)
- \log n + O(1).\Label{11-07-10}
\end{align}
\hfill $\square$\end{theorem}

For the following discussion, 
given $r$, we define $R(r)$ and $s_r \le 0$ such that
\begin{align}
-r =\chi_0^{1}(-R(r)), \quad 
\psi_p'(s_r)&=-R(r).
\Label{11-07-20}
\end{align}
This definition is equivalent with
\begin{align}
-r = -\psi_p'(s_r)s_r+\psi_p(s_r), \quad
-R(r)= \psi_p'(s_r).
\end{align}
Since $\psi_p'$ is strictly monotone increasing, $R(r)> H(p)$.

We prepare the following lemmas.
\begin{lemma}\Label{L21}
We have the relations
\begin{align}
(s_r+ \frac{1}{2}) R(r) + \psi_p(s_r) 
&= \min_{0 \le s <1} \Big(\frac{s}{1-s}r+\frac{1}{2}H_{\frac{1+s}{2}}(p) \Big)
\Label{2-1-H} \\
s_r R(r) + \psi_p(s_r) &=-r.
\Label{2-1-H2}
\end{align}
\hfill $\square$\end{lemma}

\begin{lemma}\Label{L21C}
There exist three functions $h_i(r,d_S)$ ($i=1,2,3$) satisfying the following conditions.
%To show this theorem, 
Given real numbers $B_i$ with $i=1,2,3,4,5$,
we define
$R_{n,i}:=-R(r)+\frac{B_i}{n}$ with $i=1,2,3,4,5$.
Then,
\begin{align}
&
\log \frac{P_{n,1/2}(R_{n,1})P_{n,1/2}(R_{n,2})}
{P_{n,1}(R_{n,3})^{\frac{1}{2}} P_{n,1}(R_{n,4})^{\frac{1}{2}}
P_{n,0}(R_{n,5})}
=
-n r -\frac{1}{2} \log n 
+B_5-\frac{B_3+B_4}{2} + h_2(r,d_S) +o(1).
 \Label{1-5-2XC}
\end{align}
The convergences of the differences between the LHSs and RHSs 
are compact uniform for $B_i$.

Assume that
$\hat{R}_{n}:= -R(r)+\frac{\hat{B}_n}{n}$,
$R_n=-R(r)+\frac{B_n}{n}$,
and $R_n'=-R(r)+\frac{B_n'}{n}$.
When $B_n$ and $\hat{B}_{n}$ are bounded, and $B_n-B_n'$ converges,
we have
\begin{align}
&\log \epsilon_n(R_{n})
=
\log \frac{P_{n,1/2}(R_{n})^2}{d^n P_{n,1}(R_{n})} \nonumber \\
=&
n \Big(2 \min_{0 \le s <1} \Big(\frac{s}{1-s}r+\frac{1}{2}H_{\frac{1+s}{2}}(p) \Big)- \log d\Big)
%n (2(s_r+ \frac{1}{2}) R(r) + 2\psi_p(s_r) - \log d)
-\log n + 2(s_r+ \frac{1}{2}) B_n +2 h_3(r,d_S) %\chi_1^{\frac{1}{2}}(-R(r))
+ o(1),
%-\frac{1}{2}\log n - B +2\chi^{1/2}_1(-R(r)) -\chi^{1}_1(-R(r)) +o(1)
\Label{11-07-6C}
\\
& \log \frac{P_{n,0}(R_{n}')}{P_{n,1/2}(R_{n})}
e^{-\frac{n R_{n}}{2}} 
= h_1(r,d_S)+\frac{B_n-B_{n}'}{2}+o(1) \Label{11-4-1gC} ,
\\
& \log a_n(R_{n},R_{n}') 
=-n r+o(1) ,\Label{11-07-5C} 
\\
&\frac{\Bigl(\frac{P_{n,1}(R_{n}')P_{n,0}(R_{n}')}{P_{n,1/2}(R_{n})^2}-1\Bigr)^{\frac{1}{2}}
}
{\Bigl(\frac{P_{n,1}(\hat{R}_n)P_{n,0}(R_{n})}{P_{n,1/2}(\hat{R}_n)^2}
-1\Bigr)^{\frac{1}{2}}
}
%\cong
%\Bigl(
%\frac{P_{n,1}(R_{n}')P_{n,0}(R_{n}')P_{n,1/2}(\hat{R}_n)^2}{P_{n,1/2}(R_{n})^2P_{n,1}(\hat{R}_n)P_{n,0}(R_{n})}
%\Bigr)^{\frac{1}{2}}
=
e^{\frac{2B_n -B_n '-\hat{B}_{n}}{2}}+o(1) \Label{9-7-3bC}.
\end{align}

When $B_n \to - \infty$, $\hat{B}_{n}$ is bounded, and $B_n-B_n'$ converge,
\begin{align}
&\lim_{n \to \infty}\log
\frac{P_{n,0}(R_{n}')}{P_{n,1/2}(R_{n})}
e^{\frac{n R_{n}}{2}} 
\le h_1(r,d_S)
+\lim_{n \to \infty} \frac{B_n-B_n'}{2}
\Label{9-7-2C}\\
&\lim_{n \to \infty}
\frac{\Bigl(\frac{P_{n,1}(R_{n}')P_{n,0}(R_{n}')}{P_{n,1/2}(R_{n})^2}-1\Bigr)^{\frac{1}{2}}
}
{\Bigl(\frac{P_{n,1}(\hat{R}_n)P_{n,0}(R_{n})}{P_{n,1/2}(\hat{R}_n)^2}
-1\Bigr)^{\frac{1}{2}}
}
%\cong
%\Bigl(\frac{P_{n,1}(R_{n}')P_{n,0}(R_{n}')P_{n,1/2}(\hat{R}_n)^2}{P_{n,1/2}(R_{n})^2P_{n,1}(\hat{R}_n)P_{n,0}(R_{n})}\Bigr)^{\frac{1}{2}}
=0\Label{9-7-4C} .
\end{align}
\hfill $\square$\end{lemma}
The concrete construction of $h_i$ will be given in the proof of Lemma \ref{L21C}.

\begin{proofof}{Theorem \ref{11-07-10T}}\par
\PF{Non-lattice case}
\noindent{\it Step 1:)}
For simplicity, we first consider the case when $d_S=0$, i.e., the non-lattice case.
We fix $\hat{B}$.
Due to the non-lattice property (Lemma \ref{L9-20}), we can choose
$\hat{B}_n$ and $\hat{R}_{n}:= -R(r)+\frac{\hat{B}_n}{n}
 \in \{ \frac{1}{n}\log p_J^n\}_{J \in {\cal D}^n}$
such that $\hat{B}_n \to \hat{B}$.
Then, we will show 
\begin{align}
\lim_{n \to \infty}\frac{-1}{n}\log 
\alpha_n\left(
\left.
\frac{P_{n,1/2}(\hat{R}_n)^2}{d^n P_{n,1}(\hat{R}_n)}
\right|\varphi\right)
=r.\Label{7-6-eqC}
\end{align}
Since $\frac{P_{n,1/2}(\hat{R}_n)^2}{d^n P_{n,1}(\hat{R}_n)}$ is characterized by \eqref{11-07-6C},
\eqref{7-6-eqC} implies the desired argument when $d_S=0$.
Now, we outline the derivation of \eqref{7-6-eqC}. 
To show \eqref{7-6-eqC}, we find upper and lower bounds of \eqref{7-6-eqC} whose limit behaves as $e^{-nr}$.
For this purpose,
in Step 2:), we find its upper bound by using Item (1) of Lemma \ref{Th1},
and 
in Step 3:), we find its lower bound by using Item (2) of Lemma \ref{Th1}.
In Step 4:),
calculating both bounds, we show \eqref{7-6-eqC}.

\noindent{\it Step 2:)}
Assume that $\lim_{n\to \infty}B_n$ converges.
We choose $R_{n}:=-R(r)+\frac{B_n}{n}$.
Using \eqref{11-4-1gC} and \eqref{9-7-3bC},
we have
\begin{align}
\frac{P_{n,0}(R_{n})}{P_{n,1/2}(R_{n})}
e^{\frac{n R_{n}}{2}} 
+
\frac{\Bigl(\frac{P_{n,1}(R_{n})P_{n,0}(R_{n})}{P_{n,1/2}(R_{n})^2}-1\Bigr)^{\frac{1}{2}}
}
{\Bigl(\frac{P_{n,1}(\hat{R}_n)P_{n,0}(R_{n})}{P_{n,1/2}(\hat{R}_n)^2}
-1\Bigr)^{\frac{1}{2}}
}
=
e^{h_1(r,0)}+e^{\frac{B_n- \hat{B}}{2}}+o(1).\Label{9-7-6C}
\end{align}
Given $\delta>0$,
due to the non-lattice property (Lemma \ref{L9-20}), 
we chose $B_{\alpha,n}$ such that
$R_{\alpha,n}:=-R(r)+\frac{B_{\alpha,n}}{n}$
belongs to 
$\{ \frac{1}{n}\log p_J^n| \frac{1}{n}\log p_J^n < R_n \}_{J \in {\cal D}^n}$ and
\begin{align}
\lim_{n \to \infty}B_{\alpha,n}= \hat{B}+2 \log (1-e^{h_1(r,0)})+\delta
\Label{2-22-D1C}.
\end{align}
Then, in the same way as Step 2:) of the proof of Theorem \ref{11-07-2T}, we can show that $R_{\alpha,n}$ satisfies \eqref{2-22-A}.

\noindent{\it Step 3:)}
We choose $B_{n}$ as $R_n=-R(r)+\frac{B_n}{n}$.
Then, we choose $R_n'$ as the maximum element in 
$\{ \frac{1}{n}\log p_J^n| \frac{1}{n}\log p_J^n < R_n \}_{J \in {\cal D}^n}$.
So, the non-lattice property guarantees $\lim_{n \to \infty}n(R_n-R_n')=0$. 
When $B_n \to -\infty$,
\eqref{9-7-4C} and \eqref{9-7-2C} imply that
\begin{align}
\lim_{n \to \infty}
\frac{P_{n,0}(R_{n}')}{P_{n,1/2}(R_{n})}
e^{\frac{n R_{n}}{2}} 
+
\frac{\Bigl(\frac{P_{n,1}(R_{n}')P_{n,0}(R_{n}')}{P_{n,1/2}(R_{n})^2}-1\Bigr)^{\frac{1}{2}}
}
{\Bigl(\frac{P_{n,1}(\hat{R}_n)P_{n,0}(R_{n})}{P_{n,1/2}(\hat{R}_n)^2}
-1\Bigr)^{\frac{1}{2}}
}
\stackrel{(a)}{=}
e^{h_1(r,0)} <1\Label{9-7-5bC}.
\end{align}
where 
$(a)$ follows from 
$h_1(r,0)= \chi_1^0(R(r))-\chi_1^{\frac{1}{2}}(R(r))
=\log \frac{{\psi_p'}^{-1}(R(r))+\frac{1}{2}  }{{\psi_p'}^{-1}(R(r))+1}
<0$.

%Now, we consider the case when $\limsup_{n \to \infty}B_{n}<\infty$ and $\liminf_{n \to \infty}B_{n}>-\infty$.
When $ B_{n}$ is bounded, the combination of \eqref{11-4-1gC}  and \eqref{9-7-3bC} implies that
\begin{align}
\frac{P_{n,0}(R_{n}')}{P_{n,1/2}(R_{n})}
e^{\frac{n R_{n}}{2}} 
+
\frac{\Bigl(\frac{P_{n,1}(R_{n}')P_{n,0}(R_{n}')}{P_{n,1/2}(R_{n})^2}-1\Bigr)^{\frac{1}{2}}
}
{\Bigl(\frac{P_{n,1}(\hat{R}_n)P_{n,0}(R_{n})}{P_{n,1/2}(\hat{R}_n)^2}
-1\Bigr)^{\frac{1}{2}}
}
=e^{h_1(r,0)}+ e^{\frac{B_{n}- \hat{B}}{2}}+o(1). \Label{9-6-2C}
\end{align}
%Given $\deta>0$, 
Then, 
due to the non-lattice property (Lemma \ref{L9-20}), 
we can chose $B_{\beta,n}$ such that
$R_{\beta,n}:=-R(r)+\frac{B_{\beta,n}}{n}$
belongs to 
$\{ \frac{1}{n}\log p_J^n| \frac{1}{n}\log p_J^n < R_n \}_{J \in {\cal D}^n}$ and
\begin{align}
\lim_{n \to \infty}B_{\beta,n} = \hat{B}+2 \log (1-e^{h_1(r,0)})-\delta \Label{2-22-D1BC}.
\end{align}
In the same way as Step 3:) of the proof of Theorem \ref{11-07-2T}, we can show that
$R_{\beta,n}$
satisfies \eqref{2-22-B}.

\noindent{\it Step 4:)}
\eqref{2-22-D1C} and \eqref{2-22-D1BC} show that 
the sequences $B_{\alpha,n}$ and $B_{\beta,n}$ converge to constants as well as $\hat{B}_n$.
Thus, \eqref{11-07-5C} implies that
\begin{align}
\lim_{n \to \infty}\frac{-1}{n}\log
a_n(R_{\alpha,n},\hat{R}_n)
=
\lim_{n \to \infty}\frac{-1}{n}\log
a_n(R_{\beta,n},\hat{R}_n)
=r.\Label{8-27A5C}
\end{align}
Combining \eqref{2-22-A} and \eqref{2-22-B}, we obtain \eqref{7-6-eqC}.

\PF{Lattice case}
The lattice case ($d_S>0$) can be shown in the same way as the proof of Theorem \ref{11-07-2T}
by replacing $-H(p)+\frac{A_\epsilon}{\sqrt{n}}$ and $g_i(d_S)$
by $-R(r)$ and $h_i(r,d_S)$.

Next, we proceed to the lattice case with $d_S>0$.
Similar to the proof of Theorem \ref{11-07-2T},
the different points from the non-lattice case 
are the following.
Firstly, we notice that the limit $\lim_{n \to \infty}\hat{B}_n$ does not necessarily exist.
However, we can choose $\hat{B}_n$ such that 
$\hat{B}_n$ is bounded.
The above proof works even with such a bounded case.
The second point is the relation $\lim_{n \to \infty}n(R_n-R_n')=d_S>0$, which appears only in Steps 2:) and 3:).
In these steps, we need to replace $h_1(r,0)$ by $h_1(r,d_S)$.
In Step 2:), the relations \eqref{9-7-6C} and \eqref{2-22-D1C} are replaced by
\begin{align}
&\frac{P_{n,0}(R_{n})}{P_{n,1/2}(R_{n})}
e^{\frac{n R_{n}}{2}} 
+
\frac{\Bigl(\frac{P_{n,1}(R_{n})P_{n,0}(R_{n})}{P_{n,1/2}(R_{n})^2}-1\Bigr)^{\frac{1}{2}}
}
{\Bigl(\frac{P_{n,1}(\hat{R}_n)P_{n,0}(R_{n})}{P_{n,1/2}(\hat{R}_n)^2}
-1\Bigr)^{\frac{1}{2}}
}
=
e^{h_1(r,d_S)+\frac{d_S}{2}}+e^{\frac{B_n- \hat{B}}{2}+d_S}+o(1).
\Label{9-7-6BC}
\\
&B_{\alpha,n}:= \hat{B}-2 d_S+2 \log (1-e^{h_1(r,d_S)+\frac{d_S}{2}})
+\delta
\Label{2-22-D1B-C}.
\end{align}
In Step 3:), the relations \eqref{9-7-5bC}, \eqref{9-6-2C}, and \eqref{2-22-D1BC}
are replaced by
\begin{align}
&\lim_{n \to \infty}
\frac{P_{n,0}(R_{n}')}{P_{n,1/2}(R_{n})}
e^{\frac{n R_{n}}{2}} 
+
\frac{\Bigl(\frac{P_{n,1}(R_{n}')P_{n,0}(R_{n}')}{P_{n,1/2}(R_{n})^2}-1\Bigr)^{\frac{1}{2}}
}
{\Bigl(\frac{P_{n,1}(\hat{R}_n)P_{n,0}(R_{n})}{P_{n,1/2}(\hat{R}_n)^2}
-1\Bigr)^{\frac{1}{2}}
}
\le e^{h_1(r,d_S)+d_S} <1\Label{9-7-5b-2C},
\\
&\frac{P_{n,0}(R_{n}')}{P_{n,1/2}(R_{n})}
e^{\frac{n R_{n}}{2}} 
+
\frac{\Bigl(\frac{P_{n,1}(R_{n}')P_{n,0}(R_{n}')}{P_{n,1/2}(R_{n})^2}-1\Bigr)^{\frac{1}{2}}
}
{\Bigl(\frac{P_{n,1}(\hat{R}_n)P_{n,0}(R_{n})}{P_{n,1/2}(\hat{R}_n)^2}
-1\Bigr)^{\frac{1}{2}}
}
=e^{h_1(r,d_S)+d_S}+ e^{\frac{B_{n}- \hat{B}}{2}+d_S}+o(1), \Label{9-6-2-2C}
\\
&B_{\beta,n}:= \hat{B}-2d_S+2 \log (1-e^{h_1(r,d_S)+d_S})-\delta \Label{2-22-D1B-2C}.
\end{align}
Hence, the sequence $B_{\beta,n}$ is bounded as well as $\hat{B}_n$ and $B_{\alpha,n}$.
Thus, we obtain \eqref{8-27A5C}.
Combining \eqref{2-22-A} and \eqref{2-22-B}, we obtain \eqref{7-6-eqC} even in the lattice case
$d_S>0$.
\end{proofof}

\begin{proofof}{Lemma \ref{L21}}From 
Since $\psi_p'$ is monotone decreasing and $\psi_p'(0)=-H(p)$,
$R(r) $ 

Relation \eqref{9-17-eq1}, Condition \eqref{11-07-20}, and Proposition \ref{11-4-4}, 
we have
\begin{align*}
-r=\chi_0^{1}(-R(r)) 
= R(r) {\psi_p'}^{-1}(-R(r))+\psi_p(({\psi_p'}^{-1}(-R(r)))
= R(r) s_r+\psi_p(s_r).
%=\min_{t \le 0} t R(r) +\psi_p(t),
\end{align*}
%Since the minimum is attained with $s=t$,
Thus,
\begin{align}
-r=-s_r \psi_p'(s_r)+\psi_p(s_r), 
\end{align}
which implies that
$\frac{d}{dt} \frac{2t-1}{2t} r - \frac{\psi_p(t)}{2t} |_{t=s_r}=0$.
We also have
$\frac{d}{dt}(-\frac{2t+1}{2t} r - \frac{\psi_p(t)}{2t})
=\frac{\psi_p(t)-t \psi_p'(t)+1}{2t^2}$.
The derivative of denominator is $-t \psi_p''(t) \ge 0$ for $t \le 1$.
So, the derivative $\frac{d}{dt}(-\frac{2t+1}{2t} r - \frac{\psi_p(t)}{2t})
$ is non-negative if and only if $t \ge s_r $.
So, the minimum $ \min_{t \le 0} -\frac{2t+1}{2t} r - \frac{\psi_p(t)}{2t} $
is realized when $t=s_r$.
Hence,
\begin{align*}
&\chi_0^{1/2}(-R(r)) 
= R(r) ({\psi_p'}^{-1}(-R(r))+\frac{1}{2})+\psi_p({\psi_p'}^{-1}(-R(r)) )
= \frac{R(r)}{2} -r \\
=& -\frac{2 s_r+1}{2 s_r} r - \frac{\psi_p(s_r)}{2 s_r}
= \min_{t \le 0} -\frac{2t+1}{2t} r - \frac{\psi_p(t)}{2t} \\
=& \min_{t \le 0} -\frac{2t+1}{2t} r + \frac{H_{1+t}(p)}{2}
= \min_{0 \le s <1} \frac{s}{1-s}r+\frac{1}{2}H_{\frac{1+s}{2}}(p) ,
\end{align*}
where $t=-\frac{1-s}{2}$.
\end{proofof}

\begin{proofof}{Lemma \ref{L21C}}
\noindent {\it Step 1:)}
Similar to the proof of Lemma \ref{L2-1}, we show the desired relations by applying Proposition \ref{11-4-4} in Appendix \ref{A1}. 
In Step 1:), we prepare several relations and give the form of the function $h_i$.
We reuse \eqref{1-32-1} and \eqref{1-32-2} in the proof of Lemma \ref{L2-1}.
%Now, we show the relations \eqref{11-07-6C}, \eqref{11-4-1gC}, \eqref{1-5-2XC}, \eqref{11-4-1fC} and \eqref{11-07-5C}.
%In addition to \eqref{1-32-1} and \eqref{1-32-2},
Using Proposition \ref{11-4-4}, for $R< -H(p)$, we have the following relation.
\begin{align}
\log P_{n,1}^c(R)
=& 
n (  -R {\psi_p'}^{-1}(R)  + \psi_p({\psi_p'}^{-1}(R))
-\frac{1}{2}\log n +\chi_1^{1}(R)
+o(1)
\Label{1-32-3}.
\end{align}
Using $s_r= {\psi_p'}^{-1}(-R(r))$
and $\Delta \stackrel{\rm def}{=}  n {\psi_p'}^{-1}(-R(r)+\frac{B}{n})-s_r$, 
for any real number $B$, 
we have
\begin{align}
\psi_p({\psi_p'}^{-1}(-R(r)+\frac{B}{n}))
&= \psi_p(s_r)+\psi_p'(s_r)\frac{\Delta}{n}+ o(\frac{1}{n}) \\
-(-R(r)+\frac{B}{n}) {\psi_p'}^{-1}(-R(r)+\frac{B}{n})
&=R(r)(s_r+ \frac{\Delta}{n})- \frac{B}{n}s_r+ o(\frac{1}{n}).
\end{align}
Since $ \psi_p'(s_r)= - R(r)$,
we have
\begin{align}
-\Big(-R(r)+\frac{B}{n}\Big) {\psi_p'}^{-1}\Big(-R(r)+\frac{B}{n} \Big)
+\psi_p \Big({\psi_p'}^{-1} \Big(-R(r)+\frac{B}{n}\Big)\Big)
=
R(r) s_r+ \psi_p(s_r)
+ \frac{B}{n}s_r+ o(\frac{1}{n}).\Label{2-02-1}
\end{align}
Applying \eqref{2-02-1} to \eqref{1-32-1}, \eqref{1-32-2}, and \eqref{1-32-3}, 
we have
\begin{align}
\log P_{n,1}^c(-R(r)+\frac{B}{n})
=& 
n (s_r R(r) + \psi_p(s_r))-\frac{1}{2}\log n - s_r B  +\chi_1^{1}(-R(r))+ o(1)
\Label{1-31-5A} \\
\log P_{n,1/2}(-R(r)+\frac{B}{n})
=& 
n ((s_r+ \frac{1}{2}) R(r) + \psi_p(s_r))
-\frac{1}{2}\log n - (s_r+ \frac{1}{2}) B +\chi_1^{\frac{1}{2}}(-R(r))+ o(1)
\Label{1-31-5B} \\
\log P_{n,0}(-R(r)+\frac{B}{n})
=& 
n ((s_r+ 1) R(r) + \psi_p(s_r))
-\frac{1}{2}\log n - (s_r+ 1) B +\chi_1^{0}(-R(r))+ o(1)
\Label{1-31-6B} .
\end{align}
Now, we choose 
\begin{align}
h_{1}(r,d_S)&:= \chi_{1}^{0}(R(r)) - \chi_{1}^{\frac{1}{2}}(R(r)) \\
h_{2}(r,d_S)&:= 2 \chi_{1}^{\frac{1}{2}}(R(r)) - \chi_{1}^{0}(R(r)) \\
h_{3}(r,d_S)&:=\chi_{1}^{\frac{1}{2}}(R(r)).
\end{align}

\noindent {\it Step 2:)}
\PF{Proofs of \eqref{1-5-2XC} - \eqref{11-07-5C}}
Combining \eqref{1-31-5A}, \eqref{1-31-5B}, \eqref{1-31-6B},
and \eqref{2-1-H2} of Lemma \ref{L21},
we obtain \eqref{1-5-2XC}. 
Here, the compact uniformness of these convergence is guaranteed by the compact uniformness of the convergences in Proposition \ref{11-4-4}.
Combining \eqref{1-31-5B} and \eqref{2-1-H} of Lemma \ref{L21},
we obtain \eqref{11-07-6C}.
Combining \eqref{1-31-5B} and \eqref{1-31-6B},
we obtain \eqref{11-4-1gC}.
Using \eqref{1-5-2XC}, we obtain \eqref{22-E1}, \eqref{22-E2}, and \eqref{22-E3} in the same way as the proof of Lemma \ref{L2-1}.
Thus, combining \eqref{1-31-5A}, we obtain \eqref{11-07-5C}.

\PF{Proof of \eqref{9-7-3bC}}
When $B_n$ and $\hat{B}_{n}$ are bounded, and $B_n-B_n'$ converges,
using the relation \eqref{1-5-2XC}, we have
\begin{align}
\log \frac{P_{n,1}(R_{n}')P_{n,0}(R_{n}')}{P_{n,1/2}(R_{n})^2}
&= nr+ \frac{1}{2}\log n -(B_n'-B_n)
+h_2(r, d_S) +o(1), \Label{9-8-8C} \\
\log \frac{P_{n,1}(\hat{R}_n)P_{n,0}(R_{n})}{P_{n,1/2}(\hat{R}_n)^2}
&= nr+\frac{1}{2}\log n -(B_{n}-\hat{B})
+h_2(r, d_S) +o(1) .\Label{9-8-7C}
\end{align}
Therefore, we obtain \eqref{9-7-3bC}.

\PF{Proof of \eqref{9-7-2C}}
The relation \eqref{9-9-2} of the proof of \eqref{9-7-2} holds even in the current situation.
When $B_n \to - \infty$ and $B_n-B_n'$ converges,
since $\chi_1^0(R)-\chi_1^{\frac{1}{2}}(R)$ is monotone increasing for $R$,
we have
\begin{align}
\chi_1^0(R_n')
-\chi_1^{\frac{1}{2}}(R_n)
& \le 
\chi_1^0(- R(r))
-\chi_1^{\frac{1}{2}}(- R(r)) 
= h_1(r,d_S) \Label{9-9-3}.
\end{align}
Combinig \eqref{9-9-2} and \eqref{9-9-3}, we obtain \eqref{9-7-2C}.

\PF{Proof of \eqref{9-7-4C}}
\eqref{9-7-4C} can be shown as the same way as \eqref{9-7-4}.
The different point is \eqref{8-27T}, which is
%The relations \eqref{8-27T} and \eqref{8-27S} are 
replaced as follows.
Using \eqref{1-32-1} and \eqref{1-32-2},
we have
\begin{align}
\log \frac{P_{n,0}(\hat{R}_{n})}{P_{n,1/2}(\hat{R}_n)^2}
&= nr+\frac{1}{2}\log n 
+h_2(r, d_S) +o(1) .\Label{8-27TC}
\end{align}
Here, \eqref{8-27S} holds even in the curret situation.
Hence, using the same discussion as \eqref{9-9-1}, we obtain \eqref{9-7-4C}.
\end{proofof}

\subsection{Application to hypothesis testing under separable POVMs}
\Label{sec separableb}
Now, we choose the dimension $d\stackrel{\rm def}{=} \min (d_A,d_B)$
and the pure state $\varphi= \sum_{i=1}^d \sqrt{\lambda_i} 
|i\rangle \in \mathbb{C}^d$
by using the Schmidt coefficient $\{ \lambda
 _i \}_{i=1}^{d}$ of $\ket{\Psi}$. 
Then, we have the following proposition.
\begin{proposition}[\protect{\cite[Theorem 5]{OH10}}]\Label{theorem previous paper 1}
\begin{equation}
\beta_{n,sep}(\alpha | \Psi\| \rho_{mix})=
\bar{d}^{-n} \beta_n\left( \alpha |\varphi \right ),
\Label{2-10-1}
\end{equation}
where $\bar{d}$ is defined as 
\begin{equation}\Label{eq def d max}
\bar{d} \stackrel{\rm def}{=}\max \left( d_A,  d_B \right).
\end{equation}
\hfill $\square$\end{proposition}

Combining (\ref{2-10-1}) and Theorem \ref{11-07-2T}, 
we find that $\beta_{n,sep}\left(\epsilon |\Psi\|\rho_{mix} \right) $ 
can be given by (\ref{2-5-8}) because $\log d +\log d_{max}= \log d_A d_B$.
Similarly,
combining (\ref{2-10-1}) and Theorem \ref{11-07-10T}, 
we find that $ \beta_{n,sep}\left(e^{-nr} |\Psi\|\rho_{mix} \right)$ can be given by (\ref{2-5-10}).

\section{Hypothesis testing under two-way LOCC POVM} \Label{sec two-way}
\subsection{Construction of two-round classical communication protocol}
In this section, 
we consider $C=\leftrightarrow$, that is, the local
hypothesis testing under two-way LOCC POVMs. 
The previous paper \cite{OH10} proposed a specific class of two-round classical communication two-way LOCC protocols
that are not reduced to one-way LOCC.
In this subsection, we review their construction.
Then, in the latter subsections, we show that they can achieve the Hoeffding bound and Stein-Strassen bound for the class $C=sep$
by the following protocol.

For the entangled state $\ket{\tilde{\Psi}}\stackrel{\rm def}{=} \sum_{x \in {\cal X}}\sqrt{\lambda_x}
|x\rangle \otimes |x \rangle $
and the white noise state (the completely mixed state) $\rho_{mix}$,
For a given set $\Omega$,
a collection $\{m_\omega\}_{\omega \in \Omega}$ of non-negative measures on ${\cal X}$
is called a {\it subnormalized measure collection} on ${\cal X}$
when $\sum_{\omega \in \Omega}m_\omega(x) \le 1$ for any $x\in {\cal X}$.
Here, $\omega \in \Omega$ is an index indicating the measure $m_\omega$.
For a measure $m_\omega$ on ${\cal X}$, 
we denote the support of $m_\omega$ and its cardinality by ${\cal X}_\omega$ and $|m_\omega|$
and define the operator 
\begin{equation}\Label{eq def M omega}
		M_\omega \stackrel{\rm def}{=}
	\sum_{x \in {\cal X}} m_\omega(x) \ket{x}\bra{x} .
\end{equation} 
Then, for a collection $\{m_\omega\}_{\omega \in \Omega}$ of 
non-negative measures on ${\cal X}$,
we define the operator 
\begin{equation}\Label{eq def M omega2}
M^c\stackrel{\rm def}{=} I-\sum_{\omega \in \Omega}M_\omega. 
\end{equation} 
Then, we can define the POVM $M\stackrel{\rm def}{=}\{M_\omega \} \cup \{M^c\}$.
Using the collection $\{m_\omega\}_{\omega \in \Omega}$, 
we give a tree-step LOCC protocol to distinguish the two states
$\ket{\tilde{\Psi}}$ and $\rho_{mix}$ as follows: 

\begin{enumerate}
 \item Alice measures her state with a POVM $M$.
              When Alice's measurement outcome corresponds to $M^c$,
Alice and Bob stop the protocol and
       conclude the unknown state to be $\rho_{mix}$. Otherwise, they
       continue the protocol.

\item At the second step, Bob measures his state with a POVM 
       $\{ N_j^\omega\}_{j=0}^{|m_\omega|}$ depending on Alice's
      measurement outcome $\omega$.
      For $j \in \{1, \dots, |m_\omega| \}$, $N_j^\omega$ is defined as
      $N_j^\omega=\ket{\xi_j^\omega}\bra{\xi_j^\omega}$, where
      $\{\ket{\xi_j^\omega}\}_{j=1}^{|m_\omega|}$ is a mutually unbiased basis of the
      subspace $span\{\ket{h} \}_{h \in {\cal X}_\omega}$. Then, $N_0^\omega$ is defined as
      $N_0^\omega \stackrel{\rm def}{=} I_B - \sum_{j=1}^{|m_\omega|}
      N_j^\omega$.
      When Bob observes the measurement outcome $j=0$, Alice and Bob stop the protocol
      and conclude the unknown state to be $\rho_{mix}$. Otherwise, they
      continue the protocol.

\item At the third step, Alice measures her states with a two-valued POVM $\{
      O^{\omega j}, I_A-O^{\omega j} \}$. 
      Here, the POVM element $O^{\omega j}$ is chosen as Alice's state after Bob's measurement
      when the given state is $\ket{\tilde{\Psi}}$.
      Hence, $O^{\omega j}$ is defined as
      \begin{equation}\Label{eq def O omega j}
       O^{\omega j} \stackrel{\rm def}{=} \frac{\sqrt{M_\omega
      \sigma_A}\left(\ket{\xi_j^\omega}\bra{\xi_j^\omega} \right)^T \sqrt{M_\omega
      \sigma_A}}{\bra{\xi_j^\omega} M_\omega \sigma_A
      \ket{\xi_j^\omega}},
      \end{equation}
      where $\sigma_A \stackrel{\rm def}{=}\Tr _B \ket{\tilde{\Psi}}\bra{\tilde{\Psi}}$, 
      and $T$ is the transposition in the Schmidt basis of $\ket{\tilde{\Psi}}$.  
      When Alice's measurement result $k$ is $0$, Alice
      and Bob conclude the unknown state to be $\ket{\tilde{\Psi}}$;
      otherwise, they conclude the unknown state to be $\rho_{mix}$.
\end{enumerate}
Here, the above two-round classical communication protocol depends only on
the subnormalized measure collection $\{m_\omega\}_{\omega \in \Omega}$ on ${\cal X}$.
Hence, we denote the test given above by $T[\{m_\omega\}_{\omega \in \Omega}]$.
Then, we have the following proposition.

\begin{proposition}[\protect{\cite[Lemma 4]{OH10}}]\Label{P2-16}
The first and type-2 error probabilities of the test $T[\{m_\omega\}_{\omega \in \Omega}] $
are evaluated as
\begin{align}
\beta(T[\{m_\omega\}_{\omega \in \Omega}])
=& \Tr T[\{m_\omega\}_{\omega \in \Omega}] \rho_{mix} 
=
\sum_{\omega \in \Omega}
\frac{|m_\omega|\cdot \sum_{x \in {\cal X}} \lambda_x (m_\omega(x))^2 }
{d_Ad_B \sum_{x \in {\cal X}} \lambda_x m_\omega(x)}, \\
\alpha(T[\{m_\omega\}_{\omega \in \Omega}])
=& \langle \tilde{\Psi}| (I-T[\{m_\omega\}_{\omega \in \Omega}])|\tilde{\Psi}\rangle 
= \Tr \Tr_B|\tilde{\Psi}\rangle\langle \tilde{\Psi}| M^c \nonumber \\
=& 1-\sum_{\omega \in \Omega}\sum_{x \in {\cal X}} \lambda_x m_\omega(x).
\end{align}
\hfill $\square$\end{proposition}
In the above proposition, the type-1 and type-2 error probabilities 
are swapped to each other from Lemma 4 of \cite{OH10}.

\subsection{Hoeffding bound}
Now, we apply the above two-round classical communication protocol to the case of 
$\ket{\tilde{\Psi}}=\ket{\Psi}^{\otimes n}$
with $\ket{{\Psi}}\stackrel{\rm def}{=} \sum_{x \in {\cal X}}\sqrt{\lambda_x}
|x,x \rangle $.
Then, we give a two-round classical communication protocol to achieve the Hoeffding bound 
%$\sup _{0 \le s <1/2} \frac{-s(r-\log d_A d_B+H_{1-s}(\Psi))}{1-2s}$.
$\sup _{0 \le s <1} \frac{-2s}{1-s}r-H_{\frac{1+s}{2}}(\Psi) +\log d_A d_B$
for a given $r$ as follows.
When $r \ge \log d - \frac{1}{4} H_{1/2}(\Psi)'$,
we have
$\sup_{0\le s \le 1}\frac{-2sr}{1-s}-H_{\frac{1+s}{2}}(\Psi)+\log d_A d_B 
=\log d_A d_B -H_{1/2}$,
where $H_{s}(\Psi)' := \frac{d}{dt}H_{t}(\Psi)|_{t=s}$.
Hence,
it is enough to give the following two kinds of protocols:
One is a protocol in which
the exponential decreasing rates of the type-1 and type-2 errors are $r$ and 
$\sup_{0\le s \le 1}\frac{-2sr}{1-s}-H_{\frac{1+s}{2}}(\Psi)+\log d_A d_B $
for $r < \log d - \frac{1}{4} H_{1/2}(\Psi)'$.
The other is a protocol in which
the type-1 error is zero and 
the exponential decreasing rate of the second kind of error probability
is
$\log d_A d_B -H_{1/2}$. 
Before constructing the protocols, we prepare the following lemma.
Let $P$ be a distribution $(p_x)$ on ${\cal X}$
and $P_{1/2}$ be the measure $(p_x^{1/2})$ on ${\cal X}$. 

\begin{lemma}\Label{L12}
For $r < \log d - \frac{1}{4} H_{1/2}(\Psi)'$,
we have
\begin{align}
\min_{Q:D(Q\|P)\le r}
D(Q\|P) -H(Q)
= \sup _{0 \le s <1} \frac{-2s}{1-s}r-H_{\frac{1+s}{2}}(\Psi) .
%\log d_A d_B
\Label{2-3-20a}
\end{align}
In particular,
\begin{align}
\min_{Q}
D(Q\|P) -H(Q)
&=
D(P_{1/2}\|P) -H(P_{1/2})=
- H_{1/2}(\Psi).\Label{2-3-21a} \\
\min_{Q:D(Q\|P)=0}
D(Q\|P) -H(Q)
&=
- H_{1}(\Psi).\Label{2-3-21b} 
\end{align}
\hfill $\square$\end{lemma}
This lemma will be shown in Appendix \ref{A3}.

Using the above lemmas and the type method,
we make the protocols as follows.
For this purpose, 
we prepare notations for the type method.
When an $n$-trial data 
$\vec{x}_n\stackrel{\rm def}{=}(x_1,\ldots,x_n) \in {\cal X}^n$ is given,
we focus on the distribution $p(x)\stackrel{\rm def}{=}\frac{\# \{i| x_i=x \}}{n}$,
which is called the empirical distribution for data $\vec{x}_n$.
In the type method, an empirical distribution is called a type.
In the following, we denote the set of empirical distributions on ${\cal X}$
with $n$ trials by ${\cal T}_n$.
The cardinality $|{\cal T}_n|$ is bounded by $(n+1)^{|{\cal X}|-1}$ \cite{CKbook},
which increases polynomially with the number $n$.
That is,
\begin{align}
\lim_{n \to \infty} \frac{1}{n}\log |{\cal T}_n|=0.
\Label{9-6-1}
\end{align}
This property is the key idea in the type method.
Let $T_n(Q)$ be the set of $n$-trial data whose empirical distribution is $Q$.
Then, the cardinality $|T_n(Q)|$ can be evaluated as \cite{CKbook}
\begin{align}
\bigg\lceil \frac{e^{nH(Q)}}{|{\cal T}_n|} \bigg\rceil 
\le |T_n(Q)| \le \lfloor e^{nH(Q)} \rfloor,
\Label{9-7-5}
\end{align}
where $ \lceil a \rceil$ is the minimum integer $m$ satisfying $m \ge a $,
and $\lfloor a \rfloor$ is the maximum $m$ satisfying $m \le a $.
Since any element $\vec{x} \in T_n(Q)$ 
satisfies 
\begin{align}
P^{n}(\vec{x})\stackrel{\rm def}{=} P(x_1) \cdots P(x_n)
=e^{-n (D(Q\|P)+H(Q)) }, \Label{9-7-13}
\end{align}
we obtain the important formula
\begin{align}
\frac{1}{|{\cal T}_n|} e^{-n D(Q\|P)}
\le P^{n}(T_n(Q)) \le 
e^{-n D(Q\|P)}
\Label{9-7-14}.
\end{align}

Now, we are ready to mention the main theorem of this subsection.
\begin{theorem}\Label{Th2-4A}
For any $r < - \frac{1}{4} H_{1/2}(\Psi)'$ and $n$, there is 
a subnormalized measure collection $\{m_{n,\omega}^r \}_{\omega} $ on ${\cal X}^n$ such that
\begin{align}
\beta(T[\{m_{n,\omega}^r \}_{\omega}])
=& \Tr T[\{m_{n,\omega}^r \}_{\omega}] \rho_{mix}^{\otimes n} \le
4 |{\cal T}_n|^3 
(d_A d_B)^{-n}
e^{-n \sup _{0 \le s <1} \frac{-2s}{1-s}r-H_{\frac{1+s}{2}}(\Psi)},
\Label{2-3-18}\\
\alpha(T[\{m_{n,\omega}^r \}_{\omega}])
=& \langle {\Psi^{\otimes n}}| (I-T[\{m_{n,\omega}^r \}_{\omega}])|{\Psi}^{\otimes n}\rangle 
\le  |{\cal T}_n| e^{-n r}.\Label{2-3-19}
\end{align}
For the case with $r = - \frac{1}{4} H_{1/2}(\Psi)'$, we have the following statement.
For any $n$, there is a subnormalized measure collection 
$\{m_{n,\omega}^o \}_{\omega} $ on ${\cal X}^n$ 
such that
\begin{align}
\beta(T[\{m_{n,\omega}^o \}_{\omega}])
=& \Tr T[\{m_{n,\omega}^o \}_{\omega}] \rho_{mix}^{\otimes n} \le
4 |{\cal T}_n|^3 (d_A d_B)^{-n}
e^{n H_{1/2}(\Psi)},\Label{2-3-17}\\
\alpha(T[\{m_{n,\omega}^o \}_{\omega}])
=& \langle {\Psi^{\otimes n}}| (I-T[\{m_{n,\omega}^o \}_{\omega}])|{\Psi}^{\otimes n}\rangle =0.\Label{2-3-16}
\end{align}
\hfill $\square$\end{theorem}

This theorem guarantees that 
\begin{align}
\liminf_{n \to \infty}
\frac{-1}{n}\log 
\beta_{n,\leftrightarrow}\left(e^{-nr} |\Psi\|\rho_{mix} \right)
\ge 
\sup _{0 \le s <1} 
\frac{-2s}{1-s}r-H_{\frac{1+s}{2}}(\Psi) +\log d_A d_B.
%H_{\leftrightarrow}\left(r|\Psi\|\rho_{mix} \right) 
\end{align}
Since 
$\limsup_{n \to \infty}
\frac{-1}{n}\log 
\beta_{n,\leftrightarrow}\left(e^{-nr} |\Psi\|\rho_{mix} \right)
\le
\lim_{n \to \infty}
\frac{-1}{n}\log 
\beta_{n,sep}\left(e^{-nr} |\Psi\|\rho_{mix} \right)
=
H_{\leftrightarrow}\left(r|\Psi\|\rho_{mix} \right) $,
we obtain (\ref{2-5-10}).

In the following, we will concretely construct subnormalized measure collections to realize the conditions 
\eqref{2-3-18} and \eqref{2-3-19} (\eqref{2-3-17} and \eqref{2-3-16}).
Then, Theorem \ref{Th2-4A} will be shown as the combination of Lemmas \ref{L2-4A} and \ref{L2-4B}. 

\noindent{\it Construction of the subnormalized measure collection $\{m_{n,\omega}^r\}_{\omega \in \Omega}$ with $r < \log d - \frac{1}{4} H_{1/2}(\Psi)'$: \quad}
First, we fix the distribution $P$ so that $P(x)=\lambda_x$.
Then, we consider the case of $r < \log d - \frac{1}{4} H_{1/2}(\Psi)'$.
To choose a subnormalized measure collection $\{m_\omega^r\}_{\omega \in \Omega} $
on ${\cal X}^n$,  
we give two disjoint subsets of types by employing the type method as follows.
\begin{align*}
{\cal T}_{n,r}&\stackrel{\rm def}{=} \{Q \in {\cal T}_{n}|
- H(P) %= - H_1(\Psi)
>
D(Q\|P) -H(Q) , D(Q\|P)\le r %> r - \log d_A d_B
 \}, \\
{\cal T}_{n}'&\stackrel{\rm def}{=} \{Q \in {\cal T}_{n}|
- H(P)= - H_1(\Psi)\le 
D(Q\|P) -H(Q)  \}.
%{\cal T}_{n,r,c}&\stackrel{\rm def}{=} 
%\{Q \in {\cal T}_{n}| - H(P)> D(Q\|P) -H(Q), D(Q\|P)> r \}.
\end{align*}
In this construction, we fix the element $P_{n}\in {\cal T}_{n}'$ that is closest to $P$ among elements in ${\cal T}_{n}'$ 
in terms of relative entropy.
Then, we define the subset ${\cal T}_{n}''\stackrel{\rm def}{=} {\cal T}_{n}'\setminus \{P_n\}$.
%Since $n$ is sufficiently large, $P_n$ does not belong to ${\cal T}_{n,r}$.
%Moreover, since $n$ is sufficiently large, $P^n (T_n(P_n))$

Then, we divide the set $T_n(P_n)$ into 
$|{\cal T}_{n,r}|$ disjoint sets 
$ T_n(P_n)_{Q} $ ( $Q \in {\cal T}_{n,r}$)
whose cardinalities are $\lceil |T_n(P_n)|/|{\cal T}_{n,r}|\rceil$
or $\lfloor |T_n(P_n)|/|{\cal T}_{n,r}|\rfloor$.
For a type $Q \in {\cal T}_{n,r}$, we 
divide the set $T_n(Q)$ into 
$\lceil |T_n(Q)|/|T_n(P_n)_{Q}|\rceil $ disjoint sets 
$T_n(Q)_{1}, \ldots, T_n(Q)_{\lceil |T_n(Q)|/|T_n(P_n)_{Q}|\rceil}$
whose cardinalities are less than $|T_n(P_n)_{Q}|$.
Hence, for $Q \in {\cal T}_{n,r}$, (\ref{9-7-14}) yields
\begin{align}
P^n(T_n(P_n)_{Q}) 
\ge \frac{e^{-n D(P_n\|P)}}{|{\cal T}_n|\cdot
|{\cal T}_{n,r}|}
\ge \frac{e^{-n D(P_n\|P)}}{|{\cal T}_n|^2},
\Label{2-10-11}
\end{align}
and (\ref{9-7-5}) yields
\begin{align}
|T_n(P_n)_{Q}| \le \lfloor e^{n H(P_n)} \rfloor
\le \lfloor e^{n H(P)} \rfloor
\le \lfloor e^{n H(Q)} \rfloor.
\Label{2-10-12}
\end{align}

For a type $Q \in {\cal T}_{n}''$,
we define the non-negative measure $\bar{m}_{Q}$ on ${\cal X}^n$ as
\begin{align}
\bar{m}_{Q} (\vec{x})\stackrel{\rm def}{=}
\left\{
\begin{array}{ll}
1 & \hbox{if } \vec{x} \in T_n(Q) \\
0 & \hbox{otherwise.}
\end{array}
\right.
\end{align}
For a type $Q \in {\cal T}_{n,r}$
and $k=1, \ldots, \lceil |T_n(Q)|/|T_n(P_n)_{Q}|\rceil$,
we define the non-negative measure $\bar{m}_{Q,k}$ on ${\cal X}^n$ as
\begin{align}
\bar{m}_{Q,k} (\vec{x})\stackrel{\rm def}{=}
\left\{
\begin{array}{ll}
1 & \hbox{if } \vec{x} \in T_n(Q)_{k} \\
\lceil |T_n(Q)|/|T_n(P_n)_{Q}|\rceil^{-1} 
& \hbox{if } \vec{x} \in T_n(P_n) \setminus T_n(Q)_{k} \\
0 & \hbox{otherwise.}
\end{array}
\right.
\end{align}
Hence, the cardinality $|\bar{m}_{Q,k}|$ is less than 
$|T_n(P_n)_{Q}|+ |T_n(Q)_j| \le 2 |T_n(P_n)_{Q}|$.
Now, we choose the set $\Omega$ as
$\Omega:={\cal T}_{n}''\cup \{(Q,j)\}_{Q \in {\cal T}_{n,r}}$,
where $k$ takes values in $\{1, \ldots, \lceil |T_n(Q)|/|T_n(P_n)_{Q}|\rceil\}$.
Then, we define the subnormalized measure collection 
$\{m_{n,\omega}^r\}_{\omega \in \Omega}$ as
\begin{align}
m_{n,\omega}^r:=
\left\{
\begin{array}{ll}
\bar{m}_Q & 
\hbox{ if } \omega = Q \in  {\cal T}_{n}''\\
\bar{m}_{Q,j} &
\hbox{ if } \omega =(Q,k) \hbox{ with } Q \in  {\cal T}_{n,r}.
\end{array}
\right.\Label{2-4-F}
\end{align}From the above construction, we find that $\{m_{n,\omega}^r\}_{\omega \in \Omega}$
is a subnormalized measure collection on ${\cal X}^n$.
\hfill $\square$

Then, we have the following lemma.
\begin{lemma}\Label{L2-4A}
The subnormalized measure collection $\{m_{n,\omega}^r\}_{\omega \in \Omega}$ on ${\cal X}^n$
satisfies (\ref{2-3-18}) and (\ref{2-3-19}).
\hfill $\square$\end{lemma}

To show Lemma \ref{L2-4A}, we prepare the following lemma.
\begin{lemma}\Label{11-07-30L}
Assume that 
$n$ is sufficiently large. Then,
\begin{align}
D(P_n\|P) \le \frac{2d}{n} \Label{11-07-30}.
\end{align}
\hfill $\square$\end{lemma}
\begin{proof}
We denote $P(i)- P_n(i)$ by $\delta_{n,i}$.
Since $n$ is sufficiently large,
we have 
$-\log (1+\frac{\delta_{n,i}}{P_n(i)}) \le -2 \frac{\delta_{n,i}}{P_n(i)} $.
Using the relation $|\delta_{n,i}|\le \frac{1}{n} $, we have
\begin{align}
& D(P_n\|P)=\sum_{i=1}^d 
P_n(i)
\log \frac{P_n(i)}{P_n(i)+ \delta_{n,i}} 
= 
-\sum_{i=1}^d 
P_n(i)
\log (1+\frac{\delta_{n,i}}{P_n(i)}) \nonumber \\
\le & 
\sum_{i=1}^d 
P_n(i)
-2 \frac{\delta_{n,i}}{P_n(i)} 
= 
-2 \sum_{i=1}^d \delta_{n,i}
\le \frac{2d}{n}.
\end{align}
\end{proof}

\begin{proofof}{Lemma \ref{L2-4A}}
To calculate 
$\beta(T[\{m_{n,\omega}^r\}_{\omega \in \Omega}])
= \Tr T[\{m_{n,\omega}^r\}_{\omega \in \Omega} ] \rho_{mix}^{\otimes n}$, 
we firstly evaluate\par\noindent
$ \sum_{\vec{x} \in {\cal X}^n} P^n(\vec{x}) m_\omega^r (\vec{x})^2 $
and
$\sum_{\vec{x} \in {\cal X}^n} P^n(\vec{x}) m_\omega^r(\vec{x})$
as
\begin{align}
& \sum_{\vec{x} \in {\cal X}^n} 
P^n(\vec{x}) m_\omega^r (\vec{x})^2 \nonumber \\
= &
\sum_{\vec{x} \in T_n(Q)_{j}} 
P^n(\vec{x}) m_\omega^r(\vec{x})^2 
+
\sum_{\vec{x} \in T_n(P_n)_{Q}} 
P^n(\vec{x}) 
\lceil |T_n(Q)|/|T_n(P_n)_{Q}|\rceil^{-2}\nonumber \\
= &
P^n (T_n(Q)_{j})+
P^n (T_n(P_n)_{Q})
\lceil |T_n(Q)|/|T_n(P_n)_{Q}|\rceil^{-2} ,
\Label{2-10-14}
\end{align}
and
\begin{align}
& \sum_{\vec{x} \in {\cal X}^n} 
P^n(\vec{x}) m_\omega^r(\vec{x}) \nonumber\\
= &
\sum_{\vec{x} \in T_n(Q)_{j}} 
P^n(\vec{x}) m_\omega^r(\vec{x}) 
+
\sum_{\vec{x} \in T_n(P_n)_{Q}} 
P^n(\vec{x}) 
\lceil |T_n(Q)|/|T_n(P_n)_{Q}|\rceil^{-1} \nonumber \\
= &
P^n (T_n(Q)_{j})+
P^n (T_n(P_n)_{Q})
\lceil |T_n(Q)|/|T_n(P_n)_{Q}|\rceil^{-1} \nonumber \\
\ge &
P^n (T_n(P_n)_{Q})
\lceil |T_n(Q)|/|T_n(P_n)_{Q}|\rceil^{-1} .
\Label{2-10-15}
\end{align}

Now, we evaluate the two kinds of errors
for the above collection of non-negative measures.
The first kind of error probability is evaluated as
\begin{align}
&\beta(T[\{m_{n,\omega}^r\}_{\omega \in \Omega}])
= \Tr T[\{m_{n,\omega}^r\}_{\omega \in \Omega} ] \rho_{mix}^{\otimes n} \nonumber \\
\stackrel{(a)}{\le} &
\sum_{Q \in {\cal T}_{n,r}}
\sum_j 
\frac{2 |T_n(P_n)|
\cdot 
\sum_{\vec{x} \in {\cal X}^n} 
P^n(\vec{x}) (\bar{m}_{Q,j}(x))^2 
}
{d_A^n d_B^n \sum_{x \in {\cal X}} P^n(\vec{x}) \bar{m}_{Q,j}(x)} 
+ 
\sum_{Q \in {\cal T}_{n}''}
\frac{|T_n(Q)|}{d_A^n d_B^n} \nonumber \\
\le &
\sum_{Q\in {\cal T}_{n,r}}
\frac{2 |T_n(P_n)_Q|
\cdot 
\sum_j
P^n (T_n(Q)_{j})+
P^n (T_n(P_n)_{Q})
\lceil |T_n(Q)|/|T_n(P_n)_{Q}|\rceil^{-2} 
}
{d_A^n d_B^n P^n (T_n(P_n)_{Q})
\lceil |T_n(Q)|/|T_n(P_n)_{Q}|\rceil^{-1} } \nonumber \\
&+ 
\sum_{Q \in {\cal T}_{n}''}
\frac{|T_n(Q)|}{d_A^n d_B^n} \nonumber \\
= &
\sum_{Q\in {\cal T}_{n,r}}
\frac{2 |T_n(P_n)_Q|
\cdot 
P^n (T_n(Q))+
P^n (T_n(P_n)_{Q})
\lceil |T_n(Q)|/|T_n(P_n)_{Q}|\rceil^{-1} 
}
{d_A^n d_B^n P^n (T_n(P_n)_{Q})
\lceil |T_n(Q)|/|T_n(P_n)_{Q}|\rceil^{-1} } \nonumber \\
&+ 
\sum_{Q \in {\cal T}_{n}''}
\frac{|T_n(Q)|}{d_A^n d_B^n} \nonumber \\
= &
\sum_{Q\in {\cal T}_{n,r}}
\frac{2 |T_n(P_n)_Q|
\cdot P^n (T_n(Q))}
{d_A^n d_B^n P^n (T_n(P_n)_{Q})
\lceil |T_n(Q)|/|T_n(P_n)_{Q}|\rceil^{-1} } \nonumber \\
&+
\sum_{Q\in {\cal T}_{n,r}}
\frac{2 |T_n(P_n)_Q|}
{d_A^n d_B^n} 
+ 
\sum_{Q \in {\cal T}_{n}''}
\frac{|T_n(Q)|}{d_A^n d_B^n}  \nonumber \\
= &
\sum_{Q\in {\cal T}_{n,r}}
\frac{2 |T_n(P_n)_Q|
\cdot P^n (T_n(Q))}
{d_A^n d_B^n P^n (T_n(P_n)_{Q})
\lceil |T_n(Q)|/|T_n(P_n)_{Q}|\rceil^{-1} } 
+
\frac{2 |T_n(P_n)|}
{d_A^n d_B^n} \nonumber \\
&+ 
\sum_{Q \in {\cal T}_{n}''}
\frac{|T_n(Q)|}{d_A^n d_B^n}  \nonumber \\
\le &
\sum_{Q\in {\cal T}_{n,r}}
\frac{2 |T_n(P_n)_Q|\cdot  (1+|T_n(Q)|/|T_n(P_n)_{Q}|)
\cdot P^n (T_n(Q))}
{d_A^n d_B^n P^n (T_n(P_n)_{Q})
 } 
+(|{\cal T}_{n}'|+1)
\frac{e^{nH(P)}}{d_A^n d_B^n}  \nonumber \\
= &
\sum_{Q\in {\cal T}_{n,r}}
\frac{2 (|T_n(Q)|+|T_n(P_n)_Q|) \cdot P^n (T_n(Q))}
{d_A^n d_B^n P^n (T_n(P_n)_{Q}) } 
+(|{\cal T}_{n}'|+1)
\frac{e^{nH(P)}}{d_A^n d_B^n}  \nonumber \\
\stackrel{(b)}{\le} &
\sum_{Q\in {\cal T}_{n,r}}
\frac{2d \cdot 4 
|{\cal T}_n|^2  e^{n (H(Q)-D(Q\|P))}}{d_A^n d_B^n } 
+ 2|{\cal T}_{n}'|
\frac{e^{nH(P)}}{d_A^n d_B^n}  \nonumber \\
\stackrel{(c)}{\le} &
\frac{8d |{\cal T}_n|^3
 e^{-n (\min_{Q\in {\cal T}_{n,r} } D(Q\|P)-H(Q))}}
{d_A^n d_B^n } 
\stackrel{(d)}{\le}
8d |{\cal T}_n|^3 
e^{-n \sup _{0 \le s <1} \frac{-2s}{1-s}r-H_{\frac{1+s}{2}}(\Psi) },
\Label{1-5-12}
\end{align}
where
$(a)$ follows from (\ref{2-10-14}) and (\ref{2-10-15}),
$(b)$ follows from (\ref{2-10-11}), (\ref{2-10-12}), \eqref{9-7-14},
and Lemma \ref{11-07-30L},
and
$(c)$ follows from 
the inequality $\min_{Q\in {\cal T}_{n,r} } D(Q\|P)-H(Q))
\le - H_1(\Psi) \le
\min_{Q\in {\cal T}_{n}'' } D(Q\|P)-H(Q))$.

The second kind of error probability is evaluated as
\begin{align}
&\alpha(T[\{m_{n,\omega}^r\}_{\omega \in \Omega}])
= \langle {\Psi^{\otimes n}}| (I-T[\{m_{n,\omega}^r\}_{\omega \in \Omega}])|{\Psi}^{\otimes n}\rangle 
= 
\sum_{Q \in {\cal T}_{n,r,c}}
\sum_{\vec{x} \in T_n(Q)} 
P^n(\vec{x})\nonumber\\
=&
\sum_{Q \in {\cal T}_{n,r,c}}
P^n (T_n(Q)) 
\stackrel{(a)}{\le} 
|{\cal T}_n| e^{-n \min_{Q\in {\cal T}_{n,r,c} } D(Q\|P)}
\le |{\cal T}_n| e^{-n r},
\end{align}
where
$(a)$ follows from (\ref{9-7-14}).
\end{proofof}

\noindent{\it Construction of a subnormalized measure collection with $r = \log d - \frac{1}{4} H_{1/2}(\Psi)'$: \quad}
We consider the case of $r =\log d_A d_B- H_{1/2}(\Psi)$.
In this case, we change the definition of the subset
${\cal T}_{n,r}$ of ${\cal T}_{n}$ %, ${\cal T}_{n}' $, and ${\cal T}_{n,r,c} $ 
as 
\begin{align*}
{\cal T}_{n,r}&\stackrel{\rm def}{=} \{Q \in {\cal T}_{n}|
- H(P) > D(Q\|P) -H(Q) \}.
%{\cal T}_{n}'&\stackrel{\rm def}{=} \{Q \in {\cal T}_{n}|- H(P)= - H_1(\Psi)\le D(Q\|P) -H(Q)  \}, \\
%{\cal T}_{n,r,c}&\stackrel{\rm def}{=} \emptyset.
\end{align*}
So, we find that ${\cal T}_{n,r} \cup {\cal T}_{n}'= {\cal T}_{n}$.

Then, using the same discussion as the above,
we define the collection 
$\{\bar{m}_{Q,j} \}_{Q,j}$ of non-negative measures on ${\cal X}^n$
by using the modified subset ${\cal T}_{n,r}$.
We define the subnormalized measure collection $\{m_{n,\omega}^o\}_{\omega \in \Omega}$ on ${\cal X}^n$
by using \eqref{2-4-F}.
\hfill $\square$

Then, we have the following lemma.
\begin{lemma}\Label{L2-4B}
The subnormalized measure collection $\{m_{n,\omega}^o\}_{\omega \in \Omega}$ on ${\cal X}^n$
satisfies (\ref{2-3-17}) and (\ref{2-3-16}).
\hfill $\square$\end{lemma}

\begin{proofof}{Lemma \ref{L2-4B}}
Trivially, we have (\ref{2-3-16}).
Even in this modification, \eqref{1-5-12} still holds except for $(d)$.
Instead of $(d)$, we use \eqref{2-3-21a} of Lemma \ref{L12}.
Then, we have (\ref{2-3-17}).
\end{proofof}

\subsection{Stein-Strassen bound}
Now, we give a two-round classical communication protocol to achieve the Stein-Strassen bound.
For this purpose, we prepare the following lemma. 

\begin{lemma}\Label{2-16-10}
For a given $\epsilon>0$,
there exists a subnormalized measure collection $\{m_k\}_{k=0}^{M_n}$ such that
\begin{align}
& \log \sum_{k=1}^{M_n}
| \{\vec{x}|  m_k(\vec{x})\neq 0 \}|
\frac{\sum_{\vec{x} \in {\cal X}^n}P^n(\vec{x}) m_k(\vec{x})}
{\sum_{\vec{x} \in {\cal X}^n}P^n(\vec{x}) m_k(\vec{x})^2 } 
\nonumber \\
\le &
n H_1(\Psi) + \sqrt{n} \sqrt{V(\Psi)} \Phi^{-1}(\epsilon)
-\log n +O(1), \Label{3-1-1} \\
&\sum_{k=1}^{M_n}
\sum_{\vec{x} \in {\cal X}^n}P^n(\vec{x}) m_k(\vec{x}) 
= \epsilon +o(1).
\Label{3-1-2}
\end{align}
\hfill $\square$\end{lemma}
This lemma will be shown as Lemma \ref{L2-4C}.

Now, we are ready to mention the main theorem of this subsection.
Applying Proposition \ref{P2-16} to the subnormalized measure collection given in Lemma \ref{2-16-10},
we have the following theorem by using $\epsilon'=1-\epsilon$.
\begin{theorem}\Label{th2-16}
For any real number $\epsilon \in (0,1)$,
there is a collection $\{m_{n,\omega} \}_{\omega} $
of non-negative measures on ${\cal X}^n$ such that
\begin{align}
\log \beta(T[\{m_{n,\omega} \}_{\omega}])
& \le 
-n (\log d_A d_B -H_{1}(\Psi)) -\sqrt{n} \sqrt{V(\Psi)} 
\Phi^{-1}(\epsilon') - \log n +O(1),
 \Label{2-3-18b}\\
\alpha(T[\{m_{n,\omega} \}_{\omega}])
& \to \epsilon'.\Label{2-3-19b}
\end{align}
\hfill $\square$\end{theorem}

In Subsection \ref{sec separableb}, we have already shown that 
$\beta_{n,sep}\left(\epsilon' |\Psi\|\rho_{mix} \right) $ can be given by (\ref{2-5-8}).
Hence, $\beta_{n,\leftrightarrow}\left(\epsilon' |\Psi\|\rho_{mix} \right) 
\ge -n (\log d_A d_B -H_{1}(\Psi)) -\sqrt{n} \sqrt{V(\Psi)} \Phi^{-1}(\epsilon')
- \log n +O(1)$.
Theorem \ref{th2-16} guarantees the opposite inequality.
Hence, we obtain the remaining part of (\ref{2-5-8}).

\noindent{\it Construction of subnormalized measure collection: \quad}
Now, to show Lemma \ref{2-16-10}, we construct the subnormalized measure collection $\{m_{n|\epsilon,k}\}_{k=0}^{M_n}$ as follows.
For this purpose, when $\log P(x)-\log P(x') $ is a lattice variable, 
we define the real number $c$ to be the lattice span $d_S$.
When $\log P(x)-\log P(x') $ is a non-lattice variable, 
we define the real number $c$ to be an arbitrary positive real number.
For the definitions of lattice and non-lattice variables and the lattice span $d_S$, see Appendix \ref{A1}.
We fix $a,b>0$ such that $c>a$.

%\begin{align}
%c >\max_{x\neq x' \in {\cal X}} \log P(x)-\log P(x').
%\Label{2-15-6}
%\end{align}
Then, we prepare the following lemma.
\begin{lemma}
The function $f(t) \stackrel{\rm def}{=}\min_{s\ge 0} -sH_{1+s}(\Psi)+(1+s) (H_{1}(\Psi) - ct)
- (H_{1}(\Psi) -b -a t)$ monotonically decreases for $t>0$,
and there uniquely exists  
$t_0>0$ such that $f(t_0)=0$.
\hfill $\square$\end{lemma}
\begin{proof}
Since
$f(t) =b+\min_{s\ge 0} s (H_{1}(\Psi)-H_{1+s}(\Psi)) -(sc+ c-a) t$
and $sc+ c-a>0 $,
$f(t)$ is strictly monotonically decreasing for $t>0$.

Since $H_{1}(\Psi)-H_{1+s}(\Psi)\ge 0$ with $s\ge 0$ and its equality holds only with $s=0$,
we have
$f(0)=b+\min_{s\ge 0} s (H_{1}(\Psi)-H_{1+s}(\Psi)) =b>0$.
On the other hand,
for a fixed $s \ge 0$,
$b+s (H_{1}(\Psi)-H_{1+s}(\Psi)) -(sc+ c-a) t$ goes to $-\infty$
when $t$ goes to the infinity.
Hence, $f(t)$ goes to $-\infty$ when $t$ goes to infinity.
Thus, there uniquely exists  
$t_0>0$ such that $f(t_0)=0$.
\end{proof}

Now, we fix $t \in (0,t_0)$, and define
\begin{align}
{\cal R}_{k,n|\epsilon}\stackrel{\rm def}{=}
\left\{
\vec{x} \in {\cal X}^n
\left|
\begin{array}{l}
n H_1(\Psi) + \sqrt{n} \sqrt{V(\Psi)} \Phi^{-1}(\epsilon) - ck
\ge - \log P^n(\vec{x}) \\
> n H_1(\Psi) + \sqrt{n} \sqrt{V(\Psi)} \Phi^{-1}(\epsilon) - c(k+1) 
\end{array}
\right.\right\},
\end{align}
and
\begin{align}
M_n \stackrel{\rm def}{=} \lfloor e^{nb} \rfloor , ~
N_n \stackrel{\rm def}{=} |{\cal R}_{0,n}|M_n^{-1} ,~
N_{k,n} \stackrel{\rm def}{=} N_n e^{-ka} .
%k_n  \stackrel{\rm def}{=} n t .
\end{align}

For $k \le tn$, we define $M_n$ subsets 
${\cal R}_{k,n,1|\epsilon}, \ldots, {\cal R}_{k,n,M_n|\epsilon}$
of ${\cal R}_{k,n|\epsilon}$,
whose cardinalities are $N_{k,n}$.
We define the measure $m_{n,j}$ ($j=1, \ldots, M_n$)
as the measure satisfying the following two conditions.
The support of $m_{n|\epsilon,j}$ is 
${\cal S}_{j,n}\stackrel{\rm def}{=}\cup_{k=0}^{tn} {\cal R}_{k,n,j|\epsilon}$.
For $\vec{x} \in \cup_{k=0}^{tn} {\cal R}_{k,n|\epsilon}$,
the relation $\sum_{j=1}^{M_n}m_{n|\epsilon,j}(\vec{x})=1$ holds.
That is, $\{m_{n|\epsilon,k}\}_{k=0}^{M_n}$ 
forms a subnormalized measure collection.
\hfill $\square$

\begin{lemma}\Label{L2-4C}
The subnormalized measure collection $\{m_{n|\epsilon,k}\}_{k=0}^{M_n}$ 
satisfies (\ref{3-1-1}) and (\ref{3-1-2}).
\hfill $\square$\end{lemma}

In the following, for the simplicity, we omit the subscript $|\epsilon$.
For our proof of Lemma \ref{L2-4C},
we prepare the following lemma.
\begin{lemma}\Label{2-16-1}
\begin{align}
\log |{\cal R}_{0,n}| 
&=n H_1(\Psi) + \sqrt{n} \sqrt{V(\Psi)} \Phi^{-1}(\epsilon)
-\frac{1}{2}\log n + O(1).\Label{2-12-2}  \\
\log \sum_{k=0}^{\infty} |{\cal R}_{k,n}| e^{ka}
& = n H_1(\Psi) + \sqrt{n} \sqrt{V(\Psi)} \Phi^{-1}(\epsilon)
-\frac{1}{2}\log n + O(1) \Label{2-12-1}\\
\log \max_{k=0,\ldots, tn} P^n({\cal R}_{k,n})
&=-\frac{1}{2}\log n+ O(1) \Label{2-16-7}, 
\end{align}
and
\begin{align}
&P^n\left\{
\vec{x} \in {\cal X}^n
\left|
n H_1(\Psi) + \sqrt{n} \sqrt{V(\Psi)} \Phi^{-1}(\epsilon)
< - \log P^n(\vec{x}) 
\right.\right\}
\to \epsilon. \Label{2-15-2}
\end{align}
\hfill $\square$\end{lemma}

This lemma will be shown in the end of this subsection.
Using Lemma \ref{2-16-1}, we can show the following lemma.

\begin{lemma}\Label{2-16-9}
There exist an integer $N$ and a real number $C$ 
such that any integer $n \ge N$ satisfies the following conditions.
The inequalities
\begin{align}
&N_{k,n} \le  |{\cal R}_{k,n}| 
\hbox{  for  any integer $k$ satisfying } k \le tn,
\Label{2-15-3} \\
& \log \sum_{j=1}^{M_n}
| \{\vec{x}|  m_j(\vec{x})\neq 0 \}|
\frac{\sum_{\vec{x} \in {\cal X}^n}P^n(\vec{x}) m_j(\vec{x})}
{\sum_{\vec{x} \in {\cal X}^n}P^n(\vec{x}) m_j(\vec{x})^2 } \nonumber \\
= &
 \log \sum_{j=1}^{M_n}\Big(\sum_{k=0}^{tn} N_{k,n}\Big)
\frac{\sum_{\vec{x} \in {\cal X}^n}P^n(\vec{x}) m_j(\vec{x})}
{\sum_{\vec{x} \in {\cal X}^n}P^n(\vec{x}) m_j(\vec{x})^2 } \nonumber \\
\le &
n H_1(\Psi) + \sqrt{n} \sqrt{V(\Psi)} \Phi^{-1}(\epsilon)
-\log n + C \Label{2-15-3b}, \\
&P^n\left\{
\vec{x} \in {\cal X}^n
\left|
n H_1(\Psi) + \sqrt{n} \sqrt{V(\Psi)} \Phi^{-1}(\epsilon)
- c tn
\ge - \log P^n(\vec{x}) 
\right.\right\} \nonumber  \\
\le & \min_{s \ge 0} 
e^{s n (-H_{1+s}(\Psi)+ H_1(\Psi)- c t  + \frac{1}{\sqrt{n}} \sqrt{V(\Psi)} \Phi^{-1}(\epsilon)) } \Label{2-15-1} 
\end{align}
hold.
\hfill $\square$\end{lemma}

\begin{proofsof}{Lemma \ref{L2-4C}}
From (\ref{2-15-2}), (\ref{2-15-3b}), and (\ref{2-15-1}),
we find that the above subnormalized measure collection $\{m_{n,k}\}_{k=0}^{M_n}$ 
satisfies (\ref{3-1-1}) and (\ref{3-1-2}) of Lemma \ref{2-16-10}
because the right hand side of (\ref{2-15-1}) goes to zero.
So, we obtain Lemma \ref{L2-4C}.
\end{proofsof}

\begin{proofof}{Lemma \ref{2-16-9}}\par
\PF{Proof of (\ref{2-15-1}) and (\ref{2-15-3})}
%First, due to the definition of $c$, we can choose a sufficiently large integer $n'$ such that${\cal R}_{k,n'} \neq \empty$ for $k \le n t_0$.
%Since $n'$ is sufficiently large, ${\cal R}_{k,n'} \neq \empty$ for $k \le tn$ and $n \ge n'$.
Markov inequality implies (\ref{2-15-1}) in the same way as \cite[(2.121)]{H06}.
To prove (\ref{2-15-3}), using Cram\'{e}r Theorem, we show 
\begin{align}
\lim_{n\to \infty}\frac{1}{n}\log |{\cal R}_{ t' n,n} |
=\min_{s\ge 0} -sH_{1+s}(\Psi)+(1+s) (H_{1}(\Psi) - ct').
\end{align}
As shown in Lemma \ref{2-16-1}, we have
\begin{align}
\lim_{n\to \infty}\frac{1}{n}\log |{\cal R}_{0,n} |
=H_{1}(\Psi) .
\end{align}
Hence, we have
\begin{align}
\lim_{n\to \infty}\frac{1}{n}\log \frac{|{\cal R}_{t' n,n} |}{N_{t' n,n}}
=f(t') >0
\end{align}
for any real number $t'$ satisfying that $t'<t$.
Hence, when $n$ is sufficiently large, we have (\ref{2-15-3}).

\PF{Proof of (\ref{2-15-3b})}
Next, we proceed to the proof of \eqref{2-15-3b}.
In this proof, we will derive upper and lower bounds of 
$\sum_{\vec{x} \in {\cal X}^n}
P^n(\vec{x}) m_j(\vec{x}) $
and
$\sum_{\vec{x} \in {\cal X}^n}
P^n(\vec{x}) m_j(\vec{x})^2 $.
Using these bounds, we evaluate 
$ \log \sum_{j=1}^{M_n}
| \{\vec{x}|  m_j(\vec{x})\neq 0 \}|
\frac{\sum_{\vec{x} \in {\cal X}^n}P^n(\vec{x}) m_j(\vec{x})}
{\sum_{\vec{x} \in {\cal X}^n}P^n(\vec{x}) m_j(\vec{x})^2 }$.

From the above discussion, 
for any vector $\vec{x} \in {\cal R}_{k,n} $ and any integer $k$ satisfying $k \le tn$, 
the relation
$\lfloor \frac{|{\cal R}_{k,n}|}{N_{k,n} }\rfloor
/M_n \le m_j(\vec{x}) \le 
\Big\lceil \frac{|{\cal R}_{k,n}|}{N_{k,n} }\Big\rceil /M_n$
holds. 
Then, for $j=1, \ldots, M_n$
\begin{align}
& 
\frac{1}{2 e^{c}}
\sum_{k=0}^{tn} 
P^n({\cal R}_{k,n})/M_n
\le
\sum_{k=0}^{tn} 
N_{k,n} e^{-(n H_1(\Psi) + \sqrt{n} \sqrt{V(\Psi)} \Phi^{-1}(\epsilon)+ c(k+1))}
\bigg\lfloor \frac{|{\cal R}_{k,n}|}{N_{k,n} }\bigg\rfloor
/M_n \nonumber \\
\le &
\sum_{\vec{x} \in {\cal X}^n}
P^n(\vec{x}) m_j(\vec{x}) \nonumber\\
\le &
\sum_{k=0}^{tn} 
N_{k,n} e^{-(n H_1(\Psi ) + \sqrt{n} \sqrt{V(\Psi )} \Phi^{-1}(\epsilon )+ ck)}
\bigg\lceil \frac{|{\cal R}_{k,n}|}{N_{k,n} }\bigg\rceil
/M_n
\le
2 e^{c}
\sum_{k=0}^{tn} 
P^n({\cal R}_{k,n})/M_n
\end{align}
because
$\frac{1}{2} 
\frac{|{\cal R}_{k,n}|}{N_{k,n} }
\le
\big\lfloor \frac{|{\cal R}_{k,n}|}{N_{k,n} }\big\rfloor$
and
$ \big\lceil \frac{|{\cal R}_{k,n}|}{N_{k,n} }\big\rceil
\le 2 \frac{|{\cal R}_{k,n}|}{N_{k,n}}$.
Thus,
\begin{align}
& 
\frac{1}{4 e^{c}}
\sum_{k=0}^{tn} 
P^n({\cal R}_{k,n})
\frac{|{\cal R}_{k,n}|}{N_{k,n} M_n^2} \nonumber \\
\le &
\sum_{k=0}^{tn} 
N_{k,n} e^{-(n H_1(\Psi) + \sqrt{n} \sqrt{V(\Psi)} \Phi^{-1}(\epsilon)+ c(k+1))}
\bigg\lfloor \frac{|{\cal R}_{k,n}|}{N_{k,n} }\bigg\rfloor^2
/M_n^2 \nonumber \\
\le &
\sum_{\vec{x} \in {\cal X}^n}
P^n(\vec{x}) m_j(\vec{x})^2 \nonumber\\
\le &
\sum_{k=0}^{tn} 
N_{k,n} e^{-(n H_1(\Psi ) + \sqrt{n} \sqrt{V(\Psi )} \Phi^{-1}(\epsilon )+ ck)}
\bigg\lceil \frac{|{\cal R}_{k,n}|}{N_{k,n} }\bigg\rceil^2
/M_n^2 \nonumber \\
\le &
4 e^{c}
\sum_{k=0}^{tn} 
P^n({\cal R}_{k,n})
\frac{|{\cal R}_{k,n}|}{N_{k,n} M_n^2}.
\end{align}
Hence,
\begin{align}
& \Big(\sum_{k=0}^{tn} N_{k,n}\Big)
\frac{\sum_{\vec{x} \in {\cal X}^n}P^n(\vec{x}) m_j(\vec{x})}
{\sum_{\vec{x} \in {\cal X}^n}P^n(\vec{x}) m_j(\vec{x})^2 } 
\le
\Big(\sum_{k=0}^{tn} N_{n} e^{-ka}\Big)
\frac{4 e^{c}
\sum_{k=0}^{tn} 
P^n({\cal R}_{k,n})
\frac{|{\cal R}_{k,n}|}{N_{k,n} M_n^2}}
{\frac{1}{2 e^{c}}
\sum_{k=0}^{tn} 
P^n({\cal R}_{k,n})/M_n} \nonumber \\
\le &
\frac{8 e^{2c}}{1-e^{-a}}
\cdot N_{n} \cdot
\frac{\sum_{k=0}^{tn} 
P^n({\cal R}_{k,n})
\frac{|{\cal R}_{k,n}|}{N_{k,n} M_n}}
{ P^n(\cup_{k=0}^{tn}{\cal R}_{k,n})} \nonumber \\
= &
\frac{8 e^{2c}}{1-e^{-a}}
\cdot N_{n} \cdot
\frac{\sum_{k=0}^{tn} 
P^n({\cal R}_{k,n})
\frac{|{\cal R}_{k,n}|}{|{\cal R}_{0,n}| } e^{ka}}
{ P^n(\cup_{k=0}^{tn}{\cal R}_{k,n})} \nonumber \\
\le &
\frac{8 e^{2c}}{1-e^{-a}}
\cdot N_{n} \cdot
\frac{\max_{k=0}^{tn} 
P^n({\cal R}_{k,n})
\sum_{k=0}^{tn}\frac{|{\cal R}_{k,n}|}{|{\cal R}_{0,n}| } e^{ka}}
{ P^n(\cup_{k=0}^{tn}{\cal R}_{k,n})} .
\end{align}
Therefore,
\begin{align}
& \sum_{j=1}^{M_n}\bigg(\sum_{k=0}^{tn} N_{k,n}\bigg)
\frac{\sum_{\vec{x} \in {\cal X}^n}P^n(\vec{x}) m_j(\vec{x})}
{\sum_{\vec{x} \in {\cal X}^n}P^n(\vec{x}) m_j(\vec{x})^2 } \nonumber \\
\le &
M_n \frac{8 e^{2c}}{1-e^{-a}}
\cdot N_{n} \cdot
\frac{\max_{k=0}^{tn} 
P^n({\cal R}_{k,n})
\sum_{k=0}^{tn}\frac{|{\cal R}_{k,n}|}{|{\cal R}_{0,n}| } e^{ka}}
{ P^n(\cup_{k=0}^{tn}{\cal R}_{k,n})} \nonumber \\
= &
\frac{8 e^{2c}}{1-e^{-a}}
\cdot |{\cal R}_{0,n}| \cdot
\frac{\max_{k=0}^{tn} 
P^n({\cal R}_{k,n})
\sum_{k=0}^{tn}\frac{|{\cal R}_{k,n}|}{|{\cal R}_{0,n}| } e^{ka}}
{ P^n(\cup_{k=0}^{tn}{\cal R}_{k,n})} . \Label{2-16-8}
\end{align}
Thus,
since (\ref{2-12-2}) and (\ref{2-12-1}) of Lemma \ref{2-16-1} guarantees that
\begin{align}
\log 
\sum_{k=0}^{tn}\frac{|{\cal R}_{k,n}|}{|{\cal R}_{0,n}| } e^{ka} =O(1),
\end{align}
(\ref{2-12-2}) and (\ref{2-16-7}) of Lemma \ref{2-16-1} and (\ref{2-16-8})
imply
\begin{align}
& \log \sum_{j=1}^{M_n}\bigg(\sum_{k=0}^{tn} N_{k,n}\bigg)
\frac{\sum_{\vec{x} \in {\cal X}^n}P^n(\vec{x}) m_j(\vec{x})}
{\sum_{\vec{x} \in {\cal X}^n}P^n(\vec{x}) m_j(\vec{x})^2 } \nonumber \\
= &
n H_1(\Psi) + \sqrt{n} \sqrt{V(\Psi)} \Phi^{-1}(\epsilon)
-\frac{1}{2}\log n
-\frac{1}{2}\log n +O(1).
\end{align}
Hence, we obtain (\ref{2-15-3b}).
\end{proofof}

\begin{proofof}{Lemma \ref{2-16-1}}\par
\PF{Non-lattice case}
In this proof,
we combine the saddle point approximation method given in \cite[Theorem 2.3.6]{Dembo98},\cite{Moulin13} and Cram\'{e}r-Ess\'{e}en theorem \cite[p. 538]{Feller}.
Define 
\begin{align}
v(\vec{x})&\stackrel{\rm def}{=} ( \log P^n(\vec{x})+nH_1(\Psi)+
\sqrt{n}\sqrt{V(\Psi)}\Phi^{-1}(\epsilon))/\sqrt{n} \nonumber\\
Q_n(v)&\stackrel{\rm def}{=} \sum_{\vec{x}:v(\vec{x})=v} P^n(\vec{x}) .\nonumber
\end{align}
Then, we have
\begin{align}
|\{\vec{x}| a \le v(\vec{x}) \le b \}| 
=
e^{n H_1(\Psi)+\sqrt{n}\sqrt{V(\Psi)}\Phi^{-1}(\epsilon)}
\sum_{v:a \le v \le b}
e^{-\sqrt{n} v}
Q_n(v). \Label{2-12-5}
\end{align}
Hence, 
\begin{align}
 \sum_{k=0}^{\infty} |{\cal R}_{k,n}| e^{ka} 
=&
e^{n H_1(\Psi)+\sqrt{n}\sqrt{V(\Psi)}\Phi^{-1}(\epsilon)}
\sum_{k=0}^{\infty}
\sum_{v:\frac{ck}{\sqrt{n}} \le v \le \frac{c(k+1)}{\sqrt{n}}}
e^{-\sqrt{n} v+ka}
Q_n(v) \nonumber \\
\le &
e^{n H_1(\Psi)+\sqrt{n}\sqrt{V(\Psi)}\Phi^{-1}(\epsilon)}
\sum_{k=0}^{\infty}
\sum_{v:\frac{ck}{\sqrt{n}} \le v \le \frac{c(k+1)}{\sqrt{n}}}
e^{-\sqrt{n} v + \sqrt{n} av/c}
Q_n(v) \nonumber \\
\le &
e^{n H_1(\Psi)+\sqrt{n}\sqrt{V(\Psi)}\Phi^{-1}(\epsilon)}
\sum_{v: v \ge 0}
e^{-\sqrt{n}(1-\frac{a}{c}) v }
Q_n(v) .\Label{2-12-6}
\end{align}
Similarly, we can show that
\begin{align}
\sum_{k=0}^{\infty} |{\cal R}_{k,n}| e^{ka} 
\ge 
e^{n H_1(\Psi)+\sqrt{n}\sqrt{V(\Psi)}\Phi^{-1}(\epsilon)}
\sum_{v: v \ge 0}
e^{-\sqrt{n}(1-\frac{a}{c}) v -a}
Q_n(v) .\Label{2-12-6b}
\end{align}

Next, we define the distribution function
\begin{align}
F_{n,c}(t)\stackrel{\rm def}{=}P^n\{\vec{x}|v(\vec{x}) \le t\}.
\end{align}
In the following, we consider the non-lattice case.
Now, we employ the saddle point approximation method given in \cite[Theorem 2.3.6]{Dembo98},\cite{Moulin13}.
As is known as Cram\'{e}r-Ess\'{e}en theorem \cite[p. 538]{Feller},
there exist a constant $S$ and a function $c_n$ such that
\begin{align}
F_{n,c}(t-\Phi^{-1}(\epsilon))= \Phi(t) -\frac{S}{6 \sqrt{n}}(1-t^2)
\frac{e^{-t^2/2}}{\sqrt{2\pi}}
+ \frac{c_n(t)}{\sqrt{n}}.\Label{2-16-1b}
\end{align}
and $|c_n(t)|\to 0$, which is uniformly convergent on compact sets.
Thus, we obtain (\ref{2-15-2}).

Hence,
\begin{align*}
&
\lim_{n \to \infty}\sqrt{n}
\Biggl|\sum_{v: c(k+1)/\sqrt{n} > v \ge ck /\sqrt{n}}
Q_n(v) 
-\int_{ck /\sqrt{n}}^{c(k+1)/\sqrt{n}} 
\frac{e^{-\frac{(v+\Phi^{-1}(\epsilon))^2}{2V(\Psi)}}
}{\sqrt{2\pi V(\Psi)}}
dv \Biggr| \\
\le &
\lim_{n \to \infty}\sqrt{n} \Bigl(
\int_{ck /\sqrt{n}}^{c(k+1)/\sqrt{n}} 
\frac{d}{dt}\Big( 
\frac{S}{6 \sqrt{n}}(1-(t+\Phi^{-1}(\epsilon ))^2) 
\frac{e^{-(t+\Phi^{-1}(\epsilon))^2/2}}{\sqrt{2\pi}}\Big)  dt \\
&+
\frac{c_n(c/\sqrt{n})-c_n(0)}{\sqrt{n}} \Bigr)
\\
= &
\lim_{n \to \infty}
\Bigl[\frac{S}{6}(1-(t+\Phi^{-1}(\epsilon ))^2) 
\frac{e^{-(t+\Phi^{-1}(\epsilon))^2/2}}{\sqrt{2\pi}} 
\Bigr]_{\frac{ck}{\sqrt{n}}}^{\frac{c(k+1)}{\sqrt{n}}} 
+
c_n(c/\sqrt{n})-c_n(0) \\
=& 0,
\end{align*}
and
\begin{align*}
&
\lim_{n \to \infty}\sqrt{n}
\Biggl| \sum_{v: v \ge 0}
e^{-\sqrt{n}(1-\frac{a}{c}) v }
Q_n(v) 
-\int_{0}^{\infty} 
e^{-\sqrt{n}(1-\frac{a}{c}) v }
\frac{e^{-\frac{(v+\Phi^{-1}(\epsilon))^2}{2V(\Psi)}}
}{\sqrt{2\pi V(\Psi)}}
dv \Biggr|\\
\le &
\lim_{n \to \infty}\sqrt{n}
\int_0^{\infty}
e^{-\sqrt{n}(1-\frac{a}{c}) t }
\frac{d}{dt}(\frac{S}{6 \sqrt{n}}(1-(t+\Phi^{-1}(\epsilon ))^2)
\frac{e^{-(t+\Phi^{-1}(\epsilon))^2/2}}{\sqrt{2\pi}}) dt \\
&+
\inf_{a}
(2 \sup_{v \le a} e^{-\sqrt{n}(1-\frac{a}{c}) v }
\sup_{t \le a} \frac{|c_n(t)|}{\sqrt{n}}
+
2 \sup_{v > a} e^{-\sqrt{n}(1-\frac{a}{c}) v }
\sup_{t > a} \frac{|c_n(t)|}{\sqrt{n}})
\\
= &
\lim_{n \to \infty}
\int_0^{\infty}
e^{-\sqrt{n}(1-\frac{a}{c}) t }
\frac{d}{dt}(\frac{S}{6}(1-(t+\Phi^{-1}(\epsilon ))^2)
\frac{e^{-(t+\Phi^{-1}(\epsilon ))^2/2}}{\sqrt{2\pi}}) dt\\
&+
\inf_{a}
(2 \sup_{v \le a} e^{-\sqrt{n}(1-\frac{a}{c}) v }
\sup_{t \le a} |c_n(t)|
+
2 \sup_{v > a} e^{-\sqrt{n}(1-\frac{a}{c}) v }
\sup_{t > a} |c_n(t)|)
\\
=& 0.
\end{align*}
Thus, when $tn$ satisfies  $v_0=c tn/\sqrt{n}$,
\begin{align}
&
\lim_{n \to \infty}\sqrt{n}
P^n({\cal R}_{tn,n})
=
\lim_{n \to \infty}\sqrt{n}
\sum_{v:c(tn+1)/\sqrt{n}> v \ge c tn/\sqrt{n}}
Q_n(v) \nonumber \\
=&\lim_{n \to \infty}\sqrt{n}
\int_{c tn/\sqrt{n}}^{c(tn+1)/\sqrt{n}} 
\frac{e^{-\frac{(v+\Phi^{-1}(\epsilon))^2}{2V(\Psi)}}
}{\sqrt{2\pi V(\Psi)}}dv 
=
c
\frac{e^{-\frac{(v_0+\Phi^{-1}(\epsilon))^2}{2V(\Psi)}}
}{\sqrt{2\pi V(\Psi)}},\Label{2-12-3}
\end{align}
which implies (\ref{2-16-7}).
Further,
\begin{align}
&
\lim_{n \to \infty}\sqrt{n}
\sum_{v: v \ge 0}
e^{-\sqrt{n}(1-\frac{a}{c}) v }
Q_n(v) 
=\lim_{n \to \infty}\sqrt{n}
\int_{0}^{\infty} 
e^{-\sqrt{n}(1-\frac{a}{c}) v }
\frac{e^{-\frac{(v+\Phi^{-1}(\epsilon))^2}{2V(\Psi)}}
}{\sqrt{2\pi V(\Psi)}}dv \nonumber \\
=&
\lim_{n \to \infty}\sqrt{n}
\int_{0}^{\infty} 
e^{-(1-\frac{a}{c}) x }
\frac{e^{-\frac{(x/\sqrt{n}+\Phi^{-1}(\epsilon))^2}{2V(\Psi)}}
}{\sqrt{2\pi V(\Psi) n}}dx \nonumber \\
=&
\int_{0}^{\infty} 
e^{-(1-\frac{a}{c}) x }
\frac{e^{-\frac{\Phi^{-1}(\epsilon)^2}{2V(\Psi)}}
}{\sqrt{2\pi V(\Psi)}}dx
=
\frac{1}{1-\frac{a}{c}}
\frac{e^{-\frac{\Phi^{-1}(\epsilon)^2}{2V(\Psi)}}
}{\sqrt{2\pi V(\Psi)}}.\Label{2-12-4}
\end{align}
%e^{n H_1(\Psi)+\sqrt{n}\sqrt{V(\Psi)}\Phi^{-1}(\epsilon)}
Therefore, the combination of (\ref{2-12-5}) and (\ref{2-12-3}) yields (\ref{2-12-2}), 
and
the combination of (\ref{2-12-6}), (\ref{2-12-6b}) and (\ref{2-12-4}) yields (\ref{2-12-1}).

\PF{Lattice case}
Now, we consider the lattice case.
The range of the map $v$ is contained in $\{ a_n + \frac{ck}{\sqrt{n}}  \}_{k} $
by choosing a suitable real number $a_n$ with $|a_n| \le \frac{c}{2\sqrt{n}}  $.
Then, we define the set 
${\cal T}_n\stackrel{\rm def}{=}\{ a_n + \frac{ck}{\sqrt{n}}+\frac{c}{2\sqrt{n}}  \}_{k} $.
Then, (\ref{2-16-1b}) holds for $t\in {\cal T}_n$ \cite[pp. 52-67]{Esseen}\cite[p. 540]{Feller}.
Hence, similar to (\ref{2-12-3}) and (\ref{2-12-4}), 
we can show
\begin{align}
&
\lim_{n \to \infty}\sqrt{n}
\sum_{v:a_n+\frac{c(k+1)}{\sqrt{n}}\ge v \ge a_n+ \frac{c k}{\sqrt{n}}}
Q_n(v) \nonumber \\
=&\lim_{n \to \infty}\sqrt{n}
\int_{a_n+\frac{c k}{\sqrt{n}}}^{a_n+ \frac{c (k+1)}{\sqrt{n}}} 
\frac{e^{-\frac{(v+\Phi^{-1}(\epsilon))^2}{2V(\Psi)}}
}{\sqrt{2\pi V(\Psi)}}dv 
=
c
\frac{e^{-\frac{(v_0+\Phi^{-1}(\epsilon))^2}{2V(\Psi)}}
}{\sqrt{2\pi V(\Psi)}},\Label{2-12-3c}
\end{align}
with $v_0= k/\sqrt{n}$,
and
\begin{align}
&
\lim_{n \to \infty}\sqrt{n}
\sum_{v: v \ge a_n}
e^{-\sqrt{n}(1-\frac{a}{c}) v }
Q_n(v) 
=\lim_{n \to \infty}\sqrt{n}
\int_{a_n}^{\infty} 
e^{-\sqrt{n}(1-\frac{a}{c}) v }
\frac{e^{-\frac{(v+\Phi^{-1}(\epsilon))^2}{2V(\Psi)}}
}{\sqrt{2\pi V(\Psi)}}dv \nonumber \\
=&
\frac{1}{1-\frac{a}{c}}
\frac{e^{-\frac{\Phi^{-1}(\epsilon)^2}{2V(\Psi)}}
}{\sqrt{2\pi V(\Psi)}}.\Label{2-12-4c}
\end{align}
%e^{n H_1(\Psi)+\sqrt{n}\sqrt{V(\Psi)}\Phi^{-1}(\epsilon)}
Hence, (\ref{2-12-3c}) implies (\ref{2-16-7}).
Further, the combination of (\ref{2-12-5}) and (\ref{2-12-3c}) yields (\ref{2-12-2}), 
and
the combination of (\ref{2-12-6}), (\ref{2-12-6b}) and (\ref{2-12-4c}) does (\ref{2-12-1}).
\end{proofof}

\section{Conclusion and discussion}\Label{sec summary}
In this paper, we have treated local asymptotic hypothesis testing between
an arbitrary known bipartite pure state $\ket{\Psi}$ and the white noise state (the completely mixed state) $\rho_{mix}$.
As a result, 
we have clarified the difference between
the optimal performance of one-way and two-way LOCC POVMs.
Under the exponential constraint for the type-1 error probability,
there clearly exists a difference between the optimal 
exponential decreasing rates of the type-2 error probabilities 
under one-way and two-way LOCC POVMs.
However, when we surpass the constraint for the type-1 error probability,
this kind of difference is very subtle.
That is, there exists a difference only in the third order for
the optimal exponential decreasing rates of the type-2 error probabilities 
under one-way and two-way LOCC POVMs.
This difference has been given as Theorem \ref{thm1}, which is called 
the Stein-Strassen bound.
The entanglement of Renyi entropy appears in the 
	   formulas of the optimal exponential decreasing rates of the 
	   type-2 error probabilities under both exponential and constant 
	   constraints for the type-1 error probability for the one-way LOCC, the 
	   two-way LOCC, and separable constraints. Hence, our results 
	   have clarified the relationship between the entanglement of 
	   Renyi entropy and the local hypothesis testing.

From the beginning of the study of LOCC, many studies have focused on 
the effect of increasing the number of communication rounds, 
as well as on the difference between two-way LOCC and separable operations. 
From this viewpoint, our study gives a very rare example in which
the optimal performance under the infinite-round two-way LOCC, which is different 
from the one under the one-way LOCC, 
can be attained with two-round communication and is also equal to the one  
under separable operations. 
To show the achievability by two-round communication,
we employ 
the saddle point approximation method given in \cite[Theorem 2.3.6]{Dembo98},\cite{Moulin13}.
To show the impossibility to surpass this performance even in the separable operation,
we use the strong large deviation by Bahadur-Rao \cite{BR}\cite[Theorem 3.7.4]{Dembo98}.
We believe that these methods will become very strong approaches for 
addressing several topics in quantum information.

Unfortunately, our result can be applied to the case when the state to be distinguished from the completely mixed state is a pure state.
This is a serious defect of our result.
However, since our result completely solved the asymptotic analysis of this kind of state discrimination in the pure state case,
we have very strong motivation to tackle the mixed state case.
Hence, the extension of this result to the general mixed state case is
remained as an interesting future study, which attracts future researchers.

%Finally, we point out the relation with quantum data hiding.
%Using the meta-converse for classical-quantum channel coding\cite{}, we can show the impossibility part  
As mentioned in Section 1, this type of hypothesis testing is closely related to 
many kinds of 
information theoretical tasks, such as 
data compression \cite{HK02,Nag07}, uniform random generation \cite{HK02}, 
channel coding with additive noise \cite{VH}, and resolvability of the distribution \cite{RH}.
Hence, our results are expected to be applied to extending these problems to the case with the locality condition. 
However, this kind of extension has the following problems.
Since the obtained results are limited to the pure state case,
we need to extend our result to the mixed state case for this kind of applications.
However, this defect can be escaped when 
we make several restrictions for the quantum states or the quantum channels, e.g.,
the output states of the c-q channel are assumed to be pure entangled states.
As another problem, we need careful considerations for the formulations of these extensions because there are several kinds of formulations.

For example, we can consider an extension of the c-q channel coding as follows.
%An additive channel noise can be generalized to a regular channel \cite{delsarte82}, in which the channel is given as a distribution on a probability space and the group acts on the probability space.
%In this model, 
%An encoder is given as the choice of an element of the group acting on the probability space.
%Then decoder has locality restriction as follows.
We assume that a pure entangled state is given and that we are allowed to apply local unitary as an encoder.
The decoder is restricted to a measurement satisfying the locality condition.
In this case, since the encoded states are pure entangled states,
the above condition for the c-q channel is satisfied.
So, we expect that the asymptotic performance of this extension
can be characterized by our local hypothesis testing.
Since this setting is equal to the dense coding \cite{B1}, our analysis might bring
a deeper analysis for the dense coding.

In addition, we can consider an extension of uniform random generation as follows.
We assume that an entangled state is given and that we can apply local unitary randomly based on a uniform random number
so that the average state cannot be distinguished from the white noise state by
any measurement satisfying the locality condition.
In this case, the cardinality of the random number is as small as possible.
That is, we treat the trade-off between the above difficulty of local state discrimination and the cardinality of the used random number.
In this scenario, the difference between the product of local dimensions and the cardinality of the random number can be regarded as our analogue of the size of the generated uniform random number.
Then, we expect that the asymptotic performance of this extension can be characterized by our local hypothesis testing.
Analyses of these LOCC extensions remain as future work.

\section*{Acknowledgement}
MH is grateful to Dr. Vincent Tan for explaining the strong large deviation for the lattice case.
This research was partially supported by  
% the MEXT Grant-in-Aid for Young Scientists (A) No. 20686026, and 
the MEXT Grant-in-Aid for Scientific Research (A) No. 23246071
and the National Institute of Information and Communication Technology (NICT), Japan.
The Centre for Quantum Technologies is funded by the Singapore
Ministry of Education and the National Research Foundation
as part of the Research Centres of Excellence Programme.

\appendices

\section{Results of \cite{OH10} used in Subsubsection \ref{subsub1}}\Label{A2}
Here, we summarize the results of \cite{OH10} used in Subsubsection \ref{subsub1}.
As a preparation, we explain a useful knowledge in a Euclidean space $\mathbb{R}^d$.
For two vectors $y$ and $z$ in a Euclidean space $\mathbb{R}^d$, 
and a real number $\epsilon$ satisfying $0< \epsilon \le 1$, 
we define the real number $M(y,z,\epsilon)$ as
\begin{equation}
 M(y,z,\epsilon)\stackrel{\rm def}{=} \max_{x \in \mathbb{R}^d} \{ 
  y\cdot x \ | \ \|x\| \le 1, \ x \cdot z \le \epsilon \}. 
  \Label{H5}
\end{equation}
 Then, we derive the following Lemma:
\begin{proposition}[\protect{\cite[Lemma 9]{OH10}}]\Label{Lemma 9}
Using $c\stackrel{\rm def}{=} y \cdot z$,
we calculate $M(y,z,\epsilon)$ as
\begin{align}
\Label{eq sec global c M y z epsilon = c epsilon +}
& M(y,z,\epsilon) \nonumber \\
=&
\left\{
\begin{array}{ll}
\| y \| & \hbox{Case D1)} \\
\frac{\|y\|}{\|z\|}\epsilon & \hbox{Case D2)} \\
\frac{c \epsilon + \sqrt{(\|z\|^2-\epsilon^2)(\|y\|^2\|z\|^2-c^2)}}
{\|z\|^{2}} 
& \hbox{Case D3)},
\end{array}
\right.
\end{align}
which is attained by 
\begin{align}
&x^*(y,z,\epsilon) \nonumber \\
\stackrel{\rm def}{=} &
\left\{
\begin{array}{ll}
y/\|y\| & \hbox{Case D1)}  \\
\epsilon\frac{y}{\|y\| \|z\|} & \hbox{Case D2)} \\
\frac{1}{\sqrt{\|z\|^2\|y\|^2-c^2}} 
\Big( \sqrt{\|z\|^2-\epsilon^2} y  & \\
\qquad + 
\frac{\epsilon \sqrt{\|z\|^2\|y\|^2-c^2}
- c \sqrt{\|z\|^2-\epsilon^2} }{\|z\|^{2}}z \Big)
& \hbox{Case D3)},
\end{array}
\right.
\Label{eq def x*}
\end{align}
where Cases D1), D2), and D3) are defined as
\begin{enumerate}
\item[D1)] $y\cdot z \le \epsilon \|y\|$. 
\item[D2)] $y/\|y\|=z/\|z\|$ and $y\cdot z > \epsilon \|y\| $.
\item[D3)] $y/\|y\|\neq z/\|z\|$ and $y\cdot z > \epsilon \|y\| $.
\end{enumerate}
Moreover, 
$x^*(y,z,\epsilon)$ defined by Eq. (\ref{eq def x*}) is the unique 
solution  of the optimization problem in Case D3). 
Note that the relation $\|z\|^2-\epsilon^2 \ge 0$
follows from the common condition of Cases D2) and D3).
\hfill $\square$\end{proposition}

Now, we concentrate the hypothesis testing with composite hypothesis 
formulated in Subsubsection \ref{subsub1}.
The first kind of error probability
$\alpha ( \epsilon^2 |\varphi ) $ has the following two expressions.

\begin{proposition}[\protect{\cite[Lemma 8]{OH10}}]\Label{Lemma 8}
We have the following relation
\begin{align}
1- \alpha ( \epsilon^2 |\varphi ) 
=&  \max \big \{ 
\braket{\varphi}{\phi}^2 \ \big | \ \ket{\phi}\in \Hi, \| \ket{\phi}\|^2\le 1, 
  \braket{\phi_d}{\phi}\le \epsilon, \nonumber \\
\quad & \qquad 1 \le \forall i \le d-1, \braket{i}{\phi}\ge \braket{i+1}{\phi} \ge 0,
  \big \} ,
\end{align}
where $\ket{\phi _{j}}$ is defined as 
\begin{equation} \Label{eq sec sep def phi-d}
\ket{\phi _{j}} \stackrel{\rm def}{=}  \frac{1}{\sqrt{j}}\sum _{i=1}^j \ket{i}.
\end{equation}
\hfill $\square$\end{proposition}

To give another expression for $\alpha ( \epsilon^2 |\varphi ) $,
we define the real vectors $u_l$ and $v_l$ on $\mathbb{R}^l$ as 
$u_l \stackrel{\rm def}{=}\left( \sqrt{p_1}, \cdots , \sqrt{p_l} \right)$ 
and $v_l \stackrel{\rm def}{=} \left( 1,  \cdots , 1 \right) /\sqrt{d}$
for an integer $l$ satisfying $1 \le l \le d$.
We also define the natural number $\eta=\eta_\epsilon(\varphi)$ as 
the maximum integer $1 \le l \le d$ satisfying one of the following three conditions:
\begin{enumerate}
 \item[A1)] $u_l \cdot v_l \le \epsilon\| u_l \|$.
 \item[A2)] $u_l/\|u_l\| = v_l/\|v_l\|$ and $u_l \cdot v_l > \epsilon\| u_l \|$. 
 \item[A3)] $u_l/\|u_l\| \neq v_l/\|v_l\|$, 
$u_l \cdot v_l > \epsilon\| u_l \|$, and all the elements 
of $x^*(u_l, v_l, \epsilon)$ defined by Eq. (\ref{eq def x*}) 
are non-negative.
\end{enumerate}

Since $u_1/\|u_1\| = v_1/\|v_1\|$, 
one of Conditions A1), A2), and A3) holds at least $l=1$, 
i.e., $\eta\ge 1$.
Hence, we can consider three cases.
\begin{enumerate}
\item[B1)] $u_\eta \cdot v_\eta \le \epsilon \|u_\eta\|$. 
\item[B2)] $u_\eta/\|u_\eta \|=v_\eta/\|v_\eta\|$ and $u_\eta \cdot v_\eta > \epsilon \|u_\eta\| $.
\item[B3)] $u_\eta/\|u_\eta\|\neq v_\eta/\|v_\eta\|$ and $u_\eta \cdot v_\eta > \epsilon \|u_\eta\| $.
\end{enumerate}

\begin{proposition}[\protect{\cite[Theorem 4]{OH10}}]\Label{Theorem 4}
By using $c_\eta \stackrel{\rm def}{=} u_\eta \cdot v_\eta$,
the value $\alpha ( \epsilon^2 |\varphi )$ defined in Eq. (\ref{eq sep x epsilon = max lastB}) is calculated as follows:
\begin{align}
&1- \alpha ( \epsilon^2 |\varphi )  \nonumber \\
=&
\left\{
\begin{array}{ll}
\sum _{i=1}^\eta p_i & \hbox{Case B1)} \\
\frac{\epsilon^2 \|u_\eta \|^2}{\|v_\eta \|^2} & \hbox{Case B2)} \\
\frac{\left( c_\eta \epsilon +  
 \sqrt{(\|v_\eta \|^2-\epsilon^2)(\|u_\eta \|^2\|v_\eta \|^2-c_\eta^2)}\right)^2}
{\|v_\eta \|^{4}}
& \hbox{Case B3)}.
\end{array}
\right.
\end{align}
The maximum value $1- \alpha ( \epsilon^2 |\varphi )$ is attained by 
\begin{align}
& \ket{\phi^* }
%\nonumber \\
\stackrel{\rm def}{=} 
\left\{
\begin{array}{ll}
\ket{\phi[u_\eta/\|u_\eta\|]} & \hbox{Case B1)} \\
\ket{\phi[\epsilon \frac{u_\eta}{\|u_\eta \| \|v_\eta\|}]} & \hbox{Case B2)} \\
\ket{\phi[x^*(u_\eta,v_\eta,\epsilon)]} & \hbox{Case B3)}.
\end{array}
\right.
\end{align}
Note that $x^*(u_\eta,v_\eta,\epsilon)$ is defined in Eq. (\ref{eq def x*})
and the notation $\ket{\phi[~]} $ as
\begin{align}
|\phi[a]\rangle \stackrel{\rm def}{=} \sum_{i=1}^d 
a_i |i \rangle.
\Label{H6}
\end{align}
\hfill $\square$\end{proposition}

\section{Useful observations related to Appendix \ref{A2}}\Label{A4}
For the discussions in Subsubsection \ref{subsub1}, 
we discuss Conditions A1), A2), and A3) given in Appendix \ref{A2}.
In this appendix, we employ the same notations as Appendix \ref{A2}.
For Conditions A1) and A2), we have the following lemmas.

\begin{lemma}\Label{L-2-18}
The inequality 
$(u_l)_l/\|u_l\| \le (v_l)_l/\|v_l\|$ holds, and
the equality holds only when $p_1= p_l$.
In other words,
when $p_1> p_l$, the relation 
$u_l/\|u_l\| \neq v_l/\|v_l\|$ holds.
\hfill $\square$\end{lemma}

\begin{proof}
The inequality $ l p_l \le \sum_{i=1}^l p_i$ holds,
and the equality holds only when $p_1= p_l$.
Since 
$((v_l)_l/\|v_l\|)^2
=\frac{1}{l}$
and 
$((u_l)_l/\|u_l\|)^2= \frac{p_l}{\sum_{i=1}^l p_i}$,
we obtain the desired statement.
\end{proof}

Therefore, we can ignore Condition A2) except for the case of $p_1> p_l$.

\begin{lemma}\Label{L-2-18B}
$\frac{u_l \cdot v_l}{\| u_l \|}$
is strictly monotone increasing for $l$.
\hfill $\square$\end{lemma}

Hence, when 
$\frac{u_{\hat{l}} \cdot v_{\hat{l}}}{\| u_{\hat{l}} \|}=\epsilon$,
the relation $\frac{u_l \cdot v_l}{\| u_l \|}>\epsilon$
holds for $l \ge {\hat{l}}$, i.e., Condition A1) does not hold for 
$l \ge {\hat{l}}$.

\begin{proof}
Since 
$(\frac{u_l \cdot v_l}{d \| u_l \|})^2
=\frac{(\sum_{i=1}^l \sqrt{p_i})^2}{\sum_{i=1}^l p_i}$,
it is enough to show that
$\frac{(\sum_{i=1}^{l+1} \sqrt{p_i})^2}{\sum_{i=1}^{l+1} p_i}
>\frac{(\sum_{i=1}^l \sqrt{p_i})^2}{\sum_{i=1}^l p_i}$,
which is equivalent to 
$
({\sum_{i=1}^l p_i}){(\sum_{i=1}^{l+1} \sqrt{p_i})^2}
>({\sum_{i=1}^{l+1} p_i}){(\sum_{i=1}^l \sqrt{p_i})^2}
$.
We have
\begin{align}
&({\sum_{i=1}^l p_i}){(\sum_{i=1}^{l+1} \sqrt{p_i})^2}
-({\sum_{i=1}^{l+1} p_i}){(\sum_{i=1}^l \sqrt{p_i})^2} \\
=&p_{l+1}
\Big(
({\sum_{i=1}^l p_i})
+ \frac{2}{\sqrt{p_{l+1}}}
({\sum_{i=1}^l p_i})
(\sum_{i=1}^l \sqrt{p_i})
-(\sum_{i=1}^l \sqrt{p_i})^2
\Big) \\
=& p_{l+1}
\Big(
({\sum_{i=1}^l p_i})
+ 
(\sum_{i=1}^l \sqrt{p_i})
\big(
\frac{2}{\sqrt{p_{l+1}}}
({\sum_{i=1}^l p_i})
-(\sum_{i=1}^l \sqrt{p_i})
\big) \Big) .
\end{align}
Since $2 \frac{\sqrt{p_i}}{\sqrt{p_{l+1}}}>1$, we have
\begin{align}
\frac{2}{\sqrt{p_{l+1}}}
({\sum_{i=1}^l p_i})
-(\sum_{i=1}^l \sqrt{p_i})
=
2(\sum_{i=1}^l \frac{p_i}{\sqrt{p_{l+1}}})
-(\sum_{i=1}^l \sqrt{p_i})
>0.
\end{align}
So, we obtain the desired statement.
\end{proof}

\begin{lemma}\Label{L-2-18C}
Assume that
$ \frac{u_l \cdot v_l}{\|u_l\|}>\epsilon $ and $p_l <p_1$.
All entries of $x^*(u_l,v_l,\epsilon)$ are non-negative
if and only if 
\begin{align}
\sqrt{p_l}
\frac{d^{1/2} \|v_l\|^2}{u_l \cdot v_l}
\ge
\left(1-
\frac{\sqrt{\frac{\|u_l\|^2\|v_l\|^2}{(u_l \cdot v_l)^2}-1}}
{\sqrt{\frac{\|v_l\|^2}{\epsilon^2}-1}}
\right) .\Label{L-2-18C-Eq}
\end{align}
\hfill $\square$\end{lemma}

\begin{proof}
The above non-negativity is equivalent to
the non-negativity of the $l$-th entry of $x^*(u_l,v_l,\epsilon)$,
which is equivalent to
\begin{align*}
0 \le &
\sqrt{\|v_l\|^2 -\epsilon^2} \sqrt{p_l}
+\frac{
\epsilon \sqrt{\|u_l\|^2\|v_l\|^2-(u_l \cdot v_l)^2}
- u_l \cdot v_l \sqrt{\|v_l\|^2-\epsilon^2}
}{\|v_l\|^2} \frac{1}{d^{1/2}} \\
=&
\sqrt{\|v_l\|^2 -\epsilon^2} 
\left(
\sqrt{p_l}
+
\frac{u_l \cdot v_l}{\|v_l\|^2}
\left(
\frac{\sqrt{\frac{\|u_l\|^2\|v_l\|^2}{(u_l \cdot v_l)^2} -1}}
{\sqrt{\frac{\|v_l\|^2}{\epsilon^2}-1}}
-1\right) \frac{1}{d^{1/2}} \right).
\end{align*}
This condition is equivalent to
$\sqrt{p_l} \ge
\frac{u_l \cdot v_l}{d^{1/2} \|v_l\|^2}
\left(1-
\frac{\sqrt{\frac{\|u_l\|^2\|v_l\|^2}{(u_l \cdot v_l)^2}-1}}
{\sqrt{\frac{\|v_l\|^2}{\epsilon^2}-1}}
\right)$.
That is,
\begin{align}
\sqrt{p_l}
\frac{d^{1/2} \|v_l\|^2}{u_l \cdot v_l}
\ge
\left(1-
\frac{\sqrt{\frac{\|u_l\|^2\|v_l\|^2}{(u_l \cdot v_l)^2}-1}}
{\sqrt{\frac{\|v_l\|^2}{\epsilon^2}-1}}
\right) .
\end{align}
\end{proof}

\section{Strong large deviation}\Label{A1}
Let $p$ be a non-negative measure and
$d_S$ be the lattice span of the real valued function $X$, 
which is defined as follows.
Let $S$ be the set of the support of the measure $p \circ X^{-1}$.
When there exists a non-negative value $x$ satisfying 
$ \{ a-b\}_{a,b\in S}  \subset x \bZ$,
the real valued function $X$ is called a lattice function or a lattice variable.
Then, 
the lattice span $d_S$ is defined as the maximum value of the above non-negative value $x$.
Denoting all of elements of $S$ as $a_1<a_2< \ldots < a_l$,
we have 
\begin{align}
d_S= \min_{n_i \in \bZ} \bigg\{  \sum_{i=1}^l n_i a_i \Bigg| \sum_{i=1}^l n_i =0 ,~ 
\sum_{i=1}^l n_i a_i>0\bigg\}
\end{align}
due to the following reason;
When integers $y_1, \ldots, y_l$ have the greatest common divisor $1$,
there exist integers $n_1, \ldots, n_l$ such that $\sum_{i=1}^l n_i y_i=1$.

When there does not exist such a non-negative value $x$,
the real valued function $X$ is called 
a non-lattice function or a non-lattice variable.
Then, the lattice span $d_S$ is regarded as zero.

Now, we summarize the fundamental properties for the lattice and non-lattice cases.
For this purpose,
we denote the set $ \{\sum_{i=1}^n a_i \}_{a_i \in S}$ by $S_n$.

\begin{lemma}\Label{L9-20}
%We denote all of elements of $S$ as $a_1<a_2< \ldots < a_l$.
%Let $a_{}$ and $a_{\max}$ be the minimum and maximum of $S$.
We fix a small real number $\delta>0$.
In the lattice case, 
there exists a sufficiently large integer $N$ 
such that $S_n$ satisfies the following condition for any $n \ge N$. 
Denote all of elements of $S_n\cap [n(a_1+\delta),n(a_l-\delta)]$
as $b_1<b_2< \ldots < b_k$.
We have $b_{i+1}-b_i= d_S$.

In the non-lattice case, 
for an arbitrary  small real number $\epsilon$,
there exists a sufficiently large integer $N$ 
such that $S_n$ satisfies the following condition for any $n \ge N$. 
Denote all of elements of $S_n\cap [n(a_1+\delta),n(a_l-\delta)]$
as $b_1<b_2< \ldots < b_k$.
We have $b_{i+1}-b_i< \epsilon$.
\end{lemma}

\begin{proof}
\PF{Lattice case}
Since the definition of $d_S$ guarantees that $b_{i+1}-b_i\ge d_S$,
it is enough to show that $b_{i+1}-b_i\le d_S$.
Assume that integers $n_i$ satisfies the equations 
\begin{align}
\sum_{i=1}^l n_i a_i &=d_S \Label{9-20-1} \\
\sum_{i=1}^l n_i &=0.\Label{9-20-2}
\end{align}
We define the subsets $S_+:= \{a_i \in S| n_i \ge 0 \}$ and
$S_-:= \{a_i \in S| n_i < 0 \}$,
the positive integers
$m_2:= \sum_{i:a_i\in S_+}n_i$ and $m_1:= (a_l-a_1)/d_S$,
and the positive real numbers
$A:= -m_1 \sum_{i :a_i\in S_-}n_i a_i$,
$B:= m_1 \sum_{i :a_i\in S_+}n_i a_i$,
$\delta_-:= (A-a_1 m_1 m_2)/n$, and
$\delta_+:= (a_l m_1 m_2-B+ m_1d_s)/n$.

So, we have $n (a_1+\delta_-)= a_1 (n-m_1 m_2)+A= n a_1+ (A-a_1 m_1 m_2)$
and $n (a_l-\delta_+)= a_l (n-m_1 m_2)+B= n a_l- (a_l m_1 m_2-B)$.
We choose an element $x:=n (a_1+\delta_-) + (c_1 m_1+c_2)d_S
\in [n (a_1+\delta_-),n (a_l-\delta_+)]$ with 
integers $c_1$ and $c_2 \le m_1$.
When $(c_1 m_1+c_2)$ takes the maximum, $x$ is $n (a_l-\delta_+)$, i.e.,
$c_1 m_1+c_2= (n-m_1 m_2) m_1$.
So, the maximum of $c_1$ is $n-m_1 m_2$.

Using \eqref{9-20-1} and the definitions of $\delta_-$ an $A$, 
we have 
\begin{align}x= c_1 a_l+(n-c_1-m_1 m_2) a_1
+      c_2(\sum_{i:a_i\in S_+}n_i a_i) 
- (m_1-c_2)\sum_{i:a_i\in S_-}n_i a_i \stackrel{(a)}{\in} S_n.
\end{align}
Here, the relation $(a)$ follows from the following facts;
$c_1$ and $ (n-c_1-m_1 m_2)$ are non-negative integers, 
$c_2 n_i$ is a non-negative integer for $i \in S_+$,
and $ - (m_1-c_2)n_i$ is a non-negative integer for $i \in S_-$. 
Thus, when we denote all of elements of $S_n\cap [n (a_1+\delta_-),n (a_l-\delta_+)]$ as $b_1<b_2< \ldots < b_k$.
We have $b_{i+1}-b_i \le d_S$.
When $n$ is sufficiently large, we have $\delta_-, \delta_+ \le \delta$.
So, we obtain the desired statement.

\PF{Non-lattice case}
For an arbitrary $\epsilon>0$, we can take integers $n_i$ such that
$0<\tilde{d}:=\sum_{i=1}^l n_i a_i< \epsilon$ and $\sum_{i=1}^l n_i=0$.
(If impossible, we have the minimum of $\sum_{i=1}^l n_i a_i$ with $\sum_{i=1}^l n_i=0$
is strictly larger than $0$, which contradicts $d_S=0$.)
We redefine $m_1:= \lceil(a_l-a_1)/\epsilon\rceil$, 
and define other terms in the same way by replacing $d_S$ by $\tilde{d}$.
Using the same discussion, we find that 
the element $x:=n (a_1+\delta_-) + c_1 (a_l-a_1)+ c_2 \tilde{d}
\in [n (a_1+\delta_-),n (a_l-\delta_+)]$ with $c_2 \le m_1$
belongs to $S_n$.
When $n$ is sufficiently large, we have $\delta_-, \delta_+ \le \delta$.
So, we have $b_{i+1}-b_i< \epsilon$.
\end{proof}

Here $p$ is not necessarily normalized.
Define the notation 
$E_p [X] \stackrel{\rm def}{=} \int X(\omega) p(d \omega)$.
Define the cumulant generating function 
$\tau (s)\stackrel{\rm def}{=} \log E_p [e^{sX}]$.
Denote the inverse function of the derivative $\tau' (s) $ by $\eta$.

\begin{proposition}[\protect{Bahadur and Rao \cite{BR}, \cite[Theorem 3.7.4]{Dembo98}}]\Label{11-4-4}
Assume that $\tau (0)<\infty $.
When $R > \frac{E_p[X]}{E_p[1]}$, 
we have
\begin{align}
\log p^n\{ X_n \ge n R\}
&= \chi_0(R) n - \frac{1}{2}\log n+\chi_1(R)+\chi_2(R)\frac{1}{n}+o(\frac{1}{n}) \\
\log p^n\{ X_n \le n R\}
&= n \tau (0) +o(1) ,
\end{align}
where
\begin{align}
\chi_0(R)&\stackrel{\rm def}{=}-R \eta(R) + \tau (\eta (R)) \\
\chi_1(R)&\stackrel{\rm def}{=}
\left\{
\begin{array}{ll}
-\frac{1}{2}\log 2 \pi -\log \eta(R) + \frac{1}{2}\eta'(R)
& \hbox{ if } d_S=0 \\
-\frac{1}{2}\log 2 \pi + \frac{1}{2}\eta'(R)
+ \log \frac{d_S}{1-e^{-d_S \eta(R)}} 
& \hbox{ if } d_S>0,
\end{array}
\right.
\end{align}
and $\chi_2(R)$ is a continuous function.
When $R < \frac{E_p[X]}{E_p[1]}$, 
we have
\begin{align}
\log p^n\{ X_n \ge n R\}
&= n \tau (0) +o(1) \\
\log p^n\{ X_n \le n R\}
&= \chi_0(R) n - \frac{1}{2}\log n+\chi_1(R)+\chi_2(R)\frac{1}{n}+o(\frac{1}{n}) .
\end{align}
The convergences of the differences between the LHSs and RHSs 
are compact uniform.
\hfill $\square$\end{proposition}

\section{Proof of Lemma \ref{L12}}\Label{A3}
Now, we show Lemma \ref{L12}.
For $\theta$,
we define the distribution $P_\theta$ as 
\begin{align}
P_\theta(x)\stackrel{\rm def}{=} 
\frac{P^{1-\theta} (x)}{\sum_{x \in {\cal X}}P^{1-\theta} (x)}.
\end{align}
Then, for $r < %\log d_A d_B
- H_1(\Psi)$,
we define $\theta(r) \in (0,1]$ as
\begin{align}
D(P_{\theta(r)} \|P)=r. 
\Label{2-3-2}
\end{align}

\begin{lemma}\Label{L2}
For $r < %\log d_A d_B
- H_1(\Psi)$,
we have
\begin{align}
D(P_{\theta(r)}\|P) -H(P_{\theta(r)})=
\sup_{0 \le s < 1}
\frac{-2s}{1-s}r - H_{\frac{1+s}{2}}(\Psi).
\Label{2-3-1}
\end{align}
\hfill $\square$\end{lemma}

\begin{proof}
Define the function $\varphi(\theta)\stackrel{\rm def}{=} \log \sum_{x \in {\cal X}}P^{1-\theta} (x)$.
Since $\varphi''(\theta)>0$,
the function $\varphi(\theta)$ is strictly convex.
We have
$D(P_{\theta} \|P)= \theta \varphi'(\theta)-\varphi(\theta)$
and
$H(P_{\theta} )= (1-\theta )\varphi'(\theta)+\varphi(\theta)$.
We also have $D(P_{\theta}\|P) -H(P_{\theta})=
(2\theta-1) \varphi'(\theta)-2 \varphi(\theta)$.
Since $D(P_{\theta(r)} \|P)=r$,
solving the relation 
$\theta(r) \varphi'(\theta(r))-\varphi(\theta(r))=r$,
we have
$D(P_{\theta(r)}\|P) -H(P_{\theta(r)})=f(\theta(r))$
by using 
the function $f(\theta)\stackrel{\rm def}{=}
\frac{(2\theta-1)r -\varphi(\theta)}{\theta}$.

The derivative of $f$ is
$f'(\theta)\stackrel{\rm def}{=}
\frac{\varphi(\theta)+r-\theta\varphi'(\theta)}{\theta^2}$.
The derivative of the numerator is $-\theta\varphi''(\theta) <0$
when $\frac{1}{2} \ge\theta >0$.
Hence, $\sup_{0 \le s \le \frac{1}{2}} f(\theta)$
is realized when $f'(\theta)=0$, which is equivalent to 
$\varphi(\theta)+r-\theta\varphi'(\theta)=0$, i.e., 
$D(P_{\theta} \|P)=r$.
This condition is equivalent to $\theta=\theta(r)$.
Therefore,
$\sup_{0 \le s \le \frac{1}{2}} f(\theta)
=f(\theta(r))$.
That is, we have
$D(P_{\theta(r)}\|P) -H(P_{\theta(r)})=f(\theta(r))
=\sup_{0 \le s <1} f(\theta)$.
Since $f(\theta)=\frac{-2s}{1-s}r - H_{\frac{1+s}{2}}(\Psi)$
with $1-\theta= \frac{1+s}{2}$,
we obtain (\ref{2-3-1}).
\end{proof}

\begin{lemma}\Label{L1}
For $r < %\log d_A d_B
- H_1(\Psi)$,
we have
\begin{align}
\min_{Q:D(Q\|P)\le D(P_{\theta(r)} \|P)}
D(Q\|P) -H(Q)
=
D(P_{\theta(r)}\|P) -H(P_{\theta(r)}).\Label{2-3-20}
\end{align}
\hfill $\square$\end{lemma}
Combining Lemma \ref{L2} and \ref{L1},
we obtain \eqref{2-3-20a} and \eqref{2-3-21b} of Lemma \ref{L12}.

\begin{proof}
Assume that 
for a distribution $Q$, 
there exists a parameter $\theta \in [0,1]$ such that
$H(Q)=H(P_\theta)$.
Then,
we have
$\frac{1}{1-\theta}D(Q\|P_\theta)
= 
\frac{1}{1-\theta} \sum_{x}Q(x)\log Q(x)
-\sum_{x}Q(x)\log P(x)
-\frac{\varphi(\theta)}{1-\theta}$.
Hence,
\begin{align*}
&D(Q\|P) -\frac{1}{1-\theta}D(Q\|P_\theta) \\
=&\sum_{x}Q(x) (\log Q(x) -\log P(x))
-\frac{1}{1-\theta} \sum_{x}Q(x)\log Q(x)
+\sum_{x}Q(x)\log P(x)
+\frac{\varphi(\theta)}{1-\theta} \\
=&
-\frac{\theta}{1-\theta} \sum_{x}Q(x)\log Q(x)
+\frac{\varphi(\theta)}{1-\theta} \\
=&-\frac{\theta}{1-\theta} H(Q)
+\frac{\varphi(\theta)}{1-\theta} 
=\frac{\theta}{1-\theta} H(P_\theta)
+\frac{\varphi(\theta)}{1-\theta} \\
=&-\frac{\theta}{1-\theta} \sum_{x}P_\theta(x) (1-\theta) \log P(x)
-\frac{\theta}{1-\theta}\varphi(\theta)
+\frac{\varphi(\theta)}{1-\theta} \\
=&-\theta\sum_{x}P_\theta(x) \log P(x)
+\varphi(\theta)
=D(P_\theta\|P).
\end{align*}
%Since $ \{H(Q)\}= \{H(P_\theta)\}_{\theta}$,
Since $\frac{1}{1-\theta}D(Q\|P_\theta) \ge 0$, 
for $\theta \in [0,1]$, 
we have
%First, we notice that \cite{}
\begin{align}
\max_{Q:D(Q\|P)\le D(P_{\theta} \|P)}
H(Q)
=
H(P_{\theta}).\Label{2-3-22}
\end{align}
Hence, 
\begin{align}
\min_{Q:D(Q\|P)\le D(P_{\theta(r)} \|P)}
D(Q\|P) -H(Q)
=
D(P_{\theta(r)}\|P) -H(P_{\theta(r)}).
\end{align}
\end{proof}

\begin{proofof}{(\ref{2-3-21a})}
Now, we proceed to the proof of (\ref{2-3-21a}).
(\ref{2-3-22}) implies that
\begin{align}
\min_{Q:H(Q) \ge H(P) }
D(Q\|P) -H(Q)
=
\min_{\theta}
D(P_{\theta}\|P) -H(P_{\theta}).
\end{align}
Since 
$\min_{Q}
D(Q\|P) -H(Q)
=
\min_{Q:H(Q) \ge H(P) }
D(Q\|P) -H(Q)$,
we have
\begin{align}
\min_{Q}
D(Q\|P) -H(Q)
=
\min_{\theta}
D(P_{\theta}\|P) -H(P_{\theta}).
\end{align}
In the proof of Lemma \ref{L2}, we show that
$D(P_{\theta}\|P) -H(P_{\theta})=
(2\theta-1) \varphi'(\theta)-2 \varphi(\theta)$
and 
$D(P_{\theta}\|P) -H(P_{\theta})$ realizes the  minimum 
at $\theta=1/2$.
Since $(1-1) \varphi'(1/2)-2 \varphi(1/2)=- H_{1/2}(\Psi)$, 
we obtain (\ref{2-3-21a}).
\end{proofof}

\end{document}